%% file: Identification_in_Some_Discrete_Choice_Models_A_computational_Approach.tex
\documentclass[11pt]{article}

\usepackage{bbm}

\usepackage{amsmath}

\usepackage{tikz}
\newcommand*\circled[1]{\tikz[baseline=(char.base)]{
    \node[shape=circle, draw, inner sep=0.5pt, 
        minimum height=11pt] (char) {\vphantom{1g}#1};}}

\usepackage[english]{babel}
\usepackage{beton}
\usepackage{euler}
\usepackage[T1]{fontenc}
\usepackage{color,soul}
\usepackage{graphicx}
\graphicspath{ {images/} }
\usepackage{hyperref}

\usepackage{array}
\newcolumntype{H}{>{\setbox0=\hbox\bgroup}c<{\egroup}@{}}
 
\usepackage{ccaption}

\usepackage{amsmath, amsfonts, amssymb, amsthm}
\usepackage{mathrsfs}
\newtheorem{theorem}{Theorem}[section]

\newtheorem{lemma}[theorem]{Lemma}
\newtheorem{proposition}[theorem]{Proposition}
\theoremstyle{definition}
\newtheorem{definition}[theorem]{Definition}

\newtheorem{assumption}[theorem]{Assumption}
\theoremstyle{remark}
\newtheorem{remark}[theorem]{Remark}
\newtheorem{example}[theorem]{Example}
\newtheorem{claim}[theorem]{Claim}
\usepackage{caption}
\usepackage{threeparttable}
\usepackage{subcaption}
\usepackage{adjustbox}

\usepackage{amsthm}

\usepackage{multirow}

\usepackage{enumerate}
\usepackage{textcomp}

\usepackage[margin=2.5cm]{geometry}
\usepackage[babel]{csquotes}
\usepackage[style=authoryear-comp,
                        bibstyle=authoryear,maxbibnames=9,
                        backend=bibtex
                        ]{biblatex}

\usepackage{longtable}
\usepackage{rotating}
\usepackage{booktabs}
\usepackage{indentfirst}
\usepackage{authblk}

\numberwithin{equation}{section}

\bibliography{biblio.bib}

\begin{document}

\title{Identification in Some Discrete Choice Models: A Computational Approach\thanks{I would like to thank Ivan Canay and Jim Powell for their advices and encouragements. I would also like to thank Jackson Bunting, Yu Chen, Bulat Gafarov, Leonard Goff, Jiaying Gu, Ismael Mourifi\'e, Rui Wang, and the econometrics seminar participants at UCLA, the University of Toronto, Ohio State University, Boston College, Pennsylvania State University, the University of Albany, Texas A\&M University, and UC Davis for their helpful comments.}}

\author{Eric Mbakop}
\affil{Department of Economics, University of Calgary}
\date{}

\maketitle

\begin{abstract}

This paper presents an algorithm that generates the conditional moment inequalities that characterize the identified set of the common parameter of various semi-parametric panel multinomial choice models. I consider both static and dynamic models, and consider various weak stochastic restrictions on the distribution of observed and unobserved components of the models. For a broad class of such stochastic restrictions, the paper demonstrates that the inequalities characterizing the identified set of the common parameter can be obtained as solutions of multiple objective linear programs (MOLPs), thereby transforming the task of finding these inequalities into a purely computational problem. The algorithm that I provide reproduces many well-known results, including the conditional moment inequalities derived in \cite{CM}, \cite{PP1}, and \cite{KPT}. Moreover, I use the algorithm to generate some new results, by providing characterizations of the identified set in some cases that were left open in \cite{PP1} and \cite{KPT}, as well as characterizations of the identified set under alternative stochastic restrictions.


\end{abstract}


\section{Introduction}
Multinomial choice models have been widely used in empirical research to model situations where agents face a finite set of alternatives, ever since the seminal paper by \cite{DMF}. Estimating the parameters of these models allows researchers (among other things) to evaluate the relative importance of the different factors that drive individuals' choices, estimate average partial effects, and make predictions about counterfactual scenarios, such as the introduction of a new alternative. For a comprehensive introduction and overview of multinomial choice models, including applications, the reader is referred to \cite{KET}.

In this paper, I consider panel multinomial choice models of the type
\begin{equation}
\label{eqn0000}
Y_{ti}=\arg\max_{d\in [D]} \bigl\{X_{dti}' \beta+\gamma \mathbbm{1}_{\{Y_{(t-1)i}=d\}}+\lambda_{di}+\epsilon_{dti}\bigl\},
\end{equation}
where $Y_{ti}$ denotes  the observed choice of individual i at time $t\in[T]$ \footnote{I restrict the analysis in this paper to balanced panels and, given an integer $n\in \mathbb{N}$, I will use $[n]$ to denote the set $[n]:=\{1,2,\cdots,n\}$.} among the D ($\geq 2$) alternatives that constitute the choice set $[D]$, $X_{dti}$ represents the vector of observed individual-alternative-time specific covariates, and $\lambda_{di}$ and $\epsilon_{dti}$ denote respectively the latent individual-alternative specific fixed effects and the random preference shocks. In this paper, I am interested in providing a simple characterization of the (sharp) identified set of the common parameter $\theta=(\beta,\gamma)$, under various weak (i.e. nonparametric) distributional assumptions on the joint distribution of the observed and latent variables. I will refer to models where the utilities of alternatives does not depend on lagged choices (i.e., $\gamma=0$) as "static settings", and to those where it does as "dynamic settings." It is important to note that although the approach presented in this paper applies more generally to dynamic settings where the indirect utilities may depend on any arbitrary (but finite) number of lagged choices, when discussing dynamic settings, I will mainly focus on settings with first order dynamics (one lag) for the sake of expositional simplicity.\par
In order to better describe what I do in this paper, consider the two-periods (i.e. $T=2$) binary choice setting of \cite{CM}, where the choice of individual i at time $t\in [2]$ is modeled by:
\begin{equation}
\label{eqnMans}
Y_{ti}=\mathbbm{1}_{\bigl\{X_{ti}'\theta -\lambda_i -\epsilon_{ti} > 0\bigl\}}.
\end{equation}
Under the assumption that the conditional marginal distributions of period 1 and period 2 shocks, given the covariates and fixed effects, are equal  and are continuously distributed with full support (i.e. $\epsilon_{1i}|X_{1i},X_{2i},\lambda_i\sim \epsilon_{2i}|X_{1i},X_{2i},\lambda_i$), \cite{CM} shows that the CCPs satisfy the following three inequality restrictions:
\begin{equation}
\label{eqn0001}
x_1'\theta \gtreqqless x_2'\beta \quad \text{ implies } \quad P(Y_1=1|X_1=x_1,X_2=x_2) \gtreqqless P(Y_2=1|X_1=x_1,X_2=x_2).
\end{equation} 
Conditional moment inequality restrictions like \eqref{eqn0001} are important for at least two reasons. First, when point identification is of interest, these inequalities can be used to investigate what additional assumptions, if any, are needed to guarantee point identification and can form the basis for constructing semi-parametric estimators for the parameter $\theta$. In \cite{CM}, for instance, it is shown that a sort of large support assumption on the intertemporal difference of the covariates suffices to guarantee point identification of $\theta$, and inequalities \eqref{eqn0001} are used to construct a maximum score estimator of the parameter. Another example of an estimator constructed in a similar manner is the cyclical monotonicity-based estimator proposed by \cite{SSS}.\footnote{see also \cite{KOT} and \cite{KPT}.} Second, as large support-like assumptions on the covariates seem necessary to ensure point identification of $\theta$ when only weak stochastic restrictions are placed on the random utility shocks (see Theorem 1 in \cite{GC2}), in applications where such large support assumptions are not plausible, inequalities like \eqref{eqn0001} can form the basis for a partial identification analysis. One may then ask whether the inequalities \eqref{eqn0001} characterize the sharp identified set,  in that they exhaust all of the information contained in the maintained stochastic restriction. For if they do not, the set of parameter values that are consistent with these inequalities will form an outer region of the sharp identified set. In general, showing sharpness is not an easy task, and it has only been carried out in a few cases. In particular, under similar stochastic restrictions as in \cite{CM}, and using different proof techniques, \cite{PP1} and \cite{KPT} consider special cases of Model \ref{eqn0000} (where $\gamma=0$ and $T=2$ for \cite{PP1}, and where $D=2$ for \cite{KPT}), derive inequality restrictions like \eqref{eqn0001}, and prove that the restrictions  that they obtain characterize the sharp identified set.
One limitation of the approach of these papers is that the arguments that they employ to establish sharpness are model-specific, and do not readily translate to new inequalities when the models under consideration are slightly modified. This paper provides a unified approach, and I present an algorithm that generates all the conditional moment inequality restrictions that characterize the sharp identified set of the parameter $\theta$ in models like \ref{eqn0000}, for a broad class of stochastic restrictions on the shocks. The algorithm is applicable to settings with arbitrary (but finite) number of alternatives, time periods, and lags (in dynamic settings). In the case of \cite{CM}, for instance, the algorithm generates exactly the inequality restrictions \eqref{eqn0001}. Therefore, the algorithm presented in this paper automates the hitherto challenging analytical task of deriving all the conditional moment inequality restrictions that characterize the sharp identified set.\par
 I demonstrate the usefulness and versatility of the algorithm through five examples. In Example \ref{exampPPS}, I consider the static setting of \cite{PP1}, and I show that the algorithm recovers the inequality restrictions that were obtained there by analytical means. In Example \ref{exampPPE}, I consider again the setting of \cite{PP1}, but replace their conditional stationarity stochastic restriction by the alternative (and stronger) stochastic restriction of conditional exchangeability, and I use the algorithm to generate the inequality restrictions that characterize the sharp identified set. In Example \ref{exampDyn1}, I consider the dynamic binary choice setting of \cite{KPT}, and again show that the algorithm recovers the inequality restrictions that were obtained in \cite{KPT} by analytical means. In Example \ref{exampDyn2}, I consider a two-lags extension of the model of \cite{KPT}, and use the algorithm to generate the inequalities that characterize the sharp identified set. Finally, in Example \ref{exampDyn3}, I demonstrate that the algorithm can still be informative in cases where it is not applicable due to incompatibility with the maintained stochastic restrictions. In particular, I consider the dynamic multinomial choice model of \cite{PP2}, and I replace the stochastic restriction that they consider and which is incompatible with my approach, by the compatible and weaker stochastic restriction of conditional stationarity. Interestingly, the inequalities that the algorithm generates are shown to neither imply nor be implied by the inequalities derived in \cite{PP2}. Hence, the inequalities that the algorithm generates provide additional restrictions that can be used in conjunction with the inequalities derived in \cite{PP2} to generate a smaller outer region of the identified set, and this also shows that the characterization provided in \cite{PP2} is not sharp (see Remark \ref{remPPD5}). \par

The algorithm that I present in this paper builds on few ideas that appear in the sharpness proof of \cite{PP1}. Specifically as the outcome variable is discrete, "locally", i.e. for a fixed value of the parameter and covariates, the (ambient) space of the latent variable can be partitioned into a finite number of "regions", where for each time period, all realizations of the shocks within a region induce the agents to choose the same alternative. I use this partitioning to show that each "local model" is observationally equivalent to a model where the distribution of the latent variables is finite-dimensional. I refer to the process of constructing these observationally equivalent finite-dimensional versions of the local models as the "discretization" of the model, and it can be viewed as a device for dimension reduction. When the stochastic restrictions on the latent variables can be enforced by a finite number of linear equality/inequality constraints within the discretized model, the discretized local models are polytopes $-$ that is, bounded intersections of finitely many half-spaces. I show that the normal vectors of a specific subset of the hyperplanes that bound these polytopes are in one-to-one correspondence with the conditional moment inequalities that characterize the sharp  identified set of $\theta$, and I show that solving for the set of these normal vectors can be framed as a multiple objective linear program (MOLP). The approach that I discuss in this paper then essentially consists of an algorithm that constructs the discretized local models and then solves for the solutions of the associated MOLPs. There are readily available algorithms to solve general MOLPs; however, these algorithms are computationally expensive and do not scale well with the size of the problem (measured by the number of alternatives $D$ and the number of time periods $T$). In static settings, I exploit the specific structure of the MOLPs that appear in multinomial choice models, to propose an algorithm that is tailored to solving these MOLPs, and I show through simulations that my algorithm improves upon the generic algorithms, and performs well for models of moderate size. \par
This paper contributes to the literature that studies the identification and estimation of semi-parametric panel multinomial choice models without parametric restrictions on the shocks. Papers in this literature typically impose weak nonparametric restrictions on the shocks, and leave the joint distribution of the covariates and fixed effects unrestricted. The earliest paper in this literature is \cite{CM}, which established sufficient conditions for point identification of the common parameter in a static binary choice model under a stationarity restriction on the shocks. Subsequent papers that have similarly provided conditions for point identification of the common parameter in both static and dynamic multinomial choice models under stationarity-like restrictions on the shocks include \cite{SSS}, who exploited cyclical monotonicity to derive conditions for point identification in static multinomial choice models; \cite{KOT}, who proposed a general approach and conditions for point identification of static and dynamic multinomial choice models; and \cite{KPT}, who derived the conditional inequalities that characterize the identified set for the common parameter in a dynamic binary choice model, and additionally provided conditions for point identification\footnote{See also, \cite{HK1} who provided condition for point identification in dynamic binary choice models under a serial independence restriction on the shocks, \cite{EA} who derives identifying inequalities for the common parameter of static binary choice models under a conditional independence restriction, and \cite{HL} who provide conditions for point identification of a static binary choice model with predetermined covariates.}. Papers in this literature that mainly focus on partial identification include \cite{PP1}, which considers a static multinomial choice model and derives all the inequalities that characterize the sharp identified set of the common parameter, and \cite{PP2}, which considers a dynamic multinomial choice model and derives bounds on the identified set of the state dependence parameter. \cite{GL} derives conditions for identifying the index parameter in a static semi-parametric and nonseparable multinomial choice model with high-dimensional unobserved heterogeneity.\footnote{See also \cite{CR} and \cite{CRS} who provide a general characterization of the sharp identified set of the parameters of nonseparable instrumental variable models that include multinomial choice models.}\par
A related literature studies identification and estimation of semi-parametric multinomial choice models where parametric  (usually logistic) assumptions are made on the shocks and, as above, the joint distribution of fixed effects and covariates is left unrestricted. Here the identification of the common parameter often relies on the conditional likelihood approach, or on finding (identifying) moment conditions that do not depend on the fixed effects. Papers in this literature include \cite{EBA,GC1,GC2,DRC,TM,HK1,HK2,HoW,HDP}. In particular, \cite{HoW} extend the approach proposed in \cite{SB} to computationally solve for conditional moment \emph{equalities} that identify the common parameters in various dynamic panel data Logit models and, as such, their approach is similar to the one taken in the present paper. Furthermore, it can be said that the current paper shares a conceptual connection with \cite{SB}, which provides a recipe for obtaining (potentially) identifying moment \emph{equalities} for the common parameter of nonlinear panel data models with fixed effects. By contrast, the present paper offers a more concrete recipe (a computer algorithm) to generate the conditional moment inequalities that characterize the sharp identified set of the common parameter for various panel multinomial choice models.\par

Finally, from a methodological perspective, this paper is closely related to \cite{LC1}, which builds upon the approach of \cite{ZY} and converts the task of generating new entropy inequalities (information-theoretic inequalities that are not of the Shannon-type) into MOLPs. Some steps of the algorithm described in this paper are similar to those in \cite{LC1}.

The paper is structured as follows: In Section \ref{sectionStatic}, I describe the algorithm in static settings and provide various results to establish its validity. Specifically, Subsection \ref{sectionHeuristics} illustrates the approach of this paper in the simple setting of \cite{CM} to help readers develop intuition for the algorithm. Subsections \ref{sectionDiscretization} and \ref{sectionDDCP} then demonstrate how the ideas discussed in Subsection \ref{sectionHeuristics} can be extended to the general static setting, and I provide results to establish the validity of the algorithm under the extension. In Subsection \ref{sectionImp}, I discuss the limitations of the approach presented in this paper and describe stochastic restrictions under which the algorithm is not compatible. In Subsection \ref{sectionCuttingPlane}, I provide an alternative algorithm to solve for the MOLPs associated with our models, and I provide simulation results to assess its performance. In Subsection \ref{sectionExample}, I apply the algorithm developed for static settings to two examples. In Section \ref{sectionDyn}, I discuss how the algorithm needs to be modified to accommodate dynamic settings, and Subsection \ref{sectionExampDyn} provides three examples of the application of the algorithm in dynamic settings. Finally, Section \ref{sectionConclusion} concludes and gives directions for future research. All proofs are provided in the appendix.

\section{Static models}
\label{sectionStatic}
In this section, I consider the following \emph{static} panel data multinomial choice model: Individual $i$ chooses alternative $d$ ($\in[D]$) at time $t$ ($\in[T]$), denoted $Y_{ti}=d$, if and only if
\begin{equation}
\label{eqnstat}
  X_{dti}'\theta_0+\lambda_{di}+\epsilon_{dti} > X_{d'ti}'\theta_0+\lambda_{d'i}+\epsilon_{d'ti} \  \forall d' \in [D]\backslash\{d\}.
\end{equation}
Here $[D]=\{1,2,\cdots,D\}$, with $D\geq 2$, denotes the set of alternatives, and $T\geq 2$ denotes the length of the panel.
I will use the following notation: $\lambda_i=(\lambda_{1i},\cdots,\lambda_{Di})'  \in \mathbb{R}^D$, $Y_i=(Y_{1i},\cdots,Y_{Ti})'\in \mathbb{R}^D$, $X_{ti}=(X_{1ti},\cdots,X_{Dti}) \in \mathbb{R}^{d_x\times D}$, $\epsilon_{ti}=(\epsilon_{1ti},\cdots,\epsilon_{Dti})'\in \mathbb{R}^D$, $X_i=(X_{1i},\cdots,X_{Ti}) \in \mathbb{R}^{d_x\times DT}$, and $\epsilon_i=(\epsilon_{1i},\cdots,\epsilon_{Ti})\in \mathbb{R}^{D\times T}$. \par

In the analysis that follows, I will impose either one of the following two stochastic restrictions on the joint distribution of $(X_i,\lambda_i,\epsilon_i)$.

\begin{assumption}[\textbf{Stationarity}]
\hfill
\label{assumptionS}
\begin{enumerate}[a)]
\item Conditional on $X_i$ and $\lambda_i$, for each $t\in[T]$, $\epsilon_{ti}$ is continuously distributed. 
\item Conditional on $X_i$ and $\lambda_i$, the shocks $\epsilon_{ti}$ have the same distribution:
\[
\epsilon_{ti}|X_i,\lambda_i \sim \epsilon_{si}|X_i,\lambda_i \ \ \ \forall \ \ s,t \in [T] 
\]
\end{enumerate}
\end{assumption}

\begin{assumption}[\textbf{Exchangeability}]
\hfill
\label{assumptionE}
\begin{enumerate}[a)]
\item Conditional on $X_i$ and $\lambda_i$, for each $t\in[T]$, $\epsilon_{ti}$ is continuously distributed. 
\item Conditional on $X_i$ and $\lambda_i$, the joint distribution of the shocks is exchangeable:
\[
(\epsilon_{1i},\epsilon_{2i},\cdots,\epsilon_{Ti})|X_i,\lambda_i\sim (\epsilon_{\pi(1)i},\epsilon_{\pi(2)i},\cdots,\epsilon_{\pi(T)i})|X_i,\lambda_i
\]
for all permutations $\pi$ on the set $[T]$.
\end{enumerate}
\end{assumption}

\begin{remark}
\label{rem001}
It is important to note that the preceding assumptions have varying degrees of strength. Specifically, the stationarity restriction is weaker than the exchangeability restriction, and both are weaker than assuming that the shocks are conditionally independently and identically distributed (conditionally IID) (as stated in Assumption \ref{assumptionI} below).

While both stationarity and exchangeability allow for arbitrary correlation between the fixed effects $\lambda$ and covariates $X$, as well as arbitrary correlation between the random utility shocks across alternatives and time periods, they only allow for limited heteroskedasticity of the shocks, as both assumptions imply that the shocks are homoskedastic conditional on the fixed effects and the full history of covariates (the vector  $x_i$).

To avoid issues that arise when ties are present (i.e., when there is more than one alternative that maximizes the indirect utility), we require that the shocks are continuously distributed, given the covariates and fixed effects. It should be noted that this requirement can be relaxed, but doing so makes the analysis more complicated since we must consider the selection mechanism that is used to determine the agent's choice when ties are present (as discussed in \cite{PP1}).
\end{remark}
Assumptions \ref{assumptionS} and \ref{assumptionE} are examples of weak stochastic restrictions that have been widely considered in the semi-parametric discrete choice literature, and  although I have chosen to illustrate the approach of this paper under these assumptions, it is worth noting that the method presented here is versatile and can handle a wide range of other stochastic restrictions. Specifically, it can accommodate any stochastic restriction that can be encoded by a finite number of linear equality/inequality constraints once the model is "locally discretized" (see Section \ref{sectionDiscretization}). An example of a stochastic restriction for which the method of this paper is not applicable is the conditional IID stochastic restriction (see Section \ref{sectionImp}), and I show in Section \ref{sectionImp} that it cannot be encoded by finitely many linear equality/inequality constraints once the model is discretized. Using the method of this paper on a relaxation of the conditional IID restriction (Assumption \ref{assumptionE}, for instance) will yield conditional moment inequalities that characterize an outer region within which the identified set is guaranteed to lie.
%
\begin{remark}
\label{rem2}
If we set $\zeta_{ti}=\lambda_{i}+\epsilon_{ti}$, then the stationarity (resp. exchangeability) of the shocks $\epsilon_{it}$  conditional on $(x_i,\lambda_i)$ implies that the random vectors $\zeta_{ti}$ are stationary (resp. exchangeable) conditional on $x_i$. Moreover, in the absence of additional stochastic restrictions, the fixed effects can be absorbed into the shocks, and Assumption \ref{assumptionS} for instance is equivalent to assuming that $\zeta_{ti}$ are continuously distributed and are exchangeable conditional on $x_i$\footnote{\label{foot2}As there is no restriction (like a location normalization for instance) on the shocks $\epsilon_{dti}$, the model is observationally equivalent to one where the fixed effects are equal to zero. The same point is made in \cite{PP1}.  Also, by the Law of Iterated Expectation, conditional on $x$, the random vectors $\zeta_t$'s are continuously distributed, since the $\epsilon_{t}$'s are continuously distributed conditional on $\lambda$ and $x$.}. 
\end{remark}

Given the random utility model \ref{eqnstat} above, along with either Assumption \ref{assumptionS} or Assumption \ref{assumptionE}, the goal is  to characterize the identified set for the index parameter $\theta$. This is the set of all parameter values that are consistent with the model and the identified conditional choice probabilities \footnote{It is assumed throughout that the data consists of a random sample $\{(y_i,x_i)\}_{i=1}^N$, where $x_i$ and $y_i$ represent the covariates and choices of individual i across the $T$ time periods. In this case, the CCPs and the marginal distribution of the covariates is identified from the data.}. Let ${\cal M}_S$ (resp. ${\cal M}_E$) denote the set of all sequences of $T$ random vectors, each $\mathbb{R}^D$-valued, that are stationary (resp. exchangeable),
%
and let $y(\theta,x,\lambda,e)$ $\in [D]^{T}$ denote the choices that are generated by the multinomial choice model (Equation \ref{eqnstat}) when $X=x$, the true value of the index-parameter is $\theta$, the fixed effects take the value $\lambda$, and the shocks $\epsilon$ (in $\mathbb{R}^{D\times T}$) take the value $e$ \footnote{Note that $y(\theta,x,\lambda,e)$ is well defined except for a set of realizations of the shocks $\epsilon$ of probability zero.}. Here, the $t^{th}$ component of $y(\theta,x,\lambda,\epsilon)$ represents the choice that is generated by the model in period $t$. Given a value $X=x$ of the covariates and a value $\theta$ of the index-parameter,  let ${\cal P}_S(x,\theta)$ (resp. ${\cal P}_E(x,\theta)$) denote the set of all conditional choice probability vectors at $x$ that are consistent with model \ref{eqnstat}, when the true parameter value is equal to $\theta$ and Assumption \ref{assumptionS} (resp. Assumption \ref{assumptionE}) holds: That is,

\begin{equation*}
{\cal P}_S(x,\theta):=\left\{p \in \mathbb{R}^{D^T} | \  \text{there exists} \  (\lambda,\epsilon) \  \text{such that} \ p\sim y(\theta,x,\lambda,\epsilon)|x \  \text{and} \ \epsilon|x,\lambda \ \text{is stationary}  \right\},
\end{equation*}
with a similar definition for ${\cal P}_E(x,\theta)$. It then follows from Remark \ref{rem2} that ${\cal P}_S(x,\theta)$ and ${\cal P}_E(x,\theta)$ can equivalently be defined by
\begin{equation}
\label{eqnDCPS1}
{\cal P}_S(x,\theta):=\left\{p \in \mathbb{R}^{D^T} | \ \ p\sim y(\theta,x,\zeta) \text{ for some }\zeta \in {\cal M}_S    \right\}
\end{equation}
and
\begin{equation}
\label{eqnDCPE1}
{\cal P}_E(x,\theta):=\left\{p \in \mathbb{R}^{D^T} | \ \ p\sim y(\theta,x,\zeta) \text{ for some }\zeta \in {\cal M}_E     \right\},
\end{equation}
where $y(\theta,x,e)$ ($\in [D]^{T}$) represents choices that are generated by model \ref{eqnstat} when $X=x$, the true value of the parameter is $\theta$, $\lambda=0$, and $\epsilon=e$. I will use the term "local model" to refer to sets such as ${\cal P}_S(x,\theta)$ and ${\cal P}_E(x,\theta)$. Given the identified distribution of observed choices and covariates $F_{(Y,X)}$, the identified set for the parameter $\theta$ under Assumption \ref{assumptionS}, denoted $\Theta_S $ ($=\Theta_S(F_{(Y,X)})$), is the set of all parameter values $\theta$ (possibly empty if the model is misspecified) that can rationalize the observed conditional choice probabilities at all points $x$ in the support of $X$. That is:
\begin{equation}
\label{eqnS}
\Theta_{S}:=\{\theta \in \Theta | \  \forall x \in supp(X),  \ F_{Y|X=x} \in {\cal P}_S(x,\theta) \},
\end{equation}
where $\Theta\subseteq \mathbb{R}^{d_x}$ represents the parameter space. The identified set $\Theta_E$, under Assumption \ref{assumptionE}, is defined analogously.\par

 I show in Proposition \ref{prop1} below that locally (that is, when the covariates are fixed at a value $x$, and the index parameter takes a value $\theta$), the model \ref{eqnstat} is equivalent to a "discrete model" where the random utility shocks are restricted to take on a finite set of values. I use this discrete representation of the model, in conjunction with the convexity of the stochastic restrictions (\ref{assumptionS} and \ref{assumptionE}), to show that the sets ${\cal P}_S(x,\theta)$ and ${\cal P}_E(x,\theta)$ are polytopes (Proposition \ref{prop1}) \footnote{\label{foot3}A polyhedron is a set of the form $\{x|Ax\leq b \}$ (with $A\in \mathbb{R}^{m\times n}$ and $b\in \mathbb{R}^m$, $m,n\in \mathbb{N}$), and a polytope is a bounded polyhedron (see \cite{AS}).}. I then use arguments from duality theory to produce an equivalent dual descriptions of these polytopes, and I show in particular that the sought-after conditional moment inequalities that characterize the identified set coincide with the  undominated extreme points of the \textit{duals} of these polytopes (details are provided further below).

In the following section, I use the setting of \cite{CM} to illustrate how the approach taken in this paper generates the conditional moment inequalities that characterize the identified set. As the aim is to build intuition, I will gloss over some of the details, which I will discuss further below when I consider the general setting. 

\subsection{Heuristics}
\label{sectionHeuristics}
The following model is that of \cite{CM}. In a two-periods binary choice model, let agent $i$'s choice at time $t\in [2]$ be determined by Equation \ref{eqnMans}, and suppose that Assumption \ref{assumptionS} holds (i.e., the shocks $\epsilon_{ti}$ are continuously distributed and stationary, conditional on the covariates and fixed effects). \cite{CM} showed that the following three conditional moment inequalities hold under the stationarity assumption:\footnote{Note that the inequalities \ref{eqnCM} differ slightly from those in \cite{CM}. This is due to the fact that I do not assume, as in \cite{CM}, that the conditional cumulative distribution function of the shocks (given the covariates and fixed effects) is strictly increasing.}
\begin{equation}
\label{eqnCM}
\begin{aligned}
P(Y_1=1|X=x)   &\leq P(Y_2=1|X=x)  \quad\text{whenever}\quad  x_1'\theta < x_2'\theta \\
P(Y_1=1|X=x)   &\geq P(Y_2=1|X=x)  \quad\text{whenever}\quad  x_1'\theta > x_2'\theta \\
P(Y_1=1|X=x)   &= P(Y_2=1|X=x)  \quad\text{whenever}\quad  x_1'\theta = x_2'\theta 
\end{aligned}
\end{equation}
I show in the following paragraphs how the approach of this paper can be used to obtain these inequalities by computational means.\par

 The first step of the approach of this paper consists in constructing discrete models that are locally (at a fixed value of the covariate $x$ and the parameter $\theta$) observationally equivalent to the model \ref{eqnMans} under the stationarity restriction. This \textit{discretization} step can be thought of as a dimensionality reduction, as it shows that although the distribution of the shocks in each local model is a-priori in an infinite-dimensional space (the space of stationary distributions on $\mathbb{R}^2$), there is an observationally equivalent version of each local model where the shocks are restricted to take a finite set of values (thus restricting their distributions to a finite-dimensional space).
 The idea of the discretization is borrowed from the proof of Theorem 2 in \cite{PP1}. \par

Fix $x \in supp(X)$ (all probability statements below are made conditional on $X=x$), and $\theta \in \Theta$,  and let $v_t=x_t'\theta$, for $t\in [2]$. As it will become clear from the computations that follow, the discrete (and locally equivalent) models that we construct are totally determined by the order relation between the indices $v_1$ and $v_2$. There will thus be three cases to consider: $v_2>v_1$, $v_2<v_1$ and $v_2=v_1$. Suppose that $v_2>v_1$ (the other two cases are handled analogously). From $v_1$ and $v_2$, we construct the intervals $I_1=(-\infty,v_1)$, $I_2=(v_1, v_2)$ and $I_3=(v_2,+\infty)$. Let the rectangles $R_{ij}$, $i,j\in [3]$ be defined by $R_{ij}=I_i\times I_j$, and let $p_{dd'}:=P(y_1=d,y_2=d'|x)$, $d,d' \in \{0,1\}$. For $i,j \in [3]$, let $q_{ij}$ denote the probability that $\zeta$ ($=(\zeta_1,\zeta_2)'$) is in $R_{ij}$, where $\zeta_{ti}:=\lambda_{ti}+\epsilon_{ti}$, i.e, $q_{ij}=P(\lambda+\epsilon_{1} \in I_i, \lambda+\epsilon_2 \in I_j)$. Each of the conditional choices probabilities $p_{dd'}$ can be expressed in terms of the sums of the $q_{ij}$'s: For instance
\begin{align*}
p_{11}&=P(v_1-\zeta_1>0, v_2-\zeta_2> 0)=P(\zeta_1 <v_1, \zeta_2 <v_2) \\
&=P(\zeta_1 \in I_1, \zeta_2 \in I_1\cup I_2)=P(\zeta \in R_{11} \cup R_{12})=q_{11}+q_{12}.
\end{align*}
Carrying out similar computations for the other 3 events yields the following linear relations

\[ 
\begin{bmatrix}
           p_{00} \\
           p_{01} \\
           p_{10}\\
          p_{11}         
         \end{bmatrix}
=
\bordermatrix{& q_{11}&q_{12}&q_{13}&q_{21}&q_{22}&q_{23}&q_{31}&q_{32}&q_{33}\cr
&0&0&0&0&0&1&0&0&1\cr
&0&0&0&1&1&0&1&1&0\cr
&0&0&1&0&0&0&0&0&0\cr
&1&1&0&0&0&0&0&0&0
}\qquad
\]
which I write in matrix form as $p=A(x,\theta)q$. The stationarity restriction on the shocks implies (see Remark \ref{rem2}) that the vector $q$ satisfies the linear restriction
\[
 \begin{bmatrix}
           0 \\
           0\\
           0       
         \end{bmatrix}
=
\bordermatrix{& q_{11}&q_{12}&q_{13}&q_{21}&q_{22}&q_{23}&q_{31}&q_{32}&q_{33}\cr
&0&1&1&-1&0&0&-1&0&0\cr
&0&-1&0&1&0&1&0&-1&0\cr
&0&0&-1&0&0&-1&1&1&0
}\qquad
\]
Which I write in matrix form as $R_S(x,\theta) q=0$ \footnote{\label{foot4}For instance by stationarity, we must have $P(\zeta_1\in I_1)=P(\zeta_2\in I_1)$ which is equivalent to $q_{11}+q_{12}+q_{13}=q_{11}+q_{21}+q_{31}$ or $q_{12}+q_{13}-q_{21}-q_{31}=0$, which corresponds to the first row of the restriction $R_Sq=0$. There are two other restrictions, and any one of these three restrictions is redundant given the other two, and dropping either one of these restrictions does not change the conclusion of our analysis.}.  Hence for all $p\in {\cal P}_S(x,\theta)$, there exists a $q\in \mathbb{R}^9$ such that $q\geq 0$, $R_Sq=0$ and $p=Aq$ \footnote{\label{ftnote1}Since each column of $A$ has a single non-zero entry equal to 1, the relation$p=Aq$ combined with p is a probability vector, imply that q is also a probability vector: $1=\mathbbm{1}^T p=\mathbbm{1}^T A q=\mathbbm{1}^T q$.}, and we get
\[
{\cal P}_S(x,\theta) \subseteq \{p\in R^{4}| \  p^T \mathbbm{1}=1\}\cap \{p=Aq | \ q\geq 0 \ \ R_S q=0\}.
\]
The reverse inclusion also holds, for if $q$ is a probability vector such that $R_Sq=0$, then there exists a continuously distributed random vector $\zeta\in {\cal M}_S$ such that $q_{ij}=P(\zeta\in R_{ij})$ (see Proposition \ref{prop1} below). Therefore, no information is lost in the discretization, and we have
\begin{equation}
\label{eqnh1}
{\cal P}_S(x,\theta) = \{p\in R^{4}| \  p^T \mathbbm{1}=1\}\cap \{p=Aq | \ q\geq 0, \ \ R_S q=0\}.
\end{equation}
The right hand side of \ref{eqnh1} shows that the set ${\cal P}_S(x,\theta)$ is equivalent to a model where the random utility shocks are discrete and take a finite number of values (given by the dimension of $q$). Another consequence of \ref{eqnh1} is that it shows that the set ${\cal P}_S(x,\theta)$ is a polytope, and a probability vector $p$ belongs to ${\cal P}_S(x,\theta)$ if and only if it belongs to $\{p=Aq | \ q\geq 0, \ \ R_S q=0\}$. As the restriction $p^T \mathbbm{1}=1$ is automatically satisfied by the identified CCPs (as they are probability vectors), our task then reduces to characterizing the polyhedral cone
\[
{\cal C}:=\{p=Aq | \ q\geq 0, \ \ R_S q=0\}.
\]
In particular, we want to find a concise affine description of the set ${\cal C}$; that is we want to find a minimal set of inequalities that are satisfied by a vector $p$ if and only if $p$ belongs to ${\cal C}$ (these are the so-called \textit{facet defining inequalities} for ${\cal C}$, and each polyhedron admits such an affine description$-$See \cite{AS}). \par
Once the (local) discrete models are constructed as in \ref{eqnh1}, the second step of the approach of this paper consists in using duality theory to give a dual representation of the set ${\cal C}$ that will be used in the third step to set up a computational problem to solve for its facet defining inequalities. I will show further below that the facet defining inequalities of ${\cal C}$ correspond to the conditional moment inequalities that characterize the identified set. The starting point of obtaining a dual characterization of ${\cal C}$ is Farkas' Lemma\footnote{Farkas' Lemma states that the equation $Mx=b$ has a non-negative solution (i.e., a solution $x\geq0$ component-wise) if and only if for all vectors $\lambda$ such that $\lambda^TM\leq 0$, we have $\lambda^Tb\leq 0$. In our context, we can take $M=(A^T, -R_S^T)^T$ and $b=(p^T ,0)^T$.}, which provides a partial answer, as it implies that  

\begin{equation}
\label{eqnh2}
\left[p \in {\cal C}\right] \   \text{if and only if} \   \left[ y'p \leq 0 \  \text{for all} \ y \in \mathbb{R}^4 \  \text{s.t.}\ A^Ty\leq R_S^T z \ \text{for some}\ z \in \mathbb{R}^3 \right].
\end{equation}
A consequence of \ref{eqnh2} is that any $y$ that belongs to  ${\cal Q}(x,\theta):=\{y\in \mathbb{R}^4| \ A^T y \leq R_S^Tz,  \text{for some}\ z \in \mathbb{R}^3\}$  provides a valid inequality restriction that must hold for all CCP vectors at $x$ that are consistent with the model for the parameter value $\theta$, as we necessarily have $y^Tp\leq 0$ for all $p \in {\cal P}_S(x,\theta)$. Moreover, \ref{eqnh2} implies that the set ${\cal Q}$ ($={\cal Q}(x,\theta)$) contains all the inequality restrictions that the model places on elements of  ${\cal C}$, since $p$ belongs to ${\cal C}$ if and only if $y^Tp\leq 0$ for all $y\in {\cal Q}$. However, all the inequalities in ${\cal Q}$ are not an economical way to to summarize the restrictions that the model places on the CCPs at $x$, and a more economical (using less inequalities) summary of the same restrictions is provided by considering only the extreme rays of ${\cal Q}$ (note that ${\cal Q}$ is a polyhedral cone). Alternatively, if we let ${\cal Q}_0$ denote the intersection of ${\cal Q}$ with the set $\{y\in \mathbb{R}^4| \ ||y||_{\infty}\leq 1\}$, i.e,

\[
{\cal Q}_0:={\cal Q}\cap \{y\in \mathbb{R}^4| \ A^T y \leq R_S^Tz,  \text{for some}\ z \in \mathbb{R}^3, \ ||y||_{\infty}\leq 1\},
\] 
then ${\cal Q}_0$ is a polytope, and its extreme points  provide a summary of all of the restrictions that the model places on elements of ${\cal C}$\footnote{ Indeed since ${\cal Q}_0$ is a polytope, it has a finite number of extreme points, say $\{y_i\}_{i\in[m]}$, and each element of ${\cal Q}_0$ is equal to a convex combination of these extreme points. Hence, if $p^Ty_i\leq 0$ for all $i\in[m]$, then we necessarily have that $y^Tp\leq 0$ for all $y\in {\cal Q}_0$.}. Moreover, as the inequalities in ${\cal Q}_0$ represent restrictions on non-negative vectors ($p\in {\cal C}$ implies that $p\geq 0$), we can further reduce the set of inequalities needed to characterize ${\cal C}$ by considering only the undominated extreme points of ${\cal Q}_0$; these are the extreme points $y$ of ${\cal Q}_0$ such that there does not exists $y'\in {\cal Q}_0$ such that $y'\neq y$ and $y'\geq y$ (component-wise). Indeed, note that if $\tilde{y}\leq y$, then whenever $y^T p\leq 0$ for some $p\geq 0$, we necessarily have $\tilde{y}^T p\leq 0$. Hence, the restriction that correspond to the undominated extreme points of ${\cal Q}_0$ are sufficient to sharply characterize ${\cal C}$ (and thus ${\cal P}_S(x,\theta)$).\par
The third and final step of  the approach of this paper consists in solving for the undominated extreme points of ${\cal Q}_0$. To do so, I frame the problem of solving for the undominated extreme points of ${\cal Q}_0$ as a multiple-objective linear program (MOLP). Here, a MOLP simply consists of solving for \textit{all} of the undominated extreme points of the image of a set $\{x|Mx\leq b\}$ under a matrix $C$ (see \cite{HB1} or \cite{LC2}), and MOLP are written compactly as:
\begin{equation}
\label{eqnMolp1}
\begin{aligned}
& Vmax \ \ \ Cx \\
&s.t. \ \ \ x\in \{x:Mx\leq b\}
\end{aligned}
\end{equation}
where $M\in \mathbb{R}^{m\times n}$, $b\in \mathbb{R}^m$ and $C\in \mathbb{R}^{p\times n}$ \footnote{For each $i\in[p]$, the row $C_i$ can be thought of as the objective/utility of the $i^{th}$ individual, and each element of $\{x:Mx\leq b\}$ can be thought to represent a feasible allocation. A MOLP then consists of finding all of the extreme points of the Pareto frontier (in the objective/utility space)}.\par

In the present context, solving for the undominated extreme points of ${\cal Q}_0$ is equivalent to the MOLP
\[
\begin{aligned}
\label{eqnh3}
& Vmax \ \ \ Cx \\
&s.t. \ \ \ x\in \{(y,z)' \in \mathbb{R}^7 \ | \ A^Ty\leq R_S^Tz, \ ||y||_{\infty}\leq 1\}
\end{aligned}
\]
where $C=(\mathit{I}_{4\times 4}, \mathit{0}_{4\times 3})$. There are many known algorithms to solve MOLPs; using Benson's outer approximation algorithm (see \cite{HB1}) to solve \ref{eqnh3} yields the following two solutions: $\tilde{y}_0=(0,0,0,0)$ and $\tilde{y}_1=(0,-1,1,0)$.
 The vector $\tilde{y}_0$ yields the trivial restriction $0^Tp\leq 0$ for all $p \in {\cal P}_S(x,\theta)$; it  is clearly feasible (i.e, belongs to ${\cal Q}_0$), and it is undominated as we would otherwise have a restriction implying that some components of the observed CCPs must be zero. The inequality that corresponds to $\tilde{y}_1$ is:
\[
p_{10}\leq p_{01} \ \ \text{for all} \ \ \ p \in {\cal P}_S(x,\theta).
\]
Adding $p_{11}$ to both sides of the inequality yields
\[
P(Y_1=1|X=x) \leq P(Y_2=1|X=x) \ \ \text{for all} \ \ \ p \in {\cal P}_S(x,\theta).
\]
Since the polytope ${\cal P}_S(x,\theta)$ is the same for all $\theta$ and $x$ such that $x_1'\theta<x_2'\theta$, we have in fact shown that
\[
P(Y_1=1|X=x) \leq P(Y_2=1|X=x) \ \ \text{whenever} \ \ \ x_1'\theta<x_2'\theta ,
\]
and this is the only inequality that a CCP vector $p$ must satisfy to belong to ${\cal P}_S(x,\theta)$.
Arguing by symmetry (by relabeling the time periods), we immediately get 
\[
P(Y_2=1|X=x) \leq P(Y_1=1|X=x) \ \ \text{whenever} \ \ \ x_2'\theta<x_1'\theta .
\]
Repeating the same construction with $v_0=v_1$ yields
\[
P(Y_2=1|X=x) = P(Y_1=1|X=x) \ \ \text{whenever} \ \ \ x_2'\theta=x_1'\theta .
\]
These are the inequalities of \cite{CM}, which were derived there by analytical means. An advantage of the current approach is that the inequalities that we get from our algorithm are sharp by construction, as they exhaust all the information in ${\cal Q}_0$. By contrast, \cite{CM} only established the validity of these inequalities, but did not show sharpness. Sharpness was only established recently as a corollary of a more general result of \cite{PP1}. As can be seen (for instance) by the arguments in  \cite{PP1} and \cite{KPT}, establishing sharpness is not usually an easy task. I thus view it as a main advantage of the approach of this paper that the procedure always generates inequalities that are sharp by construction. \par
Below, I show that the foregoing program can be carried through in both static and dynamic settings (with a fixed but arbitrary number of lags) with $D\geq 2$ alternatives and $T\geq 2$ time periods. Thus, absent computational limitations, the procedure of this paper provides a way to automate the the derivation of all the conditional moment inequality restrictions that characterize the identified set of the common parameters.
\subsection{Discretization and construction of the  DCPs}
\label{sectionDiscretization}
In this and the next section, I show how the procedure discussed in Section \ref{sectionHeuristics} can be generalized to the static setting with $D\geq 2$ and $T\geq 2$. I proceed as in Section \ref{sectionHeuristics}, and present the procedure in three steps. All of the steps are applied locally (i.e., for a fixed value of the covariate $x$ and a fixed parameter value $\theta$). This section discusses the first step, and provides a generalized discretization scheme. In Section \ref{sectionDDCP}, I discuss the second and the third step of the procedure, and show how the discretization can be used to set up MOLPs whose solutions correspond to the conditional moment inequalities that characterize the identified set. Since the three steps of the procedure are applied locally, changing the covariate value $x$ or the parameter $\theta$ may generate some new inequality restrictions, and we may need to carry out the procedure at all covariate and parameter values pairs to generate all the desired inequality restrictions. I will show, however, in Proposition \ref{propPatch} below that we always only need to consider a finite number of \textit{local models} to generate all the inequalities that characterize the identified set, even if some of the components of $X$ are continuously distributed (for instance, in Section \ref{sectionHeuristics}  we only needed to consider 3 local models to generate all of the inequalities that characterize $\Theta_S$).  The discretization scheme is borrowed from \cite{PP1} where they consider a setting with $T=2$. I show here how it can be extended to $T\geq 2$, and will discuss in  Section \ref{sectionDyn} how it can be extended to dynamic models. Similar discretization scheme  appear in \cite{KS} and \cite{TTY}, and some of the terminology that I use below is borrowed from these papers.\footnote{See also \cite{GRS} for another example of the utilization of discretization schemes.}\par

The model under consideration in this section is the static model (i.e, Equation \ref{eqnstat}), where I impose either Assumption \ref{assumptionS} or \ref{assumptionE}. Fix a covariate value $X=x$ and a parameter value $\theta$, and set $v=(v_1,\cdots,v_T)\in \mathbb{R}^{D\times T}$, where $v_{t}=(v_{1t},\cdots,v_{Dt})'$ and $v_{dt}:=x_{dt}'\theta$. Here, $v_t$ represents the vector of the deterministic components of the indirect utilities at time $t$, when $X=x$ and the index parameter takes the value $\theta$. For each time period $t\in[T]$, the space of period $t$ shocks, $\mathbb{R}^{D},$ can be partitioned  \footnote{\label{foot8}The partition is up to the measure zero set formed by the union of the boundaries of the regions.} into the $D$ regions $\{{\cal \varepsilon}_{d,t}\}_{d=1}^D$, where ${\cal \varepsilon}_{d,t}$ denotes the set of all realizations of the shocks that induce the agent to choose alternative $d$ at time $t$:
\begin{equation}
\label{eqnPart}
{\cal \varepsilon}_{d,t}:=\{\zeta \in \mathbb{R}^D | \ v_{dt}+\zeta_{d}>v_{d't}+\zeta_{d'}\ \ \forall d'\neq d, \ d'\in[D]\}.
\end{equation}
Here $\zeta:=\lambda+\epsilon$ represents the composite error term. We can further partition the space $\mathbb{R}^D$ by considering the coarsest partition that generates all sets in $\{\varepsilon_{dt}\}_{d\in[D],t\in [T]}$. We refer to the sets in this coarsest partition as \textit{patches.} Following, \cite{TTY}, the set of all patches represents what can be called the \textit{minimal relevant partition}, and all the realizations of the shocks in a patch induce agents to make a specific choice in each time period $t\in[T]$. The patches can be determined as follows: A sequence $(d_1, d_2,\cdots,d_T)\in [D]^T$ determines a patch if and only if:
\[
\varepsilon_{d_1,1}\cap\varepsilon_{d_2,2}\cap\cdots\cap \varepsilon_{d_T,T}\neq \emptyset.
\]
The set of all patches is denoted by $F$:
\[
F:=\{(d_1,d_2,\cdots,d_T)\in [D]^T \ | \ \varepsilon_{d_1,1}\cap\varepsilon_{d_2,2}\cap\cdots\cap \varepsilon_{d_T,T}\neq \emptyset\}.
\]
Patches, in turn, can be used to partition $\mathbb{R}^{D\times T}$ (the domain of the vector of shocks across all D alternatives and T time periods) into the \textit{rectangular regions}  of the type $f_1\times f_2\times \cdots \times f_T$, with $f_i \in F$ for all $i\in [T]$. The set of all such rectangular regions is denoted ${\cal R}$:
\[
{\cal R}:=\{f_1\times f_2\times \cdots \times f_T\ \ | \ f_i \in F, \ \forall i \in [T]\}.
\]
\par

Given $v$ as an input, a simple algorithm to determine the set of all patches is the following: The sequences $(d_1,\cdots,d_T)$ that correspond to patches are those for which the value of the following linear program is zero (i.e, the constraint region is non-empty \footnote{\label{foot9}By convention (see \cite{AS}), a minimization linear program has a value of $\infty$ if the constraint region is empty. })
\begin{equation}
\label{eqnPatch}
\begin{array}{ll@{}ll}
\text{minimize}  & 0 &\\
\text{subject to} & &\zeta \in \mathbb{R}^{D} &\\
&    \zeta_{d'}-\zeta_{d_t}&\leq v_{d_tt}-v_{d't}-\delta &\text{for all} \ t \in [T]\  \text{and} \ d'\in [D]\backslash\{d_t\},
                
\end{array}
\end{equation}
where the tolerance parameter $\delta$  is small positive constant (can for instance be set equal to $10^{-4}$) that is used to replace the strict inequalities that appear in the definition of the sets $\varepsilon_{d,t}$ by weak inequalities. Thus the set of all patches can be determined by solving $D^T$ linear programs (LP).\par
\begin{remark}
\label{rem3}
For fixed $T$, it is possible to replace the LP that determine the patches by explicit rules, at the cost of some extra derivations. When $T=2$, for instance, $(d,d')$ is a patch, with $d\neq d'$, if and only if $\Delta v_d<\Delta v_{d'}$, where $\Delta v_d:=v_{d2}-v_{d1}$ for all $d\in[D]$. Indeed, $\varepsilon_{d,1}\cap\varepsilon_{d',2}\neq \emptyset$, if and only if there exists $\zeta_d$ and $\zeta_d'$ in $\mathbb{R}$, such that $\zeta_{d}-\zeta_{d'}>v_{d'1}-v_{d1}$ and $\zeta_{d}-\zeta_{d'}<v_{d'2}-v_{d2}$ \footnote{\label{foot10}This is the case since we can effectively restrict the agent's decision to be between the alternatives $d$ and $d'$, by making the shocks of the remaining alternatives $d''\notin \{d,d'\}$ take arbitrarily large negative values.}. There exists a pair $(\zeta_d,\zeta_{d'})$ that solves both inequalities if and only if $v_{d'1}-v_{d1}<v_{d'2}-v_{d2}$, and the latter is equivalent to $\Delta v_d<\Delta v_{d'}$. A similar argument shows that $(d,d)$ is always a patch for all $d\in [D]$. Replacing the LPs that determine the patches by such explicit rules can speed up computations and is recommended when $D^T$ (the number of LPs that the procedure outline above solves) is large.
\end{remark}

As in Section \ref{sectionHeuristics}, I show that the sets ${\cal P}_r(x,\theta)$ (for $r\in \{S,E\}$) defined by equations \ref{eqnDCPS1} and \ref{eqnDCPE1} are polytopes. In particular, I show that they can be written in the form ${\cal P}_r(x,\theta)=\{p\in \mathbb{R}^{DT}| \mathbbm{1}^Tp=1\}\cap \{p=A(x,\theta)q| \ q\in \mathbb{R}^{|{\cal R}|},\ R_r (x,\theta)q=0, \ q\geq 0\}$  for some matrices $A(x,\theta)$ and $R_r(x,\theta)$ whose construction I now discuss. I will refer to the sets $\{p=A(x,\theta)q| \ q\in \mathbb{R}^{|{\cal R}|},\ R_r (x,\theta)q=0, \ q\geq 0\}$ as \textit{discrete choice polyhedrons} (DCPs) \footnote{\label{foot11}These sets are polyhedral cones, and becomes a polytope when intersected with the set $\{p\in \mathbb{R}^{D^T}| \mathbbm{1}^Tp=1\}$, which is in some sense redundant in our characterization of the identified set, as the identified conditional choice probabilities necessarily satisfy the restriction $\mathbbm{1}^Tp=1$. The terminology of \textit{discrete choice polyhedron} is similar to that of \textit{multiple choice polytope} that is used in the mathematical economics and mathematical psychology literature to refer to the set of all stochastic choice functions that are consistent with a random utility model (see \cite{Dos}, \cite{DMF1}, \cite{PF} and references therein)}. The matrix $A$ is a $D^T$-by-$|{\cal R}|$ matrix, with each row corresponding to a sequence of choices  (i.e, each row is indexed by a sequence $d_1,d_2,\cdots, d_T$, with $d_i \in [D]$), and each column corresponding to a region $\textbf{R} \in {\cal R}$. Given a choice sequence $\textbf{d}=(d_1,\cdots,d_T)\in [D]^T$ and a region $\textbf{R}=\textbf{f}_1\times \textbf{f}_2\times \cdots\times \textbf{f}_T \in {\cal R}$, the entry corresponding to the $\textbf{d}^{th}$-row and $\textbf{R}^{th}$-column is equal 1 (and equal to 0 otherwise) if and only if for each $t\in [T]$, the shocks in the patch $\textbf{f}_t$ induces the agent to choose alternative $d_t$ at time $t$ (i.e, if the patch $\textbf{f}_t$ is given by $\textbf{f}_t=(d^{(t)}_1,d^{(t)}_2,\cdots,d^{(t)}_T)$, then $d^{(t)}_t=d_t$). Intuitively, all the entries of the $\textbf{d}^{th}$ row equal to one, represent all of the regions in the partition ${\cal R}$ that induce the agent to choose alternative $d_t$ for each $t\in [T]$.\par
The matrices $R_r$, $r\in \{E,S\}$ are easy to construct, and simply enforce the stationarity or exchangeability restriction on the probability vector $\textbf{q}\in \mathbb{R}^{|{\cal R}|}$, which gives the probability that the shocks belong to the different regions of the partition $ {\cal R}$. That is, the matrix $R_S$ encodes the restriction: $\forall \textbf{f} \in F$, and for all $i,j \in [T]$

\begin{equation*}
\sum_{\{\textbf{f}_1 \times\cdots\times \textbf{f}_T | \ \textbf{f}_t \in F \ \forall t\in [T], \  \textbf{f}_i=\textbf{f} \}} \textbf{q}_{\textbf{f}_1\times\cdots\times \textbf{f}_T}=\sum_{\{\textbf{f}_1 \times\cdots\times \textbf{f}_T | \ \textbf{f}_t \in F \ \forall t\in [T], \  \textbf{f}_j=\textbf{f} \}} \textbf{q}_{\textbf{f}_1\times\cdots\times \textbf{f}_T},
\end{equation*}
and the matrix $R_E$ encodes the restriction: For all $f_1\times\cdots\times f_T \in {\cal R}$ and for all permutation $\pi$ on $[T]$
\begin{equation*}
\textbf{q}_{\textbf{f}_1\times\cdots\times \textbf{f}_T}=\textbf{q}_{\textbf{f}_{\pi(1)}\times\cdots\times \textbf{f}_{\pi(T)}}.
\end{equation*}

By construction, we always have ${\cal P}_r(x,\theta)\subseteq \{p\in \mathbb{R}^{DT}| \mathbbm{1}^Tp=1\}\cap \{p=A(x,\theta)q| \ q\in \mathbb{R}^{|{\cal R}|},\ R_r (x,\theta)q=0\}$, for $r\in \{S,E\}$. I show below that the reverse inclusion also holds, and no identifying information is lost when we consider the (local) discrete models.  For this, it suffices to show that for each probability vector $\textbf{p}=A(x,\theta) \textbf{q}$, where $\textbf{q}$ is a probability vector that satisfies $R_r(x,\theta)\textbf{q}=0$, there exists (conditional on $X=x$) a continuously distributed vector of shocks $\zeta \in \mathbb{R}^{D \times T}$ that is observationally equivalent to $\textbf{q}$ (in that it also generates the same CCP vector $\textbf{p}$), and that satisfies the same stochastic restriction as \textbf{q} (i.e., $\zeta$ is stationary (resp. exchangeable) if \textbf{q} is stationary (resp. exchangeable)).  A similar result for the stationary case appears in \cite{PP1}.\footnote{A similar result is also established in Proposition 1 of \cite{TTY}.}
\begin{proposition}
\label{prop1} 
Let ${\cal P}_r(x,\theta)$, with $r\in \{E,S\}$, be defined as in equations \ref{eqnDCPS1} and \ref{eqnDCPE1}. Then 
\[
{\cal P}_r(x,\theta)=\{p\in \mathbb{R}^{D^T}| \mathbbm{1}^Tp=1\}\cap \{p=A(x,\theta)q| \ q\in \mathbb{R}^{|{\cal R}|},\ R_r q=0, \ q\geq 0\}.
\]
\end{proposition}

\subsection{Construction of the DDCPs and characterization of the identified set}
\label{sectionDDCP}
In this section, I first give a dual characterization of the sets ${\cal P}(x,\theta)$ using their discrete representation derived in Section \ref{sectionDiscretization}. I then use the dual characterization of the DCPs to set up MOLPs whose solutions produces the conditional moment inequalities that characterize the identified set. In the main result of this section, Theorem \ref{thm1}, I show that applying the procedure at a finite number of local models suffices to generate all of the inequalities that characterize the sharp identified set. The dual characterization of the sets  ${\cal P}(x,\theta)$ that I present below was used in \cite{PP1} to prove that some inequalities that they derived (analytically) were sufficient to characterize the identified set. One of the main observations in this paper is that the undominated extreme points of the duals of the sets ${\cal P}(x,\theta)$ correspond to the conditional moment inequalities that characterize the identified set. Thus deriving the moment inequalities that characterize the identified set can be reduced to the task of solving for the undominated extreme points of some polytopes, which in turn can be framed as MOLPs and solved computationally.

Recall that given a convex set ${\cal C}$ in $\mathbb{R}^n$, its dual cone is the set ${\cal C}^*$ defined by 
\[
{\cal C}^*:=\{y \in \mathbb{R}^n | \ y^Tx\leq 0 \ \ \forall \ x \in {\cal C}\}.
\]
When ${\cal C}$ is a closed and convex cone, the dual of its dual is itself, i.e, ${\cal C}^{**}={\cal C}$ (see \cite{AS}). Thus the dual  cone of a closed convex cone provides an alternative and equivalent description of ${\cal C}$ (i.e, ${\cal C}$ and its dual ${\cal C}^*$ are in one-to-one correspondence). If ${\cal C}$ is a DCP, i.e, a set of the type ${\cal C}=\{p=Aq | \ q\geq 0, Rq=0\}$, then ${\cal C}$ is a closed and convex cone, and its dual cone is the set of all $y's$ such that $\sup_{p \in {\cal C}}y^Tp\leq 0$. By the (strong) duality theorem of linear programing, we have

\begin{align*}
\sup_{p \in {\cal C}}y^Tp&=\sup_{\{q| \ q\geq 0, \ Rq=0\}} y^TAq\\
&=\inf_{\{z| \ A^T y\leq R^T z\}} 0.
\end{align*}
The last term is either $\infty$ or $0$ depending on whether the set $\{z| \ A^T y\leq R^T z\}$ is empty or not. Thus the dual cone of ${\cal C}^*$ is given by ${\cal C}^*=\{y | \ \exists z \ \text{s.t.} \ A^T y \leq R^T z\}$ \footnote{\label{foot12} Following \cite{PP1}, the same conclusion is obtained with the use of Farkas' Lemma. The alternative derivation here is mainly used to justify why I refer to these sets as the \textit{dual} discrete choice polyhedrons. The same conclusion can also be reached by considering the polar cone of ${\cal C}$ (see \cite{AS} pg. 65).}. By the preceding duality result, we have 
\begin{align*}
{\cal C}&={\cal C}^{**}=\{p | \ p^Ty \leq 0 \ \forall y \in {\cal C}^*\}\\
&=\{p | \ p^Ty \leq 0 \ \forall y \in {\cal C}^*, \ ||y||_{\infty}\leq 1\}
\end{align*}
where the inclusion of the normalization $||y||_{\infty}$ clearly does not change the set ${\cal C}$, as ${\cal C}$ is a cone. When  the sets ${\cal C}$ are DCPs, I will refer to the associated polyhedrons $\{y | \ \exists z \ \text{s.t.} \ A^T y \leq R^T z, \ ||y||_{\infty}\leq 1\}$ as the \textit{Dual Discrete Choice Polyhedrons} (DDCPs). Each DDCP is a polytope, and as such it admits a finite number of extreme points, say $\{y_i\}_{i=1}^m$, and a simple argument shows that 
\[
{\cal C}=\{p | \ p^Ty_i \leq 0 \ \forall i \in [m]\}.
\]
Thus the DCP is completely determined by the extreme points of the DDCP. As the number of extreme points of the DDCPs can be large, a simpler characterization can be obtained by considering the notion of dominance defined below.
\begin{definition}
\label{defDom}
A point $y$ belonging to a convex set ${\cal C}$ is dominated (in ${\cal C}$) if there exist $y'\in {\cal C}$ such that $y'\neq y$ and $y'\geq y$ (component-wise). A point $y \in {\cal C}$ is undominated if it is not dominated. \footnote{\label{foot13}Equivalently, using the Separating Hyperplane Theorem, a point $y' \in {\cal C}$ is undominated iff there exists a $w \in \mathbb{R}^n$, with $w> 0$ (i.e., all components are strictly positive), such that $y' \in \ argmax_{\{y\in {\cal C}\}} y^Tw$ (see \cite{YZ}).}\footnote{Note that \cite{LC1} uses the term "non-dominated" instead of "undominated".}
\end{definition}
As the DCPs are contained in $\mathbb{R}_{+}^{n}$, if the set of undominated extreme points of the DDCP is given by $\{y_j\}_{j=1}^{m'}$, then a simple argument shows that 
\[
{\cal C}=\{p | \ p\geq 0, \ p^Ty_i \leq 0 \ \forall i \in [m']\}.
\]
That is, the DCP is completely determined by the undominated extreme points of the DDCP.\par

Solving for the undominated extreme points of a DDCP can be framed as a MOLP (see \cite{HB1} and equation \ref{eqnMolp1}). Indeed the undominated extreme points of the set $\{y | \ \exists z \ \text{s.t.} \ A(x,\theta)^T y \leq R_r(x,\theta)^T z, \ ||y||_{\infty}\leq 1\}$, $r\in \{E,S\}$, are the solutions to the MOLP:

\begin{equation}
\label{eqnMolp}
\begin{aligned}
& Vmax \ \ \ Cx \\
&s.t. \ \ \ x\in \{(y,z)' \in \mathbb{R}^{d_y+d_z} \ | \ A(x,\theta)^Ty\leq R_r(x,\theta)^Tz, \ ||y||_{\infty}\leq 1\}
\end{aligned}
\end{equation}
where $C=(\mathit{I}_{d_y\times d_y}, \mathit{0}_{d_y\times d_z})$, $dim(y)=d_y$ and $dim(z)=d_z$\footnote{The dimension $d_z$ of $z$ depends on the number of non-empty patches obtained from the discretization step. The dimension $d_y$ of $y$, on the other hand is equal to $D^T$, the number of rows of $A$.}. Let ${\cal I}_S(x,\theta)$ (resp. ${\cal I}_E(x,\theta)$) denote the set of all the solutions to the MOLP \ref{eqnMolp} under the stationarity (resp. exchangeability) restriction.\par
Proposition \ref{propPatch} below shows that the sets ${\cal I}_r(x,\theta)$ (for $r\in \{E,S\}$), defined above, only take a finite number of values as $x$ and $\theta$ vary through their respective domains ${\cal X}$ and $\Theta$. The main implication of the proposition is that we only need to consider a finite number of local models (and by consequence we only need to solve for a finite number of MOLPs) to obtain all of the conditional moment inequalities that characterize the sharp identified set, this being the case even if some components of the explanatory variables $X$ are continuously distributed. The proposition is a direct consequence of the fact that from the discretization scheme in Section \ref{sectionDiscretization} the sets ${\cal P}_r(x,\theta)$ are completely determined by the set of patches that are used in their discrete representation, and there is only a finite number of possible configurations of patches (see Appendix A for details). 
\begin{proposition}
\label{propPatch}
Fix $r\in \{E,S\}$. Then there exists some $m\in \mathbb{N}$, a finite partition $\{O_k\}_{k\in [m]}$ of ${\cal X}\times \Theta$, and a finite collection of sets $\{I_k\}_{k=1}^m$ (each one a subset of $\mathbb{R}^{D^T}$), such that for $k\in[m]$ we have ${\cal I}_r(x,\theta)=I_k$ for all $(x,\theta) \in O_k$.
\end{proposition}
\begin{remark}
\label{rem1}
For some concrete examples that I consider below, I will show how to construct the partitions $\{O_i\}_{i\in [m]}$. In the two-periods binary choice model, it follows from the derivations of Section \ref{sectionDiscretization} that we can take $m=3$, and define the sets $\{O_k\}_{k\in[3]}$ as follows:
$O_1:=\{(x,\theta)\in {\cal X}\times \Theta \ |\ x_1'\theta>x_2'\theta\}$, $O_2:=\{(x,\theta)\in {\cal X}\times \Theta \ |\ x_1'\theta<x_2'\theta\}$ and $O_3:=\{(x,\theta)\in {\cal X}\times \Theta \ |\ x_1'\theta=x_2'\theta\}$. 
\end{remark}
Given a partition $\{O_k\}_{k\in [m]}$ as in Proposition \ref{propPatch}, to generate all of the conditional moment inequalities that characterize the identified set, it suffices to solve a finite number $m$ of MOLPs; one for each $O_k$.  I now state the main result of this section, which gives a characterization of the sharp identified set in term of the solutions ${\cal I}_i(x,\theta)$ of the MOLPs. Theorem \ref{thm1} is a direct consequence of Proposition \ref{propPatch} and the of the dual characterization of the DCPs.

\begin{theorem}
\label{thm1}
Fix $r\in \{E,S\}$, and let the partition $\{O_k\}_{k\in [m]}$ and the sets $\{I_k\}_{k\in [m]}$ be as in Proposition \ref{propPatch}. Then given the identified distribution $F_{Y,X}$ of observables, the corresponding sharp identified set for the index parameter $\theta$ is given by
\[
\begin{aligned}
&\Theta_r(F_{Y,X})\\
&=\{\theta \in \Theta \ | \ \forall x \in supp(X) \ \text{s.t. } (x,\theta)\in O_k, \text{ for some } k\in [m], \text{ we have } y^Tp(x)\leq 0, \ \forall y\in I_k  \},
\end{aligned}
\]
where $p(x)=F_{Y|X=x}$ denotes the identified conditional choice probability vector at $X=x$.
\end{theorem}
When $D^T$ is small (say less than 20), some known algorithms, like Benson's algorithm (see \cite{HB1}), can solve the MOLPs in \ref{eqnMolp} in reasonable time. However, as all known algorithms to solve MOLPs have worst case computational complexity that is exponential in the size of inputs and outputs, all known methods are in general not computationally feasible when $D^T$ is large (say $D^T=100$). I have reported in Tables \ref{Table1} and \ref{Table2} the running time of Benson's algorithm to solve the MOLPs in \ref{eqnMolp} for few values of $D$ and $T=2$ \footnote{For a fixed value of $(x,\theta)$, the computations consist of: creating the $A$ and $R$ matrices from Section \ref{sectionDiscretization}, and solving \ref{eqnMolp} using either Benson or the cutting-plane algorithm, and then running the redundancy elimination algorithm discussed at the end of Section \ref{sectionCuttingPlane}. The reported times represent the average running time over 10  runs with randomly drawn values of $(x,\theta)$. All computations were done on single desktop machine with 32GB of memory and a 2.4 GHz 8-Core Intel Core i9 processor. }. As can be seen from the table, as $D$ increases, the algorithm quickly becomes impractical. This suggests that one may need to exploit the specific structure of the DCPs and DDCPs to create algorithms that are tailored to solving MOLPs that arise from discrete choice model, and that work for ``larger models" where Benson's algorithm becomes impractical. In Section \ref{sectionCuttingPlane}, I have made an attempt in that direction, but much work remains;  there,  I prove that the DDCPs are \textit{integral}, and I use this observation to develop a \textit{cutting plane procedure} that works well for models of moderate size. For a comparison, I have included the running time for this new procedure in Tables \ref{Table1} and \ref{Table2}.

\begin{center}
\begin{adjustbox}{max width = \linewidth}
\begin{threeparttable}

\caption{Running Time (in seconds)}

\vspace{-1cm}\setlength\tabcolsep{10.pt} 
\renewcommand{\arraystretch}{0.8} 

\input{Runningtime1.tex}
\begin{tablenotes}
      \small
      \item Running time for model \ref{eqnstat} under Assumption \ref{assumptionS}. An asterisk (*)  is used to indicate a running time that exceeds 3600s (1hr).  
    \end{tablenotes}

\label{Table1}
\end{threeparttable}
\end{adjustbox}
\end{center}

\begin{remark}
It is important to note that if we are only interested in obtaining some (but not necessarily all) of the inequalities that characterize the identified set, then this can easily be done by solving the LPs $argmax\{w'y \ | \ y\in {\cal Q}\}$ for some arbitrary objective vectors $w>0$ (see Footnote \ref{foot13}), where ${\cal Q}$ is the DDCP. Indeed, the solutions to such LPs will be undominated extreme points of ${\cal Q}$. Thus, by randomly drawing vectors $w_k$, for $k=1,\cdots,K$ (for some large integer $K$, say $K=10^4$), from a distribution supported on $ {\mathbb{R}_{+}}^{D^T}$, and solving for $y_k=argmax\{w_k^Ty\ | \ y \in {\cal Q}\}$, the inequalities corresponding to the solutions $\{y_k\}_{k\in[K]}$ will represent a set of valid conditional moment restrictions that the model imposes on the elements of the corresponding DCP. However, this approach will likely generate redundant inequalities \footnote{ As there are only finitely many (undominated) extreme points, for large values of $K$, many of the $y_k's$ will coincide. Moreover, some of the $y_k's$ can be redundant in the sense that the corresponding inequalities will be implied by the inequalities corresponding to the $y_i's$ with $i\neq q$.}. In such cases, the redundancy elimination procedure described at the end of Section \ref{sectionCuttingPlane} can be used to generate a minimal set of equivalent inequalities. Remark \ref{remPPD3} below provides an example where I carry out this "probabilistic approach". Essentially, generating \textit{a} conditional moment inequality implied by the model is easy and computationally cheap, as it corresponds to solving for \textit{an} undominated extreme point of the DDCP, which can be framed as solving a LP. However solving for all the conditional moment inequalities that characterize the sharp identified set is difficult and computationally expensive, as it corresponds to solving for \textit{all} the undominated extreme points of the DDCP, which can be framed as solving MOLPs. 
\end{remark}
\begin{remark}
Under the exchangeability assumption, it is possible to simplify the representation of the DDCPs ${\cal Q}_E(x,\theta)$ in a way that makes the MOLPs\ref{eqnMolp} more computationally tractable. Indeed, note that by Farkas' lemma, the set $\{y \ | \ \exists z \ \text{s.t}\ A^T y\leq R_E z\}$ is equal to the set $\{y| \ \forall q\geq 0\  \text{s.t}\ R_Eq=0,\ \text{we have}\ q^TA^Ty\leq 0\}$ (see Theorem 1.1 in \cite{EB1}). Since $q\neq 0$, $q\geq 0$ and $R_Eq=0$ implies that $q$ is up to scale an exchangeable distribution, the condition $q^TA^Ty\leq 0$ for all $q\geq 0$ such that $R_Eq=0$, is equivalent to $q^TA^Ty\leq 0$ for all $q$ that are extreme points of the set of exchangeable distributions. Let $P_E$ be the matrix with each row corresponding to such an extreme point. I now discuss the construction of the matrix $P_E$. Let $F=\{f_i\}_{i=1}^{|F|}$ represent an enumeration\slash indexation of all the patches that occur in the discretization of ${\cal P}_E(x,\theta)$. For all sequences $1 \leq i_1 \leq i_2 \leq \cdots \leq i_T\leq |F|$, let the vector $q^{(i_1,\cdots,i_T)} \in \mathbb{R}^{|{\cal R}|}$ be defined by: For all $j_1,j_2,\cdots,j_T \in [|F|]$ and $R=f_{j_1}\times \cdots\times f_{j_T}\in {\cal R}$, let $q^{(i_1,\cdots,i_T)}_R=1$ iff there exists a permutation $\pi$ on $[D]$ such that $\pi(i_k)=j_k$ for all $k \in [D]$ (i.e., the vector $q^{(i_1,\cdots,i_T)}$ has an entry of one in the positions corresponding to regions $R\in {\cal R}$ whose patches are a permutation of the patches $f_{i_1},\cdots,f_{i_T}$). The matrix $P_E$ is the matrix with rows given by the vectors $q^{(i_1,\cdots,i_T)}$, with $1\leq i_1\leq i_2\leq \cdots\leq i_T\leq |F|$. It can be shown that the vectors $q^{(i_1,\cdots,i_T)}$ represent (up to scale normalization), the extreme points of all exchangeable probability distributions on the discrete set obtained by considering the T-fold product of the set $F$. When $T=2$, then the rows of $P_E$ are given by the vectors $q^{(i,i)}=\delta_{(i,i)}$ for $i\in [|F|]$, and $q^{(i,j)}=\delta_{(i,j)}+\delta_{(j,i)}$ for all $i<j$, $i,j\in [|F|]$. Here $\delta_{(i,j)}\in \mathbb{R}^{|F|^2}$ represents the vector with the component corresponding to the region $R=f_i\times f_j$ equal to 1, and all other entries equal to zero.
 From the foregoing and the definition of $P_E$, it then follows that the DDCPs ${\cal Q}_E=\{y | \ \exists z \ \text{s.t.} \ A^T y \leq R^T z, \ ||y||_{\infty}\leq 1\}$ have the equivalent representation
\begin{equation}
\label{eqnRep}
{\cal Q}_E=\{y | \ P_E A^T y\leq 0, \ ||y||_{\infty}\leq 1\}
\end{equation}
 and the latter characterization of ${\cal Q}_E$ no longer involves the auxiliary variable $z$. I have used this simplification in the implementation of both Benson and the cutting-plane algorithm that is presented in Table \ref{Table2}. A similar characterization can be given under the stationarity assumption, where the matrix $P_E$ is replaced with the matrix $P_S$ that contains all the extreme points of stationary discrete probabilities on the T-fold product of $F$. Such a characterization will however not be useful, as the matrix $P_S$ is too large for computational purposes \footnote{The number of extreme points of stationary probability distributions is exponential in $D$, for $T$ fixed; when $T=2$, $|F|\sim D^2$ and the number of extreme points of stationary probability distributions on $F\times F$ is equal to the number of cyclic permutations on $[|F|]$, which is larger than $F!$ (see the proof of Theorem \ref{thmIntE}).} 
\end{remark}

\begin{center}
\begin{adjustbox}{max width = \linewidth}
\begin{threeparttable}
\caption{Running Time (in seconds)}

\vspace{-1cm}\setlength\tabcolsep{20.pt} 
\renewcommand{\arraystretch}{0.8} 

\input{Runningtime2.tex}

\begin{tablenotes}
      \small
      \item Running time for model \ref{eqnstat} under Assumption \ref{assumptionE}. An asterisk (*)  is used to indicate a running time that exceeds 3600s (1hr). 
    \end{tablenotes}

\label{Table2}
\end{threeparttable}
\end{adjustbox}
\end{center}

Before discussing the cutting plane procedure, I first present an impossibility result. In the following section, I show show that the method of this paper is not applicable, and that there is no "simple" characterization of the identified set, if the stochastic restriction in Assumption \ref{assumptionE} or \ref{assumptionS} is replaced by an alternative stochastic restriction under which the sets ${\cal P}(x,\theta)$ are not polytopes. I show that the latter is for instance the case if we assume that the random utility shocks are IID conditional on the fixed effect and the explanatory variables.
\subsubsection{An impossibility result}
\label{sectionImp}
In this section, I try to determine when it is possible to characterize the identified set for $\theta$ using finitely many ``implications" of the form
\begin{equation}
\label{eqnImplication}
\text{If } x \ \text{and }\theta \ \text{satisfy condition} \cdots, \ \text{then } \alpha^Tp(x)\leq \beta,
\end{equation}
where the $\alpha$'s are vectors, and $\beta$'s are scalars. I refer to such a characterization as a "simple characterization" of the identified set. Note that the characterization in \ref{eqnCM}, as well as the characterizations obtained in \cite{KPT} and \cite{PP1} are of this form. In Theorem \ref{thm2} below, I demonstrate that whether or not such a characterization is possible depends in part on the geometry of the sets ${\cal P}(x,\theta)$. Notably, if the sets ${\cal P}(x,\theta)$ are not polytopes, then the sharp identified set does not admit a simple characterization, and the computational approach described in this paper is not applicable. For example, if the exchangeability or stationarity restriction on the random utility shocks is replaced by the conditional IID stochastic restriction (defined below), the resulting sets ${\cal P}(x,\theta)$ are not polytopes (Proposition \ref{propImp} below), and the identified set does not admit a simple characterization.\par

\begin{assumption}[\textbf{Independence}]
\hfill
\label{assumptionI}
\begin{enumerate}[a)]
\item Conditional on $x_i$ and $\lambda_i$, for each $t\in[T]$, $\epsilon_{ti}$ $(\in \mathbb{R}^D)$ is continuously distributed. 
\item Conditional on $x_i$ and $\lambda_i$, the shocks $\epsilon_{ti} $ are independent and identically distributed (as $t$ varies in $[T]$).
\end{enumerate}
\end{assumption}
Under Assumption \ref{assumptionI}, we can define the analogue ${\cal P}_I(x,\theta)$ of the sets ${\cal P}_S(x,\theta)$ and ${\cal P}_E(x,\theta)$. That is, the set ${\cal P}_I(x,\theta)$ is the set of all CCPs that are consistent with the discrete choice model \ref{eqnstat} under the stochastic restriction \ref{assumptionI}, when the index parameter takes the value $\theta$ and the covariates take the values $X=x$. The following proposition shows that unlike the sets ${\cal P}_S(x,\theta)$ and ${\cal P}_E(x,\theta)$, the sets ${\cal P}_I(x,\theta)$ are not polytopes.
\begin{proposition}
\label{propImp}
The sets ${\cal P}_I(x,\theta)$ are closed and convex, but are not polytopes.
\end{proposition}

Theorem \ref{thm2} below shows that when the sets ${\cal P}(x,\theta)$ are not polytopes, the sharp identified set for $\theta$ does not admit a simple characterization as in Theorem \ref{thm1}. Formally, the sharp identified set admits a simple characterization if for some positive integer $M$ there exists a finite collection of subsets $\{S_k\}_{k\in [M]}$ of ${\cal X}\times\Theta$ and a collection of "associated inequalities" $\{(\alpha_k,\beta_k)\}_{k\in[M]}$ such that $\cup_{k\in[M]}S_k={\cal X}\times \Theta$ (here the sets $S_k$ are allowed to have non-empty intersection, and can even be equal for different values of the subscript $k$),  and given any distribution of the observables $F_{Y,X}$, the identified set for the parameter $\theta$ is given by
\begin{equation}
\label{eqn_2}
\Theta(F_{Y,X})=\{\theta \in \Theta \ |\ \forall x\in supp(X) \ s.t. \ (x,\theta) \in S_k, \ \text{for some } k\in[M],\ \text{we have } \alpha_k^Tp(x)\leq \beta_k \}.
\end{equation}
Note that the representation \ref{eqn_2} essentially implies that for each $k\in[M]$ we have  
\begin{equation}
\label{eqn_1}
 \alpha_k^Tp\leq \beta_k\quad \forall \ p\in {\cal P}(x,\theta) \quad \text{such that } (x,\theta) \in S_k ,
\end{equation}
and the latter $M$ implications are sufficient to characterize the identified set.


\begin{remark}
\label{rem01}
The characterization of the identified set given in \cite{CM}, \cite{PP1} and \cite{KPT}, all have the form \ref{eqn_2}. From Equation \ref{eqnCM}, in the setting of \cite{CM} (for instance) the identified set can be characterized by ($M=$) 2 implications, where the sets $S_k=O_k\cup O_3$, for $k\in [2]$, with the sets $O_k$ defined as in Remark \ref{rem1},  and where the associated inequalities are given by $\alpha_1=(0,1,-1,0)'$, $\alpha_2=(0,-1,1,0)$, and $\beta_1=\beta_2=0$\footnote{As in Section \ref{sectionHeuristics}, given $\alpha=(a_1,a_2,a_3,a_4)'$ and a CCP vector $p$, $\alpha^Tp=a_1p_{00}+a_2p_{01}+a_3p_{10}+a_4p_{11}$ where $p_{dd'}=P(Y_1=d,Y_2=d'|x)$.}. The characterization of the identified set in Theorem \ref{thm1} is also of the form \ref{eqn_2}, where $M:=\sum_{j\in[m]}|I_j|$ (where $|I_j|$ is the number of elements in the set $I_j$), the sets $\{S_k\}_{k\in[M]}$ consist of $|I_j|$ copies of the set $O_j$ for each $j\in[m]$, all the $\beta_k$'s are equal to 0, and each $\alpha_k$ corresponds to an elements of $I_j$ if $S_k=O_j$.
\end{remark}

\begin{theorem}
\label{thm2}
Suppose that for some $\theta_0\in \Theta$ and  $x_0\in {\cal X}$, the set of CCPs that are consistent with the model at $x_0$ when the parameter value is $\theta_0$, ${\cal P}(x_0,\theta_0)$, is not a polytope. Then the sharp identified set for $\theta$ cannot be characterized by finitely many linear inequality restrictions as in \ref{eqn_2}.
\end{theorem}
 A direct corollary of Proposition \ref{propImp} and Theorem \ref{thm2} is that at least infinitely many linear inequality restrictions are needed to characterize the sharp identified set for $\theta$ in model \ref{eqnstat} under the (conditional) IID stochastic restriction \ref{assumptionI}, and this is the case even if the vector of covariates $X$ is degenerate and takes a single value $x_0$.\footnote{Note that although the identified set cannot be characterized by a finite number of \emph{linear} conditional moment inequality restrictions, it may still be possible to characterize it using a finite number of \emph{nonlinear} conditional moment inequalities. See \cite{DGK} for an example of a setting where the sets ${\cal P}$ are not polytopes, and where a characterization of the identified set is provided that uses a finite set of \emph{nonlinear} conditional moment inequality restrictions.}
\begin{remark} In the setting of Section \ref{sectionHeuristics}, under the alternative stochastic restriction \ref{assumptionI}, although a characterization of the identified set of the form \ref{eqn_2} is not possible, a simple extension of the approach of this paper can yield the following representation for the identified set: 
\[
\begin{split}
\Theta_I(F_{Y,X})= {}  \{& \theta \in \Theta \ | \ \forall x\in supp(X) \ s.t.  \ (x,\theta) \in O_k, \ \text{for } k\in[3],\ \text{we have } w^Tp(x)\leq \mu^{(k)}(w)\\
& \forall \ w \text{ s.t. } \|w\|_{\infty}= 1 \}
\end{split}
\]
where the sets $\{O_k\}_{k\in[3]}$ are as in Remark \ref{rem1}, and the functions $\mu^{(k)}$, which can be computed explicitly (see the proof of Proposition \ref{propImp} in Appendix A) represent the (common) support functions of the sets ${\cal P}_I(x,\theta)$ for values of $(x,\theta)$ in $O_k$. Hence the implications that characterize the identified set are now of the type:
\[
\text{If } x \ \text{and }\theta \ \text{are in } O_k, \ \text{then }w^Tp(x)\leq \mu^{(k)}(w)\ \forall \ w \text{ s.t. } \|w\|_{\infty}= 1.
\]
\end{remark}

\subsection{Cutting-plane algorithm}
\label{sectionCuttingPlane}
In this section, I propose a \textit{cutting-plane} algorithm, as an alternative to Benson's algorithm, to solve the MOLPs in \ref{eqnMolp}. The algorithm relies on an integrality result. In particular, I show in Theorem \ref{thmIntE} below that the extreme points of the DDCPs are integral (i.e., have integral coordinates), and all have equal and maximal rank (see Definition \ref{defRank}). The result is only established for the setting where $T=2$, but extensive computations suggest that the result is still valid for $T>2$. For comparison, the running time of the cutting-plane algorithm for few cases is presented in tables \ref{Table1} and \ref{Table2}. Note however, that whereas the cutting-plane algorithm that is introduced below is only formally justified to solve the MOLPs that arise in static models when $T=2$, Benson's algorithm is always valid in all settings (although it will not be practical whenever $D^T$ is large). \par
Before stating the main results of this section, I introduce the notion of rank.
\begin{definition}
\label{defRank}
The rank of a vector $y \in \mathbb{R}^n$ is defined as the sum of its components, i.e, 
\[
rank(y)=\mathbbm{1}^Ty.
\]
\end{definition}
The following theorem gives a partial characterization of the undominated extreme points of the DDCPs.
\begin{theorem}
\label{thmIntE}
Let $T=2$, $D\geq 2$, and consider the static model \ref{eqnstat} under either the exchangeability assumption \ref{assumptionE}, or the stationarity assumption \ref{assumptionS}. Then the DDCPs are integral, i.e, their extreme points are in $\{0,\pm 1\}^{D^T}$, and all of their undominated extreme points have the same rank, equal to the maximal rank.
\end{theorem}


\begin{remark}
\label{rem4}
 Extensive simulations seem to suggest that the conclusion of Theorem \ref{thmIntE} also holds when $T\geq 2$. When $T=2$, I prove the results by linking the investigation of extreme points of the DDCPs to the investigation of flows on networks, and the techniques developed in the literature on network flows, to study integral flows, can be used to obtain the results stated in Theorem \ref{thmIntE}. This approach (study of DDCPs through network flows) does not seem to generalize to settings where $T>2$. 
\end{remark}
As the cutting plane algorithm relies on the integrality of the DDCP, although our results are not applicable when $T\geq 3$, the following characterization of integrality can be used to assess through a numerical simulation whether or not the DDCP is integral. The result is due to \cite{EG1} (see also \cite{AS}, pg. 74).
\begin{proposition}
\label{propEG}
A (rational) polyhedron $P \subseteq \mathbb{R}^n$ is integral if and only if for each $c\in \mathbb{Z}^n$ the value of $\max\{c^Tx| x\in P\}$ is an integer if it is finite.\footnote{\label{foot14}Here a polyhedron $P:=\{x| \ Ax\leq b\}$ is rational if the matrix $A$ and the vector $b$ have rational components. The DCPs and the DDCPs are rational as the matrices and vectors that appear in their definitions are integral (i.e, the entries are integers). Moreover, for the DDCPs, the condition that the value of $\max\{c^Tx| x\in P\}$ is finite holds, as the DDCPs are compact sets.}
\end{proposition}
\begin{remark}
The proposition implies that the integrality of a polytope ${\cal Q}$ (the DDCP in our case) can be assessed using the following numerical simulation: Let $\sigma>0$ be large (say $\sigma=100$), and independently draw a large number K (say $K=10^4$) of vectors $\{w_k\}_{k\in[K]}$, where $w_k\sim N(0,\sigma I_n)$, and let $c_k=\lfloor  w_k \rfloor$ (component-wise), i.e, $c_k$ is the largest integral vector that is less than $w_k$ component-wise. Compute $V_k$, the value of the LP $\max\{c_k^Tx |  \ x \in {\cal Q}\}$. If one of the values $\{V_k\}_{k\in[K]}$ is not an integer, then ${\cal Q}$ is not integral, and the cutting-plane procedure is not adequate (in this case ${\cal Q}$ has some non-integral extreme points). If on the other hand all of the $V_k's$ are integers , then it is suggestive evidence that  ${\cal Q}$ is integral, and that the cutting-plane procedure is likely valid.
\end{remark} 
I now describe the cutting-plane algorithm. Let the polytopes ${\cal Q}$ under consideration be the DDCPs, i.e., ${\cal Q}=\{y | \ \exists z \ \text{s.t.} \ A^T y \leq R^T z, \ ||y||_{\infty}\leq 1\}$, where $R\in\{R_S,R_E\}$. The cutting-plane algorithm that I propose works by iteratively locating integral points in ${\cal Q}$, say $\tilde{y}$, and by appending an inequality $\alpha^Ty \leq \beta$ to ${\cal Q}$, called a  \textit{cut}, such that $\tilde{y}$ does not satisfy the inequality, i.e, $\alpha^T\tilde{y}>\beta$, but any other integral point in ${\cal Q}$ does. That is the cut $\alpha^Ty \leq \beta$ excludes a unique integral point from ${\cal Q}$, the point $\tilde{y}$.  I iterate this procedure (i.e, iteratively append more cuts to the polytope) until there are no more integral points of maximal rank in ${\cal Q}$. In practice, although the number of integral vectors in $\{0,\pm 1\}^{D^T}$ of a given rank can be very large, the number of maximal rank integral vectors in ${\cal Q}$ is usually much smaller (for instance, less than 50 if $D=5$ and $T=2$), and the algorithm terminates in reasonable time for problems of moderate size\footnote{\label{foot15}Cutting-plane methods are a technique used to solve integer/mixed-integer linear programs (see \cite{AS} pg. 84). }.\par
To implement the cutting plane algorithm, consider the auxiliary polytope ${\cal Q}'$ defined by
\[
\{w=(u,v)|\ u, v\in \mathbb{R}^{D^T}, \ u,v\geq 0, \ u+v\leq \mathbbm{1},\ \exists z \  \text{s.t} \ A^Tu-A^Tv\leq R^T z\}.
\]
For $y\in \mathbb{R}^{D^T}$, let $y^+$ and $y^-$, both in $\mathbb{R}^{D^T}$, denote respectively, the positive and negative part of $y$, i.e., $y^+=max\{y,0\}$ (component-wise) and $y^{-}=max\{-y,0\}$. It can be shown that for $y\in {\cal Q}$, $(y^+,y^-)\in {\cal Q}'$, and that if $y$ is integral then so is $(y^+,y^-)$. Moreover for each $w=(u,v) \in {\cal Q}'$, the vector $y=u-v$ belongs to ${\cal Q}$, and $y$ is integral if $w$ is integral. Thus if $r:=max\{\mathbbm{1}^Ty | \ y\in {\cal Q}\}$ is the maximal rank in ${\cal Q}$, we have:
\[
r=max\{\mathbbm{1}^Ty | \ y\in {\cal Q}, \text{and y is integral}\}=max\{\mathbbm{1}^Tu-\mathbbm{1}^Tv | \ (u,v) \in {\cal Q}' \text{and (u,v) is integral}\}.
\]
where, for the first equality, we have used  the fact that there exist integral maximizers of the rank function in ${\cal Q}$ (the undominated extreme points of ${\cal Q}$), which follows from Theorem \ref{thmIntE}. From the foregoing, it follows that we can recover all integral points of ${\cal Q}$ of maximum rank from all the integral points of ${\cal Q}'$ that maximize the function $r'$ on ${\cal Q}'$, where $r'$ is defined by 
\[
r'(u,v):=\mathbbm{1}^Tu-\mathbbm{1}^Tv.
\]
If $(u,v)\in {\cal Q}'$ is an integral point that maximizes $r'$, then $y=u-v$ is an integral point in ${\cal Q}$ of maximum rank; and if $y\in {\cal Q}$ is an integral point of maximum rank, then $(y^+,y^-)$ is in ${\cal Q}'$ and maximizes $r'$. Thus to recover all integral extreme points of ${\cal Q}$ of maximum rank, it suffices to recover all integral points of ${\cal Q}'$ that maximize the function $r'$. Now, since integral points in ${\cal Q}'$ are binary vectors (i.e, all entries are in $\{0,1\}$), if $\tilde{w}$ is an integral vector in ${\cal Q}'$ that maximizes the function $r'$, then the inequality 
\begin{equation}
\label{eqnCut}
(\mathbbm{1}-\tilde{w})^T(\mathbbm{1}-w)+\tilde{w}^Tw \leq 2 D^T-1
\end{equation}
is satisfied by all other integral points in $w \in {\cal Q}'$, but not satisfied by $\tilde{w}$; inequalities of this type will be used as our cuts.\par
All of the integral points that maximize $r'$ on ${\cal Q}'$ can be recovered as follows: First find an integral point $y_1\in {\cal Q}$ of maximal rank  \footnote{\label{foot16}Using the simplex algorithm for instance will return a corner solution (i.e, extreme point solution) which will be integral by Theorem \ref{thmIntE}.}, and let $w_1=(y_1^+,y_1^-)$. Then iteratively solve, for $k\geq2$, the following mixed-integer linear program
\[
\begin{aligned}
w_k=argintmax {} & \{0\ | \ w\in {\cal Q}', (\mathbbm{1}-w_s)^T(\mathbbm{1}-w)+w_s^Tw \leq(2 D^T-1), \\
& \text{for all } s\in [k-1], \text{ and}\  r'(w)=r\}
\end{aligned}
\]

until infeasibility \footnote{\label{foot17}Although ${\cal Q}$ is integral, ${\cal Q}'$ is not necessarily integral. Moreover, the cuts that we append to ${\cal Q}'$ may introduce non-integral corners. That is why a mixed-integer linear program is used from the second step onward, which is unfortunately computationally costly. Here the program is a mixed-integer linear program, since $w$ is required to be integral, but not necessarily $z$ (see the definition of ${\cal Q}'$).}. Let $M+1$ be the last iteration when the program terminates (i.e., the first iteration when the mixed-integer linear program becomes infeasible), and for each $k\in [M]$ define $y_k=u_k-v_k$, where $w_k=(u_k,v_k)$. Then the set $\{y_k\}_{k\in[M]}$ represents the set of all the integral points in ${\cal Q}$ of maximum rank.\par

\subsubsection{Redundancy elimination}
Since the algorithm recovers all the integral points in ${\cal Q}$ of maximum rank, and not just the extreme points, some of the vectors $y_k$ may be redundant, in the sense that the inequality $y_k^Tp\leq 0$ is implied by $y_s^Tp\leq 0$ for all $s\in [M]\backslash\{k\}$.  To get rid of the redundant inequalities, and thus obtain a more compact summary of the information in ${\cal Q}$, we use Farkas' Lemma, which yields that $y_t$ is redundant if and only if there exists $\lambda\geq 0$, $\lambda \in \mathbb{R}^{M-1}$ such that 
\[y_k\leq Y_{-k}\lambda,
\]
 where $Y_{-k}$ is the matrix whose columns are given by the vectors $y_s$, $s\in [M]\backslash\{k\}$. Hence, a simple algorithm to get rid of redundancies goes as follows: Let $L_1:=\{y_k\}_{k\in [M]}$ be the initial list; starting from $k=1$ until $k=M$, remove the vector $y_k$ from the list, and set $L_{k+1}=L_k \backslash\{y_k\}$, whenever 
\[
0=min\{ 0 | \ y_k \leq L_{k,-k} \lambda, \ \lambda\geq 0 \}\  
\]
where $L_{k,-k}$ is the matrix whose columns are all the vectors in the list $L_k$ except $y_k$; let $L_{k+1}=L_k$ otherwise. It is not difficult to show that no vector in the final list $L_{M+1}$ is redundant, and every vector in $L_1\backslash L_{M+1}$ is redundant given the vectors in $L_{M+1}$. Thus the vectors in $L_{M+1}$ are a compact summary of the information in $L_1$, and the list $L_{M+1}$ is the final output of the cutting-plane procedure.

\subsection{Examples of applications of the algorithm in static settings}
\label{sectionExample}
In this section, I implement the foregoing procedure on two examples. The first example considers the model in \cite{PP1} (where $T=2$) with $D=4$ alternatives. Here I show that the algorithm recovers the inequalities derived in \cite{PP1}. In the second example, I consider the same setting of \cite{PP1} (with $D=4$ alternatives), where I replace their assumption of stationarity with the exchangeability assumption, and  I obtain some new inequalities. I then proceed to prove these inequalities analytically, and through that process, I guess their generalization to a setting with an arbitrary number $D$ of alternatives. When $D^T$ is large, it is recommended to use the latter approach, as the method discussed in this paper may be impractical. In such cases, one can first uses the algorithm to generate the inequalities that characterize the identified set for small values of $D$. Once one obtains these inequalities, one can then try to prove them analytically. From the proof one can then guess and prove their generalization to a setting with arbitrarily large $D$. 
\begin{example}[\textbf{\cite{PP1}}]
\label{exampPPS}
This example considers the setting of \cite{PP1}, i.e., model \ref{eqnstat} under Assumption \ref{assumptionS}, with $T=2$. I consider a setting with 4 alternatives (D=4). Let $v_1=(v_{11},\cdots,v_{41})=(0,0,0,0)$ and $v_2=(v_{12},\cdots,v_{42})=(4,3,2,1)$ represent fixed values of the deterministic component of the indirect utilities for period 1 and 2 respectively (see the second paragraph of Section \ref{sectionDiscretization}). Running our algorithm with $v_1$ and $v_2$ as inputs \footnote{\label{foot18}This involves creating the A and $R_S$ matrices, and then running the cutting-plane algorithm of Section \ref{sectionCuttingPlane} (or running Benson's algorithm to solve the corresponding MOLP) on the DDCPs (constructed using the aforementioned matrices A and $R_S$), and then running the redundancy elimination algorithm described at the end of Section \ref{sectionCuttingPlane}.}, yields the following list of vectors (one vector per row):
\[
\bordermatrix{& p_{11}&p_{12}&p_{13}&p_{14}&p_{21}&p_{22}&p_{23}&p_{24}&p_{31}&p_{32}&p_{33}&p_{34}&p_{41}&p_{42}&p_{43}&p_{44}\cr
&0&1&1&1&-1&0&0&0&-1&0&0&0&-1&0&0&0\cr
&0&0&1&1&0&0&1&1&-1&-1&0&0&-1&-1&0&0\cr
&0&0&0&1&0&0&0&1&0&0&0&1&-1&-1&-1&0
}\qquad
\]
The first row corresponds to the inequality
\[
p_{12}+p_{13}+p_{14} \leq p_{21}+p_{31}+p_{41}
\]
and adding $p_{11}$ to both sides yields the inequality
\[
P(y_1=1|x) \leq P(y_2=1|x).
\]
The second row corresponds to the inequality
\[
p_{13}+p_{14}+p_{23}+p_{24}\leq p_{31}+p_{32}+p_{41}+p_{42}
\]
and adding $p_{11}+p_{12}+p_{21}+p_{22}$ to both sides yields 
\[
P(y_1\in\{1,2\}|x)\leq P(y_2 \in \{1,2\}|x).
\]
The third row corresponds to the inequality
\[
p_{14}+p_{24}+p_{34}\leq p_{41}+p_{42}+p_{43} 
\]
and adding $p_{44}$ to both sides yields the inequality 
\[
P(y_1\in \{1,2,3\}|x)\leq P(y_2\in \{1,2,3\}|x).
\]
The inequalities that we have derived are completely determined by the matrices $A$ and $R_S$, which in turn are completely determined by the rankings of the index function differences ${\Delta v_{d}}{d\in [4]}$, where $\Delta v_d:=v{d2}-v_{d1}$ (see Remark \ref{rem3}). Thus, these inequalities are always valid whenever the index function differences have the same ranking as in our example. Specifically, the inequalities are valid whenever $\Delta v_1>\Delta v_2>\Delta v_3> \Delta v_4$ holds (as is the case for $v_1$ and $v_2$ with the values given above), or equivalently, whenever $\Delta \theta^Tx_1>\Delta\theta^Tx_2>\Delta\theta^Tx_3> \Delta\theta^Tx_4$ holds, where $\Delta \theta^Tx_d:= \theta^Tx_{d2}-\theta^Tx_{d1}$. Furthermore, as we can freely relabel the alternatives, we have shown that
\begin{equation}
\label{eqnPP}
P(y_1 \in U_{i}|x)\leq P(y_2 \in U_i|x)
\end{equation}
for $i\in[3]$, where $U_i$ represents the set of indices with i largest values of index function differences. Equivalently, in the context of Theorem \ref{thm1}, we have shown that the inequalities in \ref{eqnPP} are valid for all CCP vector at $(x,\theta)$, whenever $(x,\theta)\in O$, where the set $O$ is the subset of ${\cal X}\times \Theta$ defined by
 \[
O=\{(x,\theta) \in {\cal X}\times \Theta \ | \ \theta^Tx_{d2}-\theta^Tx_{d1} \neq \theta^Tx_{d'2}-\theta^Tx_{d'1}, \ \forall d\neq d', \ d,d'\in[4] \}.
\]
As our algorithm recovers all the restrictions that the model puts on the set ${\cal P}_S(x,\theta)$, the inequalities in \ref{eqnPP} represent all the restrictions that the model puts on the observable CCPs, whenever the index function differences are distinct, and $D=4$. The same conclusion (when restricted to $D=4$) was derived in \cite{PP1} by analytical means.
\begin{remark}
\label{rem5}
As can be seen from Theorem \ref{thm1}, the inequalities that characterize the identified set only need to be computed for a finite number of cases (equal to m in Theorem \ref{thm1}). In the setting of the preceding example ($D=4$ and $T=2$), only 8 cases need to be considered (i.e., we only need to solve 8 MOLPs). Each such case corresponds to one way in which ties can arise among the ranked index function differences. In particular, assume w.l.o.g. that $v_1=(0,0,0,0)$ in all cases. \footnote{\label{foot19}Since we are conditioning on $x$, and there is no location normalization on the shocks, the period 1 deterministic component of the indirect utilities can be absorbed into the shocks, making $v_1=0$, and the period 2 deterministic components of the utilities can be redefined to be $v_2=\Delta v$.} Case 1 can be the one above, i.e.,  $v_2=(4,3,2,1)$; for case 2 use $v_2=(4,3,2,2)$; for case 3 use $v_2=(4,3,3,3)$; for case 4 use $v_2=(4,3,3,3)$; for case 5 use $v_2=(4,4,3,2)$; for case 6 use $v_2=(4,4,3,3)$; for case 7 use $v_2=(4,4,4,3)$; and for case 8 use $v_2=(4,4,4,4)$. Here case 1 corresponds to the case where there is no tie in the index function differences across the 4 alternatives, and case 2 through 8 are all possible ways that the ties among the ordered index function differences can arise. Running the algorithm in these 8 cases and arguing by symmetry as above (i.e., we can freely relabel the alternatives) will yield all the inequalities restrictions for any possible values of $v_1$ and $v_2$.\footnote{It is important to note that this step can be parallelized, and all 8 MOLPs can be solved independently and simultaneously. Also, since we can also freely relabel the time periods, the inequalities for case 5 can be obtained from those of case 2, and the inequalities for case 7 can be obtained from those of case 4. Thus we really only need to solve 6 MOLPs.}\par
 One of the drawbacks of the procedure is that the identifying inequalities that the algorithm generates are case specific, i.e, the conclusion does not generalize to other models. For instance, after running the procedure on a model with $D=4$ alternatives (and considering all 8 cases mentioned above), we have to re-run the procedure again if we now want to include an extra alternative in our model (i.e., $D=5$), in order to generate the identifying inequalities for a setting with $D=5$. If one wants to generalize from the inequalities obtained for $D=4$ to other values of $D$, an intermediate step is required where one will need to prove analytically the inequalities generated by the algorithm for $D=4$,  and in the process guess their generalizations. I do this in the next example.

\end{remark}

\end{example}
\begin{example}[\textbf{\cite{PP1} with exchangeability}]
\label{exampPPE}
The setting here is the same as that of the preceding example, with the only difference being that I replace the stationarity assumption considered by \cite{PP1} with the exchangeability assumption. Running the algorithm with the same values of $v_1$ and $v_2$ as in Example \ref{exampPPS} yields the following 13 inequalities:
\begin{enumerate}
\item $p_{12}+p_{13}+p_{14}+p_{23}+p_{24}+p_{34}\leq p_{21}+p_{31}+p_{32}+p_{41}+p_{42}+p_{43}$
\item $p_{12}+p_{13}+p_{14}+p_{23}+p_{24}\leq p_{21}+p_{31}+p_{32}+p_{41}+p_{42}$
\item $p_{12}+p_{13}+p_{14}+p_{24}+p_{34}\leq p_{21}+p_{31}+p_{41}+p_{42}+p_{43}$
\item $p_{12}+p_{13}+p_{14}+p_{24}\leq p_{21}+p_{31}+p_{41}+p_{42}$
\item $p_{12}+p_{13}+p_{14}\leq p_{21}+p_{31}+p_{41}$
\item $p_{13}+p_{14}+p_{23}+p_{24}+p_{34}\leq p_{31}+p_{32}+p_{41}+p_{42}+p_{43}$
\item $p_{13}+p_{14}+p_{23}+p_{24}\leq p_{31}+p_{32}+p_{41}+p_{42}$
\item $p_{13}+p_{14}+p_{24}+p_{34}\leq p_{31}+p_{41}+p_{42}+p_{43}$
\item $p_{13}+p_{14}+p_{24}\leq p_{31}+p_{41}+p_{42}$ 
\item $p_{13}+p_{14}\leq p_{31}+p_{41}$
\item $p_{14}+p_{24}+p_{34}\leq p_{41}+p_{42}+p_{43}$ 
\item $p_{14}+p_{24}\leq p_{41}+p_{42}$
\item $p_{14}\leq p_{41}$.
\end{enumerate}
By construction these inequalities represent the set of all restrictions that the model places on the observable CCPs if the index function differences satisfy the relations $\Delta v_1>\Delta v_2>\Delta v_3>\Delta v_4$. Moreover, each inequality here is non-redundant, in the sense that each inequality is not implied by the others \footnote{\label{foot21}See the redundancy elimination algorithm at the end of Section \ref{sectionCuttingPlane}}. Inequalities 5,7 and 11 are \cite{PP1} inequalities (there are the same inequalities that were obtained in Example \ref{exampPPS}), and the remaining inequalities are of a new kind. Thus, as expected, the stronger assumption of exchangeability translates into more restrictions on the observable CCPs than when stationarity is assumed. An attempt to prove these inequalities analytically suggests their natural generalization to an arbitrary number $D$ of alternatives, and I provide this generalization in Theorem \ref{thmE} below. The inequalities in Theorem \ref{thmE} can then be used in settings with large $D$, where our algorithm is not computationally feasible.
\end{example}
Let $D\geq 2$, and let $(x,\theta) \in {\cal X}\times \Theta$ be fixed. For $d\in[D]$ and $t\in[2]$, let $v_{dt}:=\theta^Tx_{dt}$ and $\Delta v_d:=v_{d2}-v_{d1}$. Let $d^{(i)}$, with $i\in[D]$, represent the alternative with the $i^{th}$ largest value of $\Delta v_d$: i.e., $\Delta v_{d^{(1)}}\geq \Delta v_{d^{(2)}}\geq \cdots\geq \Delta v_{d^{(D)}}$ (with an arbitrary ranking among ties). For $i,j \in [D]$, let the subsets of alternatives $U_i$ and $L_{j}$ be defined by $U_{i}:=\{d^{(1)},\cdots,d^{(i)}\}$ and $L_{j}=\{d^{(j)},\cdots, d^{(D)}\}$ ($U_i$ denotes the set of all the $i$ alternatives with index function differences of largest rank, and $L_{j}$ denotes the set of the $D-j+1$ alternatives with index function difference of lowest rank). Finally, let the family of sets ${\cal A}(x,\theta)$ (subsets of $[D]\times[D]$) be defined by 
\begin{equation}
\label{eqnA}
{\cal A}(x,\theta):=\left\{\bigcup_{i=1}^m U_{k_i}\times L_{k'_i}|\ m\in \mathbb{N}, \ k_i,k'_i \in [D] \right\}=\left\{\bigcup_{i=1}^m U_{k_i}\times L_{k'_i}|\ 1\leq m\leq D, \ k_i,k'_i \in [D]\right\}
\end{equation}
where the second equality follows since it can be shown that an arbitrary union of sets of the type $U_{k_i}\times L_{k'_i}$ can always be written as a union of up to $D$ sets of the same type\footnote{If we assume w.l.o.g. that $\Delta v_1\geq \Delta v_2\geq \cdots\geq \Delta v_D$ (which can be accomplished by relabeling the alternatives), then ${\cal A}$ can be identified with the non-decreasing sequences on $[D]$ of length at most $D$: $\forall A\in {\cal A}$, $\exists m \in[D]$ and a non-decreasing sequence $1\leq i_1\leq \cdots\leq i_m\leq D$, such that $A=\{(d,d') \in [m]\times[D]\ | \ 1\leq d\leq m \ \text{and}\ d'\geq i_d\}$ (this can be seen by plotting the sets $A$ on a D-by-D grid). Moreover, every set of the latter type is in ${\cal A}$, and we get $|{\cal A}|={2D \choose D}-1\sim 4^D/\sqrt{D}$.}. Given a set $A\in {\cal A}(x,\theta)$, with $A=\bigcup_{i=1}^m U_{k_i}\times L_{k'_i}$, define the associated set $per(A)$ by
\begin{equation}
\label{eqnPer}
per(A):=\bigcup_{i=1}^m L_{k'_i}\times U_{k_i}.
\end{equation}
\footnote{It can be easily shown that $per(A)$ is well defined and does not depend on the particular representation of $A$.}The following theorem generalizes the inequalities produced by our algorithm.

\begin{theorem}
\label{thmE}
Let $D\geq 2$ and suppose that Assumption \ref{assumptionE} holds. Let $\theta$ be the true value of the index parameter. Then for all $x\in supp(X)$ and for all $A\in {\cal A}(x,\theta)$ we have:
 
\begin{equation}
P((Y_1,Y_2)\in A|x)\leq P((Y_1,Y_2)\in per(A)|x)
\end{equation}
where ${\cal A}(x,\theta)$ is as defined in Equation \ref{eqnA}.

\end{theorem}
\begin{remark}
\label{rem6}  The inequalities of \cite{PP1} correspond to the sets $A=\bigcup_{i=1}^m U_{k_i}\times L_{k'_i}$ in ${\cal A}(x,\theta)$, with $m=1$ and $k'_1=1$. Under exchangeability, we get many more inequalities ($|{\cal A}|={2D \choose D}-1$). However, some of our inequalities may be redundant. This can be seen by noting that when $D=4$, even though  Theorem \ref{thmE} considers $|{\cal A}|={8 \choose 4}-1=69$ inequalities, our algorithm only reports 13 such inequalities.\footnote{ Recall that our algorithm deletes all redundant inequalities by construction (see the end of Section \ref{sectionCuttingPlane})} Although the procedure of this paper is impractical for large values of $D$, the set of inequalities derived in Theorem \ref{thmE} is likely to characterize the sharp identified set for all values of D.\footnote{This was verified computationally for up to 8 alternatives.}
\end{remark}
\begin{remark}
\label{rem8}
It is important to note that the inequalities of Theorem \ref{thmE}  (and more generally, the inequalities generated by our procedure) are useful even if one is interested in point identification; in that case, one can try to determine additional (and preferably minimal) assumptions that when combined with the derived inequalities restrictions, yield point identification, and the inequalities can then be used to construct maximum score like estimators of the parameter $\theta$ (see for instance \cite{CM}, \cite{KPT} and \cite{SSS}). 
\end{remark}

\section{Dynamic models}
\label{sectionDyn}
In this section, I show how the procedure that is discussed in Section \ref{sectionStatic} can be extended to dynamic settings. In particular consider the following simple dynamic multinomial choice model:\footnote{Note that the type of dynamics that we consider in this paper are "non-structural", as we are not necessarily assuming that the agents making the choices are forward-looking.} Agent $i$ chooses alternative $d\in[D]$ at time $t\in[T]$, which I represent by $Y_{ti}=d$, if and only if
\begin{equation}
\label{eqndyn}
  X_{dti}'\beta+\gamma \mathbbm{1}_{[Y_{(t-1)i}=d]}+\lambda_{di}+\epsilon_{dti} > X_{d'ti}'\beta+\gamma \mathbbm{1}_{[Y_{(t-1)i}=d']}+\lambda_{d'i}+\epsilon_{d'ti} \quad  \forall d' \in [D],\ d'\neq d.
\end{equation}
Here, the data consists of a random sample of observed choices and covariates of individuals across $T$ time periods, as well as $Y_0$, which represents the individuals' choices in period $0$ (i.e., the initial condition). In other words, we have access to a random sample of the variables $(Y,X)$, where $Y=(Y_0,Y_1,\cdots,Y_T)'\in [D]^{T+1}$ and $X=(X_1,\cdots,X_T)\in \mathbb{R}^{d_x\times DT}$ (see the notation following equation \ref{eqnstat}). Having a specification of the indirect utility that includes both fixed effects and lagged choices is motivated by the desire to distinguish the impact of unobserved heterogeneity from that of state dependence (see \cite{JJH1,JJH2, JJH3, JJH4}). 
The construction presented below easily extends to dynamic models that consider more than one lag and where the lag/state dependence parameters are alternative dependent (see Example \ref{exampDyn2} below for a setting with two lags). I chose to use a simple setting with one lag primarily for the sake of expositional simplicity \par
Following \cite{KPT}, I study the identification of the common parameters $\theta=(\beta,\gamma)$ While imposing only mild stochastic restrictions on the distribution of both observed and unobserved variables in our model. Specifically, I consider two types of restrictions: the first type imposes stationarity or exchangeability of the shocks conditional on the covariates and fixed effects, but not on the initial conditions (as assumed in assumptions \ref{assumptionS} and \ref{assumptionE}). The second, stronger type of restriction, which I formally define next, imposes stationarity or exchangeability of the shocks conditional on the initial condition, the covariates, and the fixed effects.
\begin{assumption}[\textbf{Conditional Stationarity}]
\hfill
\label{assumptionCS}
\begin{enumerate}[a)]
\item Conditional on $Y_{0i}$, $X_i$ and $\lambda_i$, for each $t\in[T]$, $\epsilon_{ti}$ is continuously distributed. 
\item Conditional on $Y_{0i}$, $X_i$ and $\lambda_i$, the shocks $\epsilon_{ti}$ have the same distribution:
\[
\epsilon_{ti}|Y_{0i},X_i,\lambda_i \sim \epsilon_{si}|Y_{0i},X_i,\lambda_i \ \ \ \forall \ \ s,t \in [T] 
\]
\end{enumerate}
\end{assumption}

\begin{assumption}[\textbf{Conditional Exchangeability}]
\hfill
\label{assumptionCE}
\begin{enumerate}[a)]
\item Conditional on $Y_{0i}$, $X_i$ and $\lambda_i$, for each $t\in[T]$, $\epsilon_{ti}$ is continuously distributed. 
\item Conditional on $Y_{0i}$, $X_i$ and $\lambda_i$, the joint distribution of the shocks is exchangeable:
\[
(\epsilon_{1i},\epsilon_{2i},\cdots,\epsilon_{Ti})|Y_{0i},X_i,\lambda_i\sim (\epsilon_{\pi(1)i},\epsilon_{\pi(2)i},\cdots,\epsilon_{\pi(T)i})|Y_{0i},X_i,\lambda_i
\]
for all permutations $\pi$ on the set $[T]$.
\end{enumerate}
\end{assumption}
\begin{remark}
In the dynamic setting, as in the static case, the approach I outline below can accommodate other stochastic restrictions, provided they can be represented by a finite number of linear inequality constraints once the model is 'locally discretized.' Furthermore, as mentioned in Remark \ref{rem2}, since there are no extra location restrictions on the shocks, we can assume without loss of generality that the fixed effects are degenerate and take the value of 0 in the assumptions of conditional stationarity and conditional exchangeability.
\end{remark}
I now discuss the construction of the DCPs and DDCPs. First, I will discuss their construction assuming stationarity or exchangeability conditional on the initial condition (Assumptions \ref{assumptionCS} and \ref{assumptionCE}). Then, I will discuss further below the analogous construction under the assumption of unconditional stationarity or exchangeability (Assumptions \ref{assumptionS} and \ref{assumptionE}). 
\subsection{Construction of DCP under conditional stochastic restrictions}
\label{sectionCond}
In this section, I discuss the construction of the DCP when the conditional (on the initial condition) stationarity or exchangeability assumption is maintained (i.e., assumptions \ref{assumptionCS} and \ref{assumptionCE}). Following the exposition in Section\ref{sectionDiscretization}, I first discuss the construction of the patches, and discuss how the patches can be used to discretize the space of shocks.  I then discuss the construction of the $A$ matrices and $R$ matrices that appear in the definition of the DCPs.\par
Fix a covariate value $x$, a parameter value $\theta=(\beta,\gamma)$ and a value $y_0\in [D]$ of the initial condition. For $t\in [T]$ and $d\in[D]$, let $v_{dt}:=x_{dt}'\theta$. For $i\in\{CE,CS\}$, let ${\cal P}_i(y_0,x,\theta)$ denote the set of all CCP vectors that are consistent with the model \ref{eqndyn} at the covariate value $x$, when the initial condition is $y_0$, the parameter value is $\theta$, and either the conditional stationarity restriction (CS) or conditional exchangeability restriction (CE) holds (see equations \ref{eqnDCPS1} and \ref{eqnDCPE1} for analogous definitions in the static setting). I define the patches in terms of the choices that an individual makes at each time period $t\in[T]$ and in each possible "state". Here, as we are considering a setting with one lag, the state in period t refers to the individual's choice in period $t-1$.\footnote{In a setting with L lags, the state in period t is determined by the choices in the past L periods. See Example \ref{exampDyn2} for a case with two lags.} At time 1, the initial condition determines the unique state, and for each $t\geq 2$, there are $D$ possible states, one for each of the possible choices in period $t-1$. For each $t\geq 2$, each $d\in [D]$ and each $s \in [D]$, let $\varepsilon_{d,t,s}$ denote the set of all possible realizations of the shocks that induce an individual to choose alternative $d$ at time $t$ when the state is $s$:
\begin{equation}
\label{eqnPartDyn1}
{\cal \varepsilon}_{d,t,s}:=\{\zeta \in \mathbb{R}^D | \ v_{dt}+\gamma \mathbbm{1}[d=s]+\zeta_{d}>v_{d't}+\gamma \mathbbm{1}[d'=s]+\zeta_{d'}\ \ \forall d'\neq d, \ d'\in[D]\}.
\end{equation}
When $t=1$,for each $d\in[D]$, let $\varepsilon_{d,1,y_0}$ represent the set of shocks that induce an individual to choose option $d$ at time 1 (given the state $y_0$ which is fixed by the initial condition):
\begin{equation}
\label{eqnPartDyn2}
{\cal \varepsilon}_{d,1,y_0}:=\{\zeta \in \mathbb{R}^D | \ v_{d1}+\gamma \mathbbm{1}[d=y_0]+\zeta_{d}>v_{d'1}+\gamma \mathbbm{1}[d'=y_0]+\zeta_{d'}\ \ \forall d'\neq d, \ d'\in[D]\}.
\end{equation}
I define a patch as a (non-empty) subset of realizations of the shocks that induces an agent to make a prescribed choice in each possible state and time period. More concretely, a tuple $f=(c,C)$, where $c\in[D]$ and $C$ is a (T-1)-by-D matrix with all of its entries in $[D]$, indexes a patch if and only if 
\begin{equation}
\label{eqnPatchDyn}
\varepsilon_{c,1,y_0} \cap (\cap_{t=2}^T\cap_{s=1}^D \varepsilon_{C(t-1,s),t,s})\neq \emptyset
\end{equation}
where $C(i,j)$ denotes the $(i,j)^{th}$ entry of the matrix $C$. Essentially, given a patch $f=(c,C)$, $c$ denotes the choice in period 1, in the unique state determined by the initial condition, and the entry in the $t^{th}$ row and $s^{th}$ column of $C$ represents the choice in period $t+1$ when the state is $s$. I denote the set of all possible patches by $\bf{F}$. Determining whether a pair $(c,C)$ (with $c\in[D]$ and $C\in [D]^{(T-1)\times D}$) represents a patch, can be achieved by checking for the feasibility of a linear program as in \ref{eqnPatch}. The set of patches partition $\mathbb{R}^D$ (up to the boundary of the patches). As in Section \ref{sectionDiscretization}, I use the patches to partition the space of joint (across all T time periods) shocks into rectangular regions of the type  $f_1\times \cdots \times f_T$ where $f_t \in F$ for each $t\in[T]$. I denote the set $|F|^T$-element set of all such regions by ${\cal R}$.\par
Once we obtain the set $F$ of patches, the matrices $A(y_0,x,\theta)$ and $R_i(y_0,x,\theta)$ used in the definition of the DCPs can be constructed as follows. The matrix $A$ is a $D^T$ by $|{\cal R}|$ matrix, with each row representing a ``choice sequence" and each column indexing a particular region of ${\cal R}$. Let $d=(d_1,\cdots,d_T)$ and $R=f_1\times \cdots \times f_T$, with each $d_t\in [D]$ and each $f_t=(c_t,C_t)\in F$. The entry of the $d^{th}$ row and $R^{th}$ column of $A$, is equal to 1 (and is equal to zero otherwise) if and only if for each time period $t\in[T]$, realizations of the shocks in the patch $f_t$, when the state is $d_{t-1}$ (with $d_0=y_0$), induce the agent to choose alternative $d_t$. That is, $A(d,R)=1$ if and only if 
\[
d_1=c_1 \quad \text{and}\quad d_t=C_t(t-1,d_{t-1}) \ \ \forall \ t\in[T]\backslash[1].
\]
The matrices $R_i(y_0,x,\theta)$, with $i\in\{CS,CE\}$ are easy to construct, and simply enforce the stochastic restrictions \ref{assumptionCS} and \ref{assumptionCE} on distributions defined on the discrete set ${\cal R}$. For instance, the matrix $R_{CS}$ simply enforces the restrictions: $\forall f \in F$ and $\forall i,j\in[T]$
\[
\sum_{\{R=f_1\times \cdots \times f_T \in {\cal R}\ | \ f_i=f\}} q_R=\sum_{\{R=f_1\times \cdots \times f_T \in {\cal R}\ | \ f_j=f\}} q_R.
\]
The proof of Proposition \ref{prop1} can easily be modified to show that for $i \in\{CS,CE\}$, we have
\[
{\cal P}_i(y_0,x,\theta)=\{p\in \mathbb{R}^{D^T}| \mathbbm{1}^Tp=1\}\cap \{p=A(y_0,x,\theta)q| \ q\in \mathbb{R}^{|{\cal R}|},\ R_i(y_0,x,\theta) q=0, \ q\geq 0\}.
\]
Moreover, in the current setting, Proposition \ref{propPatch} and Theorem \ref{thm1} can be adapted straightforwardly and remain valid.
\subsection{Construction of DCP under unconditional stochastic restrictions}
\label{sectionUncond}

In this section, I discuss how the construction of Section \ref{sectionCond} needs to be modified if the stationarity or exchangeability restriction that we impose on the shocks is not made conditional on the initial condition (i.e., if the shocks satisfy assumptions \ref{assumptionS} or \ref{assumptionE}).\par
Fix a covariate value $x$ and a parameter value $\theta=(\beta,\gamma)$. For $t\in [T]$ and $d\in[D]$, let $v_{dt}:=x_{dt}'\theta$. For $i\in{UE,US}$, let ${\cal P}_i(x,\theta)$ denote the set of all possible CCP vectors that are consistent with model \ref{eqndyn} at the covariate value $x$ and the parameter value $\theta$, under either the unconditional stationarity (US) or the unconditional exchangeability (UE) restriction. I define the patches in terms of the choices made by an individual at each time period $t\in[T]$ and in each possible state. Here (again), the state in period $t$  simply refers to the individual's choice in period $t-1$. For each $t\in [T]$, there are $D$ possible states, one for each $s\in[D]$. For each $d\in [D]$, each $t\in [T]$, and each $s\in [D]$, let $\varepsilon_{d,t,s}$ denote the set of all possible realizations of the shocks that induce an individual to choose alternative $d$ at time $t$ when the state is $s$:
\begin{equation}
\label{eqnPartDyn3}
{\cal \varepsilon}_{d,t,s}:=\{\zeta \in \mathbb{R}^D | \ v_{dt}+\gamma \mathbbm{1}[d=s]+\zeta_{d}>v_{d't}+\gamma \mathbbm{1}[d'=s]+\zeta_{d'}\ \ \forall d'\neq d, \ d'\in[D]\}.
\end{equation}
Here, I define a patch as a non-empty set of values of the shocks that induces an agent to make a prescribed choice in each possible state and time period. Specifically, a T-by-D matrix $f$, with all of its entries in $[D]$, indexes a patch if and only if
\begin{equation}
\label{eqnPatchDyn2}
\cap_{t=1}^T\cap_{s=1}^D \varepsilon_{f(t,s),t,s}\neq \emptyset
\end{equation}
where $f(i,j)$ denotes the $(i,j)^{th}$ entry of the matrix $f$. Here, the entry in the $t^{th}$ row and $s^{th}$ column of $f$ represents the choice that an agent will make at time $t$ and in state $s$, if the realization of her shock in period $t$ belong to the patch f. I denote the set of all possible patches by $\bf{F}$. Determining whether a matrix $f\in [D]^{T\times D}$ indexes a patch can be done by checking for the feasibility of a linear program as in \ref{eqnPatch}. The set of patches partitions $\mathbb{R}^D$ (up to the boundary of the patches). I use the patches to partition $\mathbb{R}^{DT}$, the space of shocks across all alternatives and time periods,  into rectangular regions of the type  $f_1\times \cdots \times f_T$ where $f_t \in F$ for each $t\in[T]$. I denote the set $|F|^T$-element set of all such regions by ${\cal R}$. Let $\{R_1,R_2,\dots,R_{|{\cal R}|}\}$ denote a fixed indexation/enumeration of the elements of ${\cal R}$.\par
I now discuss the construction of the matrix $A$. For each $g\in[D]$, let $A_g$ be a $D^{T+1}$ by $|{\cal R}|$ matrix that encodes the choices made by individuals with the initial condition $Y_{0}=g$, in each region $R \in {\cal R}$. Each row of $A_g$ represents a  possible ``choice sequence" (i.e., an element of $[D]^T$) and a possible value of the initial condition (i.e., an element of $[D]$), and each column of $A_g$ represents a region in ${\cal R}$. Let $d=(d_0,d_1,\cdots,d_T)$ and $k\in[|{\cal R}|]$. Then the entry in the $d^{th}$ row and $k^{th}$ column of $A_g$ is equal to 1 (and is equal to zero otherwise) if and only if
\[
d_0=g \quad \text{and} \quad f_t(t,d_{t-1})=d_t, \ \text{for each } t\in[D],
\]
where $f_1,\cdots,f_T$ represent the patches that make up the region $R_k$, i.e., $R_k=f_1\times\cdots\times f_T$. 
The matrix $A(x,\theta)$ is obtained by concatenating the $A_g$ matrices 
\[
A=[A_1 \ A_2 \cdots \ A_D].
\]
I now discuss the construction of the matrices $R_i(x,\theta)$, with $i\in\{US,UE\}$, that enforce the stochastic restrictions \ref{assumptionS} and \ref{assumptionE}. For $g\in [D]$ and $R\in{\cal R}$, let $q_{g,R}$ represent the probability (conditional on the covariate value $x$) that an individual chooses alternative $g$ in period 0 (i.e., $Y_0=g$) and that her shocks are in region $R$:
\[
q_{g,R}=P(Y_0=g, (\zeta_1,\cdots,\zeta_T)\in R|x).
\]
Let $q\in \mathbb{R}^{D|{\cal R}|}$ (the same dimension as the number of columns of $A$) be the vector with $\left(k+(g-1)|{\cal R}|\right)^{th}$ entry (for $g\in[D]$ and $k\in |{\cal R}|$) given by $q_{g,R_k}$. Under the unconditional stationarity, the matrix $R_{US}(x,\theta)$ encodes the linear restrictions: For each $f\in F$ and each $i,j\in [T]$
\[
\sum_{g=1}^D \ \  \sum_{\{R=f_1\times \cdots \times f_T \in {\cal R}\ | \ f_i=f\}} q_{g,R}=\sum_{g=1}^D \ \ \sum_{\{R=f_1\times \cdots \times f_T \in {\cal R}\ | \ f_j=f\}} q_{g,R}.
\]
Similarly, under unconditional exchangeability, the matrix $R_{UE}(x,\theta)$  encodes the linear restrictions: For each region $R=f_1\times\cdots\times f_T$ and each permutation $\pi$ on $[T]$
\[
\sum_{g=1}^D q_{g,f_1\times\cdots\times f_T}=\sum_{g=1}^D q_{g,f_{\pi(1)}\times\cdots\times f_{\pi(T)}}.
\]
The proof of Proposition \ref{prop1} can be adapted to show that for $i\in\{US,UE\}$, we have

\[
{\cal P}_i(x,\theta)=\{p\in \mathbb{R}^{D^T}| \mathbbm{1}^Tp=1\}\cap \{p=A(x,\theta)q| \ q\in \mathbb{R}^{|{\cal R}|},\ R_i(x,\theta) q=0, \ q\geq 0\}.
\]
Moreover, it can easily be shown that the adaptation of Proposition \ref{propPatch} and Theorem \ref{thm1} to the current setting remain valid.
\subsection{Examples of applications of the algorithm in dynamic settings}
\label{sectionExampDyn}
In this section, I present three examples of dynamic settings where the algorithm described in the preceding sections is implemented. The first example examines the dynamic binary choice model with one lag studied in \cite{KPT}. There, I use to algorithm to recover the inequalities that were obtained in \cite{KPT} by analytical means. In the second example, I consider a two-period dynamic multinomial choice model with one lag. This model can be seen as a simultaneous extension of the settings of \cite{PP1} and \cite{KPT}.\footnote{It can be viewed as adding dynamics to the setting of \cite{PP1}, or it can alternatively be viewed as extending the setting of \cite{KPT} to more than two alternatives.} Using the algorithm, I generate the inequalities that characterize the identified set of the common parameter when there are four alternatives (D=4). These inequalities are then proven analytically, and I provide their generalizations to settings with any number of alternatives. The set of inequalities obtained is novel and subsumes the inequalities of \cite{PP1} when $\gamma=0$ (i.e., when there is no dynamic) as well as the inequalities of \cite{KPT} when $D=2$. Moreover, the inequalities that I obtain are complementary to those derived in \cite{PP2} under a stronger stochastic restriction. In the third example, a three-period dynamic binary choice model with two lags is examined. Here, some of the inequalities that the algorithm generates are natural extensions of those in the one-lag setting of \cite{KPT}, while the others are less intuitive, as evidenced by their intricate analytical proofs.

\begin{example}[\textbf{\cite{KPT}} ]
\label{exampDyn1}
In a two periods dynamic binary choice model, let agent's $i$ choice in each period $t\in[2]$ be modeled by
\begin{equation}
\label{eqnKPT}
Y_{ti}=\mathbbm{1}\{X_{it}'\beta+\gamma Y_{(t-1)i}-\lambda_i-\varepsilon_{ti}> 0\}
\end{equation}
and suppose that the random utility shocks satisfy the unconditional stationarity condition (Assumption \ref{assumptionS}). As discussed in Section \ref{sectionUncond}, it is without loss of generality to assume that the fixed effects are degenerate and take the value $0$. To implement the procedure, it is necessary to construct the set F of patches. In the current setting, I now provide a detailed discussion of the construction of these patches.\par
Let us assume a value of the covariates $x$ and a value of the parameter $\theta=(\beta,\gamma)$. For $t\in[2]$, set $v_t=x_t'\beta$. Here a patch $f$ is an element of $\{0,1\}^{2\times 2}$, where the entry $f(t,s)$ (with $t,s \in[2]$) is interpreted as the choice made by the individual at time $t$ when the state (i.e., the choice in period t-1) is given by $s-1$.  An matrix $f\in \{0,1\}^{2\times 2}$ represents a patch if and only if the following system of inequalities admits a solution $\zeta$:
\begin{equation}
\begin{aligned}
\label{eqnKPT1}
f(1,1)(\zeta-v_1)+(1-f(1,1))(v_1-\zeta) &< 0 \\
f(1,2)(\zeta-\gamma-v_1)+(1-f(1,2))(v_1+\gamma-\zeta) &<0\\
f(2,1)(\zeta-v_2)+(1-f(2,1))(v_2-\zeta) &<0\\
f(2,2)(\zeta-\gamma-v_2)+(1-f(2,2))(v_2+\gamma-\zeta) &<0.
\end{aligned}
\end{equation}
As the latter is a system of linear inequalities in a scalar variable $\zeta$, checking whether such a system admits a solution simply reduces to checking whether four intervals (one interval for each inequality in the system) have non-empty intersections, and can be framed as a LP:
\begin{equation}
\label{eqnPatchKPT}
\begin{array}{ll@{}ll}
\text{minimize}  & 0 &\\
&\text{subject to }  \zeta \in \mathbb{R} &\\
&   f(1,1)(\zeta-v_1)+(1-f(1,1))(v_1-\zeta) \leq \delta \\
&f(1,2)(\zeta-\gamma-v_1)+(1-f(1,2))(v_1+\gamma-\zeta) \leq \delta \\
&f(2,1)(\zeta-v_2)+(1-f(2,1))(v_2-\zeta) \leq \delta \\
&f(2,2)(\zeta-\gamma-v_2)+(1-f(2,2))(v_2+\gamma-\zeta) \leq \delta.
\end{array}
\end{equation}
Here, $\delta$ is a small nonnegative tolerance parameter (say $\delta=-10^{-4}$) that is used to replace the strict inequalities  in \ref{eqnKPT1} by weak inequalities, and the set of patches $F$ consists of all matrices $f\in \{0,1\}^{2\times 2}$ (there are 16 such matrices, and we thus only have to consider 16 LPs) for which the preceding LP has the value 0.  After obtaining the set $F$ of patches, we can construct the matrices $A$ and $R_S$ using the definitions provided in Section \ref{sectionUncond}.\par
Implementing the algorithm \footnote{This involves solving for the patches, constructing the matrices A and $R_S$, solving for the resulting MOLP in \ref{eqnMolp} using Benson's algorithm, and applying the redundancy elimination algorithm at the end of Section \ref{sectionCuttingPlane}. For the current example, this all takes less than one second of running time.} with values of $x$ and $\theta$  chosen such that $v_1=0$, $v_2=1$ and $\gamma=2$, \footnote{As in Footnote \ref{foot19}, it is without loss of generality to assume that $v_1=0$.} yields following set of vectors 

\[
\bordermatrix{& p_{000}&p_{001}&p_{010}&p_{011}&p_{100}&p_{101}&p_{110}&p_{111}\cr
&-1&0&-1&-1&0&1&-1&-1\cr
&0&-1&1&0&0&-1&0&-1\cr
&-1&-1&1&0&-1&-1&1&0
}\qquad
\]
where for $d_0,d_1,d_2 \in\{0,1\}$, $p_{d_0d_1d_2}$ represents the entry of the CCP vector equal to $P(Y_0=d_0, Y_1=d_1, Y_2=d_2|X=x)$.
The first row corresponds to the inequality
\[
p_{101}\leq p_{000}+p_{010}+p_{011}+p_{110}+p_{111}
\]
and adding $p_{001}$ to both sides yields the inequality
\begin{equation}
\label{eqnKPT2}
P(Y_1=0,Y_2=1|x) \leq 1- P(Y_0=1,Y_1=0|x).
\end{equation}
The second row corresponds to the inequality
\[
p_{010}\leq p_{001}+p_{101}+p_{111}
\]
and adding $p_{110}+p_{100}+p_{000}$ to both sides yields 
\begin{equation}
\label{eqnKPT3}
P(Y_2=0 | x)\leq1- P(Y_0=0, Y_1=1|x).
\end{equation}
The third row corresponds to the inequality
\[
p_{010}+p_{110}\leq p_{000}+p_{001}+p_{100}+p_{101}
\]
which can be rewritten as
\begin{equation}
\label{eqnKPT4}
P(Y_1=1,Y_2=0|x)\leq P(Y_1=0|x).
\end{equation}
From the system of inequalities \ref{eqnKPT1}, the set of patches $F$ that arise in the discretization is uniquely determined from the ordering of the quantities $v_1$ ,$v_2$, $v_1+\gamma$, and $v_2+\gamma$. Hence the inequalities that we have derived above represent the only restrictions that the model places on the CCP vectors whenever
\begin{equation}
\label{eqnKPT5}
v_1<v_2<v_1+\gamma<v_2+\gamma,
\end{equation}
as this is the ordering that holds for our chosen values of $v_1=0$, $v_2=1$ and $\gamma=2$. When Inequality \ref{eqnKPT5} holds, it was established in \cite{KPT} that the set of all the restrictions that the model places on the identified CCP vectors are given by inequalities \ref{eqnKPT2}--\ref{eqnKPT4} as well as:
\begin{equation}
\label{eqnKPT6}
P(Y_0=1,Y_1=1|x) \leq 1- P(Y_1=1,Y_2=0|x)
\end{equation}
and
\begin{equation}
\label{eqnKPT7}
P(Y_0=0,Y_1=1|x) \leq 1- P(Y_1=0,Y_2=0|x).
\end{equation}
As our algorithm generates a minimal set of inequalities that characterize the identified set, it can be shown that the latter two inequalities are redundant given \ref{eqnKPT2}--\ref{eqnKPT4}. Indeed, inequality \ref{eqnKPT6} follows from \ref{eqnKPT4} since $P(Y_1=0|x)\leq 1-P(Y_0=1,Y_1=1)$, and  inequality \ref{eqnKPT7} follows from \ref{eqnKPT4} since $P(Y_1=0,Y_2=0|x)\leq P(Y_2=0|x)$. To obtain all the implications that characterize the identified set we need to consider other possible orderings of the quantities in \ref{eqnKPT5}. The three main additional cases to consider correspond to the orderings: $v_1<v_1+\gamma<v_2<v_2+\gamma$, $v_1+\gamma<v_2+\gamma<v_1<v_2$ and $v_1+\gamma<v_1<v_2+\gamma<v_2$.\footnote{By relabeling the time periods if necessary, we can alway assume without loss of generality that $v_1\leq v_2$. Then, the latter three cases and \ref{eqnKPT5} represent all possible orderings of the quantities $v_1$, $v_2$, $v_1+\gamma$ and $v_2+\gamma$ that do not involve ties. There are then 6 possible additional orderings of the same quantities that involve ties. These are: $v_1=v_2<v_1+\gamma=v_2+\gamma$, $v_1+\gamma=v_2+\gamma<v_1=v_2$, $v_1=v_2=v_1+\gamma=v_2+\gamma$, $v_1=v_1+\gamma<v_2=v_2+\gamma$, $v_1<v_1+\gamma=v_2<v_2+\gamma$  and $v_1+\gamma<v_1=v_2+\gamma<v_2$. Hence a total of 10 cases (and we thus only need to solve 10 MOLPs) need to be considered to generate all the inequalities that characterize the identified set.} Repeating the above procedure on each of these additional cases should generate all the inequalities that were obtained in \cite{KPT} for $T=2$. 
Note that while \cite{KPT} showed that the conditional moment inequalities that characterize the sharp identified set of $\theta$ when $T=2$ also characterize the sharp identified set when $T>2$ (if we consider all such inequalities for all possible pairs of time periods), the algorithm presented in this paper can only be applied on a case-by-case basis. Specifically, running the algorithm for $T=2$ generates all the inequalities that characterize the identified set for $T=2$, but our approach does not provide information on whether these inequalities, when applied to all pairs of time periods, will characterize the identified set for $T>2$.

\end{example}
\begin{example}[\textbf{\cite{PP1}} with one lag ]
\label{exampDyn3}

In this example, I extend the model presented in \cite{PP1} to a dynamic setting. Specifically, I analyze model \ref{eqndyn} under the assumption of conditional stationarity \ref{assumptionCS}, where $T=2$ and $D=4$.\footnote{This example can also be viewed as an extension of the model of \cite{KPT} to a setting with more than two alternatives.} Following the procedure outlined in Section \ref{sectionCond}, I construct the matrices A and $R_S$ for the input values: $v_1=(0,\cdots,0)$, $v_2=(0,3,5,7)$, $\gamma=7$, and the initial condition $y_0=3$. Here, similar to Example \ref{exampPPS}, given a fixed covariate value $x$ and a fixed parameter value $\beta$, for $t\in[2]$ and $d\in[4]$, we have $v_{t}=(v_{1t},\cdots,v_{Dt})$, where $v_{dt}=x_{dt}'\beta$ represents the "index component" of utility for alternative $d$ at time $t$. Solving the resulting MOLP using Benson's algorithm, yields the following set of inequalities, where $p_{dd'}=P(Y_1=d,Y_2=d'|X=x,Y_0=3)$ (the whole procedure takes about 1800 seconds):
\begin{enumerate}
\item $p_{31}+p_{32}\leq p_{11}+p_{12}+p_{13}+p_{14}+p_{21}+p_{22}+p_{23}+p_{24}$
\item $p_{12}+p_{13}+p_{43}\leq p_{21}+p_{22}+p_{24}+p_{31}+p_{32}+p_{33}+p_{34}$
\item $p_{12}+p_{13}+p_{14}\leq p_{21}+p_{22}+p_{24}+p_{31}+p_{32}+p_{33}+p_{34}+p_{41}+p_{42}+p_{44}$
\item $p_{31}\leq p_{11}+p_{12}+p_{13}+p_{14}$
\item $p_{41}+p_{42}+p_{43}\leq p_{14}+p_{22}+p_{24}+p_{34}$
\item $p_{21}+p_{23}+p_{41}+p_{43}\leq p_{11}+p_{12}+p_{14}+p_{32}+p_{33}+p_{34}$
\item $p_{21}+p_{23}+p_{24}\leq p_{11}+p_{12}+p_{14}+p_{32}+p_{33}+p_{34}+p_{42}+p_{44}$
\item $p_{13}+p_{23}+p_{43}\leq p_{31}+p_{32}+p_{33}+p_{34}.$
\end{enumerate}
 With a little algebra, the latter inequalities can be rewritten as (all probabilities below are implicitly conditioned on $X=x$ and $Y_0=3$):

\begin{enumerate}
\item $P(Y_1=3,Y_2\in\{1,2\})\leq P(Y_1\in\{1,2\})$
\item $P(Y_1=1, Y_2=2)+P(Y_1\in \{1,2,4\},Y_2=3)\leq P(Y_1\in \{2,3\})$
\item $P(Y_1=1,Y_2\in\{2,3,4\})+P(Y_1\in \{2,4\},Y_2=3)\leq P(Y_1\in \{2,3,4\})$
\item $P(Y_1=3,Y_2=1)\leq P(Y_1=1)$
\item $P(Y_2\in\{1,3\})+P(Y_1\in\{1,3,4\},Y_2=2)\leq P(Y_1\in\{1,2,3\})$
\item $P(Y_1\in\{2,3,4\},Y_2=1)+P(Y_1\in\{1,2,4\},Y_2=3)\leq P(Y_1\in\{1,3\})$
\item $P(Y_1\in \{2,3,4\},Y_2=1)+P(Y_1\in \{1,2,4\},Y_2=3)+P(Y_1=2,Y_2=4)\leq P(Y_1\in \{1,3,4\})$
\item $P(Y_1\in \{1,2,4\},Y_2=3)\leq P(Y_1=3)$.
\end{enumerate}

These inequalities represent all of the restrictions that the model places on the CCP vectors in ${\cal P}_{CS}(y_0,x,\theta)$, when $y_0=3$ and $x$ and $\theta$ are such that $\gamma=7$, $v_1=(0,\cdots,0)$ and $v_2=(0,3,5,7)$. As in Example \ref{exampPPE}, I proceed to prove these inequalities analytically, and from the proof I guess and prove their general form (Theorem \ref{thmPPD1} below), which I now discuss.\par

Given a parameter value $\theta$, a covariate value $x$ and an initial condition $y_0\in [D]$ (with $D\geq2$) let the family of sets ${\cal A}(y_0,x,\theta)\subseteq 2^{[D]}$ be defined as follows: For each $d_1\in [D]$, let $\Delta(d,d_1,y_0,x,\theta)$ be defined by
\[
\Delta(d,d_1,y_0,x,\theta)=v_{d2}-v_{d1}+\gamma\mathbbm{1}\{d=d_1\}-\gamma\mathbbm{1}\{d=y_0\}
\]
where for $t\in[2]$, $v_{dt}:=x_{dt}'\beta$. The quantity $\Delta(d,d_1,y_0,x,\theta)$ represents the improvement (from period 1 to 2) of the deterministic component of the utility of alternative $d$, when the "state" in period 2 is given by $d_1$ (i.e., $d_1$ is the alternative chosen in period 1) and the state in period 1 is given by $y_0$ (i.e., $y_0$ is the initial condition). Note that if $\gamma=0$, then $\Delta(d,d_1,y_0,x,\theta)$ simply represents the index function differences, and does not depend on the initial condition $y_0$ nor on $d_1$. Let ${\cal L}(d_1,y_0,x,\theta)$ be the family of subsets of $[D]$, the "lower sets", defined by
\begin{equation}
\label{PPD1}
{\cal L}(d_1,y_0,x,\theta)=\{A\subset [D]\ | \ \emptyset \subsetneq A \subsetneq [D], \ \Delta(d',d_1,y_0,x,\theta)\leq \Delta(d'',d_1,y_0,x,\theta) \ \forall  d'\in A,  d''\in [D]\backslash A\}.
\end{equation}
For instance, if we have 3 alternatives, and the quantities $\Delta(d,d_1,y_0,x,\theta)$ satisfy the ordering $\Delta(2,d_1,y_0,x,\theta)<\Delta(1,d_1,y_0,x,\theta)=\Delta(3,d_1,y_0,x,\theta)$, then ${\cal L}(d_1,y_0,x,\theta)=\{ \{2\},\{1,2\},\{2,3\} \}$. The family ${\cal A}(y_0,x,\theta)$ is then defined by
\begin{equation}
\label{PPD2}
{\cal A}(y_0,x,\theta)=\bigcup_{d\in[D]}{\cal L}(d,y_0,x,\theta).
\end{equation}
With the foregoing definitions, the following theorem generalizes the eight inequalities generated by the algorithm to a setting with arbitrary $D$, arbitrary values of the initial condition $y_0$, the covariates $x$, and arbitrary values of the parameter $\theta$; its proof is provided in Section \ref{remDetails} of the Appendix.

\begin{theorem}
\label{thmPPD1}
Consider model \ref{eqndyn} in a two-period setting, and suppose that Assumption \ref{assumptionCS} holds. Let $\theta_0=(\beta_0,\gamma_0)$ denote the true parameter value. Then for all $(y_0,x)\in supp(Y_0,X)$ and for all $A\in {\cal A}(y_0,x,\theta_0)$, we have:
\begin{equation}
\label{eqnthmPPD}
\begin{split}
&P\left(\bigcup_{d \in [D]}\Biggl\{Y_1=d,Y_2\in \bigcup \Bigl\{B \ | \ B\subseteq A, B \in {\cal L}(d_1,y_0,x,\theta_0)\Bigl\}\Biggl\} \bigg|Y_0=y_0,X=x\right) \\
 &\quad \quad \quad \leq P(Y_1\in A|Y_0=y_0,X=x).
\end{split}
\end{equation}
\end{theorem}

\begin{remark}
\label{remPPD1}
Note that for the input values we use in our algorithm to generate the 8 inequalities listed above (i.e., $D=4$, $y_0=3$, $\gamma=7$, $v_1=(0,0,0,0)$ and $v_2=(0,3,5,7)$), the inequalities in Theorem \ref{thmPPD1} correspond exactly to these 8 inequalities. Indeed, for these inputs, when $d_1=1$, the quantities $\Delta(d,d_1,y_0,x,\theta)$ are given by:
\[
\Delta(1,1,y_0,x,\theta)=7 \quad \Delta(2,1,y_0,x,\theta)=3 \quad \Delta(3,1,y_0,x,\theta)=-2 \quad \Delta(4,1,y_0,x,\theta)=7
\] 
leading to the ordering
\[
\Delta(3,1,y_0,x,\theta)<\Delta(2,1,y_0,x,\theta)<\Delta(1,1,y_0,x,\theta)=\Delta(4,1,y_0,x,\theta),
\]
and the resulting set ${\cal L}(1,y_0,x,\theta)$ is given by 
\[
{\cal L}(1,y_0,x,\theta)=\{\{3\},\{2,3\},\{1,2,3\},\{2,3,4\}\}.
\]
Similar computations show that the sets ${\cal L}(2,y_0,x,\theta)$, ${\cal L}(3,y_0,x,\theta)$ and ${\cal L}(4,y_0,x,\theta)$ are given by

\begin{align*}
{\cal L}(2,y_0,x,\theta)&=\{\{3\},\{1,3\},\{1,3,4\} \}\\
{\cal L}(3,y_0,x,\theta)&=\{\{1\},\{1,2\},\{1,2,3\} \}\\
{\cal L}(4,y_0,x,\theta)&=\{\{3\},\{1,3\},\{1,2,3\} \}.
\end{align*}
Hence the set ${\cal A}(y_0,x,\theta)$ is given by
\[
{\cal A}(y_0,x,\theta)=\{  \{1\}, \{3\}, \{1,2\}, \{1,3\}, \{2,3\}, \{1,2,3\}, \{1,3,4\}, \{2,3,4\}       \}.
\]
The latter set has 8 elements, and I now show that each of its elements corresponds to one of the 8 inequalities generated by the algorithm. Consider the case where $A=\{1\}$. Since the set $\{1\}$ only belongs to ${\cal L}(d,y_0,x,\theta)$ when $d=3$, inequality \ref{eqnthmPPD} can be written as follows in this case:
\[
P(Y_1=3,Y_2=1)\leq P(Y_1=1),
\]
which coincides with the $4^{th}$ inequality generated by the algorithm. When $A=\{1,2,3\}$, the  largest element of ${\cal L}(d,y_0,x,\theta)$ contained in $A$ is given by $\{1,2,3\}$ when $d\in\{1,3,4\}$, and by  $\{1,3\}$ when $d=2$. Thus, inequality \ref{eqnthmPPD} yields
\[
P(Y_1\in\{1,3,4\},Y_2\in \{1,2,3\})+P(Y_1=2,Y_2\in\{1,3\})\leq P(Y_1\in\{1,2,3\}),
\]
which coincides with the $5^{th}$ inequality generated by the algorithm. The remaining inequalities can be derived by analogous computations. Furthermore, as the algorithm deletes redundant inequalities, the 8 inequalities that the theorem produces in this case give a compact (i.e., free of redundancies) characterization of all the restrictions that the model puts on the CCPs.
\end{remark}

\begin{remark}
\label{remPPD2}
When $\gamma=0$, the dynamic model in \ref{eqndyn} reduces to the static model in \ref{eqnstat}. In this case, the assumptions of conditional stationarity (\ref{assumptionCS}) and stationarity (\ref{assumptionS}) are observationally equivalent. This is so since, for a fixed value of the index parameter, the sets of conditional choice probability vectors (for periods 1 and 2) that are consistent with either assumption coincide. Furthermore, the inequalities of Theorem \ref{thmPPD1} coincide with those of Proposition 1 in \cite{PP1}. Indeed, when $\gamma=0$, the quantities $\Delta(d,d',y_0,x,\theta)$ and the sets ${\cal L}(d'',y_0,x,\theta)$ do not respectively depend on $d'$ and $d''\in[D]$, and we have $A\in {\cal A}(y_0,x,\theta)$ if and only if $A^c\in \bar{\mathbb{D}}(x,\theta)$, where the set $\bar{\mathbb{D}}(x,\theta)$ is as defined in Proposition 1 of \cite{PP1}. Therefore, when $\gamma=0$, Theorem \ref{thmPPD1} encompasses the result of Proposition 1 in \cite{PP1}.
\end{remark}

\begin{remark}[\textbf{Comparison with the inequalities in \cite{PP2}}]
\label{remPPD5} 

\cite{PP2} also consider model \ref{eqndyn},\footnote{\label{footPPSC1}In \cite{PP2}, the actual specification of the utility level for alternative d at time t is given by $U_{dti}=X_{dti}'\beta-\gamma \mathbbm{1}_{[d\neq y_{(t-1)i}]}+\lambda_{di}+\epsilon_{dti}$ (a positive $\gamma$ in this specification can be interpreted as a "switching cost"). The specification in model \ref{eqndyn} is obtained from the latter by adding the constant $\gamma$ to all utility levels (a positive $\gamma$ in the resulting specification can be interpreted as  "inertia" or "stickiness").} and derive inequality restrictions on the CCPs under a stochastic restriction that is stronger than Assumption (\ref{assumptionCS}). As such, the inequalities of Theorem  \ref{thmPPD1} are valid in the setting of \cite{PP2}. Moreover, I show that under the assumptions of \cite{PP2}, inequalities \ref{eqnthmPPD} are complementary to those derived in Theorem 3.5 of \cite{PP2}, as they neither imply nor are implied by the latter inequalities. Thus, under their maintained assumptions, the inequalities derived in Theorem 3.5 of \cite{PP2} do not characterize the sharp identified set for the parameter $\theta$, and Theorem \ref{thmPPD1} provide some additional inequality restrictions that can be used to identify $\theta$.\par

Specifically, \cite{PP2} consider the model given by Equation \ref{eqndyn}, and assume that the shocks satisfy the following stochastic restriction:
\[
\epsilon_{ti}|Y_{(t-1)i},Y_{(t-2)i},\cdots,Y_{1i},Y_{0i},X_i,\lambda_i \ \sim \ \epsilon_{ti}|\lambda_i \  \sim \ \epsilon_{1i}|\lambda_i.
\]
It can be shown that the latter is observationally equivalent to assuming that, conditional on the fixed effects, the shocks are serially i.i.d., and independent of the initial condition $Y_0$ and the covariates $X$:\footnote{Note that given a parameter value $\theta\in \Theta$, the CCPs that are consistent with the model under either assumption are those that can be represented as $P(Y_1=d_1,\cdots,Y_T=d_T|Y_0=d_0,X=x)=\int \prod_{t\in [T]}G_{\lambda}(\varepsilon_{d_t,t,d_{t-1}})d\mu(\lambda|Y_0=d_0,X=x)$, for all $d_0,d_1,\cdots,d_T\in [D]$, where $\varepsilon_{d_t,t,d_{t-1}}$ is as defined in Equation \ref{eqnPartDyn1} and $G_{\lambda}$ are continuous distributions on $\mathbb{R}^D$.} i.e.,
\begin{equation}
\label{eqnPPSC1}
 \epsilon_{1i}\perp\cdots\perp \epsilon_{Ti}|\lambda_i \quad \text{and} \quad \epsilon_{ti}|\lambda_i \sim \epsilon_{1i}|\lambda_i\quad  \text{ and } \quad \epsilon_i|\lambda_i,Y_{0i},X_i \sim \epsilon_i|\lambda_i.
\end{equation}
Clearly, as a conditionally i.i.d. sequence is stationary, Assumption \ref{assumptionCS} is implied by Assumption \ref{eqnPPSC1}, and the inequalities of Theorem \ref{thmPPD1} are valid in the setting of \cite{PP2}.\footnote{ As stated in \cite{PP2}, under Assumption \ref{eqnPPSC1}, all serial dependence in choices not associated with the covariates is attributed to the fixed effects and to the state dependence parameter. Under our weaker assumption \ref{assumptionCS}, all serial dependence in the choices that is not associated with the covariates can additionally be attributed to the serial correlation of the shocks.}\par

When $T=2$, the inequalities in \cite{PP2} can be stated as follows. Let ${\cal D}(y_0,x,\theta)$ be the family of subsets of $[D]$ defined by
\begin{equation}
\label{eqnPPSC2}
{\cal D}(y_0,x,\theta)=\{A\subsetneq [D] \ |\ A^c\in {\cal L}(d,y_0,x,\theta), \forall d\in A\},
\end{equation}
where the family ${\cal L}(d_1,y_0,x,\theta)$ is as defined in Equation \ref{PPD1}.  Then the inequalities provided in Theorem 3.5 of \cite{PP2} are: For all $A\in {\cal D}(y_0,x,\theta)$, we have
\begin{equation}
\label{eqnPPSC3}
P(Y_2\in A|Y_1\in A,Y_0=y_0,X=x)\geq P(Y_1\in A|Y_0=y_0,X=x).
\end{equation}
I show in Section \ref{remDetails} of the Appendix that the latter inequalities do not imply, nor are implied by inequalities \ref{eqnthmPPD}, as it is possible for inequalities \ref{eqnPPSC3} to be informative about $\theta$ while inequalities \ref{eqnthmPPD} are not, and vice-versa. As a consequence, the inequalities of Theorem 3.5 of \cite{PP2} do not yield a characterization of the sharp identified set, and using them in conjunction with the inequalities of Theorem \ref{thmPPD1} can yield more informative bounds on the parameter $\theta$ (i.e., a smaller outer set of the identified set). Furthermore, as the conditional exchangeability assumption (\ref{assumptionCE}) is stronger than conditional stationarity (\ref{assumptionCS}) but weaker than Assumption \ref{eqnPPSC1}, the algorithm can be used to generate the inequalities that characterize the identified set for $\theta$ under the assumption of conditional exchangeability, and the resulting inequalities will be valid in setting of \cite{PP2}, and will provide further identifying restrictions for $\theta$. \par
Note that Assumption \ref{eqnPPSC1} does not impose any restrictions on the conditional distribution of the fixed effects, given the initial condition and the covariates. Therefore, model \ref{eqndyn} under Assumption \ref{eqnPPSC1} is observationally equivalent to model \ref{eqndyn} under the assumption\footnote{Note that if the covariates $X$ contain a "strongly exogenous" component $Z$ (here $X=(W,Z)$), in the sense that condition \ref{eqnPPSC1} holds and $\lambda_i | Y_{0i},W_i,Z_i\sim \lambda_i | Y_{0i},W_i$, then such a $Z$ makes it possible to exploit "between" variations in choices across individuals to identify $\theta$.}
\begin{equation*}
 \epsilon_{1i}\perp\cdots\perp \epsilon_{Ti}|Y_{0i},X_i,\lambda_i \quad \text{and}\quad  \epsilon_{ti}|Y_{0i},X_i,\lambda_i \sim \epsilon_{1i}|Y_{0i},X_i,\lambda_i.
\end{equation*}
The latter is a natural extension of the conditional IID stochastic restriction \ref{assumptionI} to the dynamic setting, and as in Proposition \ref{propImp} it can be shown that the set ${\cal P}(y_0,x,\theta)$ of CCPs that are consistent with the model at the covariate value $X=x$, the initial condition $Y_0=y_0$, and the parameter value $\theta$, are not polytopes \footnote{If the parameter value $\theta=(\beta,\gamma)$ is such that $\gamma=0$ (i.e., there are no dynamics), then the sets ${\cal P}(y_0,x,\theta)$ coincide with those considered in Proposition \ref{propImp}}. Hence, it follows from Theorem \ref{thm2} that, under the assumptions of  \cite{PP2}, the sharp identified set for the parameter $\theta$ does not admit a "simple" characterization,\footnote{Note however that the inequalities \ref{eqnPPSC3} of \cite{PP2} are nonlinear in the CCP vector $p=\{P(Y_1=d_1,Y_2=d_2| Y_0=y_0,X=x)\}_{d_1,d_2\in[D]}$.} and our algorithm cannot be used to generate the inequalities that characterize the sharp identified set.

%

\end{remark}
\begin{remark}
\label{remPPD3}
When $D=T=2$, the setting of Theorem \ref{thmPPD1} reduces to a two-period dynamic binary choice model with the assumption of conditional stationarity on the shocks. In this setting, the sharp identified set for the common parameter $\theta=(\beta,\gamma)$ is characterized by the conditional moment inequalities provided in Theorem 2 of \cite{KPT}. In the Appendix, I demonstrate that these inequalities are equivalent to those of Theorem \ref{thmPPD1}. Therefore, Theorem \ref{thmPPD1} encompasses the findings of Theorem 2 of \cite{KPT} when $D=T=2$.\footnote{When $D=2$, the model in Theorem \ref{thmPPD1} is relatively small, thus to confirm that the inequalities of Theorem \ref{thmPPD1} characterize the sharp identified set, one can use the computational algorithm proposed in this paper for all relevant configurations of the initial condition $y_0$, covariates $x$, and parameter $\theta$ (see Proposition \ref{propPatch}).} For further details, see Section \ref{remDetails} of the Appendix.

\end{remark}
\begin{remark}
\label{remPPD4}

When $D$ is large and Benson's algorithm is not computationally practical, we can use a probabilistic approach to determine whether the inequalities in Theorem \ref{thmPPD1} characterize the sharp identified set. Following Footnote \ref{foot13}, this approach involves solving the LPs $\max\{w'y \ | \ A(y_0,x,\theta)^Ty\leq R_{CS}(y_0,x,\theta)^Tz,\ \|y\|_{\infty} \leq 1 \}$ for a large number, $K$, of randomly chosen objective vectors $w>0$. The solution of each such LP is an undominated extreme point of the DDCP and corresponds to an inequality restriction on the CCPs. When $K$ is large, the inequalities that correspond to the solutions of these $K$ LPs will contain most, if not all, of the inequality restrictions that the model places on the CCPs. We can then compare this resulting set of inequalities to those of Theorem \ref{thmPPD1}, to determine whether all of the inequalities produced by the probabilistic approach are accounted for by the inequalities in \ref{eqnthmPPD} (or are redundant given the inequalities in \ref{eqnthmPPD}). If not, then the inequalities in Theorem \ref{thmPPD1} do not characterize the sharp identified set of $\theta$, and the new (unaccounted for) inequalities represent additional restrictions that the model places on the CCPs beyond those in Theorem \ref{thmPPD1}.\par
It is important to reemphasize that the probabilistic approach is not guaranteed to retrieve all of the undominated extreme points of the DDCP, and as such it may not retrieve all of the inequality restrictions that the model places on the CCPs. However, it is a useful alternative when Benson's algorithm is not feasible, as solving even a relatively large number of LPs can be done within a reasonable amount of time. Below, I illustrate the performance of the probabilistic approach with a simulation experiment. I consider the same inputs as those used by the algorithm to generate the eight inequalities listed above (i.e., $y_0=3$ and $x$ and $\theta$ are such that $\gamma=7$, $v_1=(0,\cdots,0)$ and $v_2=(0,3,5,7)$). These eight inequalities represent all of the inequality restrictions that the model places on CCPs vector in the local model ${\cal P}(y_0,x,\theta)$ associated with the given input values. I implement the probabilistic approach with different values of $K$: I consider $K\in\{50,100,500,1000\}$. For a given value of $K$, I consider 100 iterations. In each iteration, I randomly draw $K$ independent and identically distributed  objective vectors $\{w_i\}_{i=1}^K$ from a distribution $\mu$, where each objective vector $w_i$ has dimension equal to $D^T=16$ (the dimension of the CCP vectors). In particular, I let $\mu$ be the distribution of a random vector of dimension 16, whose components are independently and identically distributed with distribution given by an exponential distribution with parameter 1.\footnote{Other choices of $\mu$ are possible, and the only requirement is that the closure of the radial projection of the support of $\mu$ onto the unit sphere is equal to the set of all vectors in the unit sphere (in $\mathbb{R}^{16}$) with nonnegative entries . This projection is of interest because the LPs with objectives $w$ and $\alpha w$ have the same solutions for any scalar $\alpha>0$. Essentially, we want $\mu$ to "sample" all non-negative directions and only non-negative directions.} Given each objective $w_i$, I solve the corresponding LP and obtain the corresponding solution $y_i$. Thus solving all K LPs yields a set $\{y_i\}_{i=1}^K$ of undominated extreme points of the DDCP, and each one of these solutions corresponds to an inequality restriction on the CCPs. When $K$ is large the resulting set of inequalities is bound to have redundancies, and I use the redundancy removal algorithm outlined at the end of Section \ref{sectionCuttingPlane} to produce an equivalent and minimal (i.e., redundancy free) subset of $M$ inequalities. I report in Table \ref{table3} below the number of time (out of the 100 iterations) that $M$ takes each one of its possible values, for different values of $K$. Note that since Benson's algorithm solves for all the undominated extreme points of the DDCP, and using it produces  eight non-redundant inequalities, it must be the case that $M$, the number of non-redundant inequalities generated by the probabilistic approach, is always at most eight, and is equal to eight only when the probabilistic approach is able to retrieve all of the eight inequalities generated by Benson's algorithm. As can be seen from the table's output, when $K=1000$, for 99 out of 100 iteration, the probabilistic approach is able to obtain all eight inequalities that were generated by Benson's algorithm. Moreover, these iterations each have an average running time of around 33 seconds, compared to the running time of 1800s for Benson's algorithm. Even when $K=500$, the probabilistic approach successfully retrieves all eight inequalities for about 90\% of the time, with an average running time of around 13 seconds per iteration, and the approach is able to recover seven of the eight restrictions for the remaining 10\% of the iterations. Thus, our experiment suggests that even for moderate values of $K$, the probabilistic approach is able to retrieve, with probability close to 1, all of the inequality restrictions that the model places on the CCPs, while using a running time that is much smaller than that of Benson's algorithm.\footnote{While it would be ideal to have a theory that guides the selection of the sampling distribution $\mu$ and determines the appropriate number of objective vectors $K$ needed to ensure that the probabilistic approach recovers all inequalities with a probability greater than $1-\delta$, for a given value of $\delta\in (0,1)$, such an analysis is beyond the scope of this paper and is left for future research.}

\begin{center}
\begin{adjustbox}{max width = \linewidth}
\begin{threeparttable}

\caption{Simulation output}

\vspace{-1cm}\setlength\tabcolsep{10.pt} 
\renewcommand{\arraystretch}{0.8} 

\input{simulation.tex}
\begin{tablenotes}
      \small
      \item The running time listed in the last column represents the running time for all 100 iterations. This includes the time to compute the $A$ and $R$ matrices used in the LPs. This experiment was conducted on the same machine used to generate the eight inequalities that precede Theorem \ref{thmPPD1}, and the running times in the table (divided by 100) can be directly compared to  that of Benson's algorithm (1800 seconds).
    \end{tablenotes}

\label{table3}
\end{threeparttable}
\end{adjustbox}
\end{center}
\end{remark}

\end{example}

\begin{example}[\textbf{\cite{KPT}} with two lags ]
\label{exampDyn2}

I now consider an AR(2) extension of the model in Example \ref{exampDyn1}. Specifically, I generalize the model in \ref{exampDyn1} to include an additional lag as follows:

\begin{equation}
\label{eqnAR21}
Y_{ti}=\mathbbm{1}\{X_{it}'\beta+\gamma_1 Y_{i,t-1}+\gamma_2 Y_{i,t-2}-\lambda_i-\varepsilon_{ti}> 0\},
\end{equation}

where $\gamma_1$ and $\gamma_2$ denote respectively the coefficients on the first and second lags. The goal of the analysis is now to characterize the identified set of the common parameters $\beta$, $\gamma_1$, and $\gamma_2$, denoted by $\theta=(\beta,\gamma_1,\gamma_2)$.

I consider a setting with three time periods, i.e., $t\in[3]$, and when $t=1$, I assume that the pair $(Y_0,Y_{-1})$ is observed and represents the initial condition. In contrast to Example \ref{exampDyn1}, I replace the unconditional stationarity restriction on the shocks with the following conditional (on the initial condition) stationarity restriction: For $s,t\in[3]$

\begin{equation}
\label{eqnAR22}
\varepsilon_{ti}|Y_{0i},Y_{-1i},X_i,\lambda_i \sim \varepsilon_{si}|Y_{0i},Y_{-1i},X_i,\lambda_i ,
\end{equation}

where $X_i=(X_{1i}',X_{2i}',X_{3i}')'$ is the vector of covariates for all three time periods.\par

Below, I extend the constructions of Section \ref{sectionCond} to the current setting with two lags, and use it to generate the conditional moment inequalities that characterize the identified set. Given a parameter value $\theta$, a covariate value $x$, and a value $y=(y_0,y_{-1})$ of the initial condition, let ${\cal P}(y,x,\theta)$ denote the set of all CCP vectors that are consistent with the model \ref{eqnAR21} and the stochastic restriction \ref{eqnAR22} at the covariate value $x$ and the initial condition $y$, when the true parameter is equal to $\theta$. For $t\in [3]$, let $v_t:=x_t'\beta$. As in Section \ref{sectionCond} the sets ${\cal P}(y,x,\theta)$ are polytopes, and I now show how to construct their "discretization". 

The first step in the discretization consists of constructing the patches, which as in Section \ref{sectionCond} are sets of values of realizations of the shocks that induce an agent to make a determined choice in each time period $t\in[3]$ and in each potential state. Here, the "state" in period t, denoted $s_t$, simply represents the choices that were made in the preceding two time periods. At time 1, since we are conditioning on the initial condition, the unique state is given by the initial condition; at time 2, there are only 2 potential states since the choice made in period 0 is determined by the initial condition; there are four states in period 3, corresponding to all possible choices made in period 1 and 2. Thus, a patch can be encoded by a vector $f=(d_1,\cdots,d_7)\in \{0,1\}^7$, where the $d_1$ represents the choice that the agent makes in period 1 when the unique state is given by $s_1=(y_0,y_{-1})$ (i.e., the state is equal to the initial condition), the second and third entries of $f$ represent respectively the choices made in period 2 when the state is given by $s_2=(0,y_0)$ and $s_2=(1,y_0)$ (for interpretation, the state $s_2=(1,y_0)$ indicates that alternative 1 was chosen in period 1 and alternative $y_0$ was chosen in period 0). Finally, the fourth through seventh entry of $f$ represent the choices made in period 3 when the states are respectively given by $s_3=(0,0)$, $s_3=(0,1)$, $s_3=(1,0)$ and $s_3=(1,1)$. The set of all patches $F$ then consists of all vectors $f=(d_1,\cdots,d_7)\in \{0,1\}^7$ such that the following system of inequalities in $\varepsilon$ admits a solution:
\begin{equation}
\begin{aligned}
\label{eqnAR23}
(2d_1-1)(\varepsilon-(v_1+\gamma_1y_0+\gamma_2y_{-1}))&<0\\
(2d_2-1)(\varepsilon-(v_2+\gamma_2y_0))&<0\\
(2d_3-1)(\varepsilon-(v_2+\gamma_1+\gamma_2y_0))&<0\\
(2d_4-1)(\varepsilon-v_3)&<0\\
(2d_5-1)(\varepsilon-(v_3+\gamma_2))&<0\\
(2d_6-1)(\varepsilon-(v_3+\gamma_1))&<0\\
(2d_7-1)(\varepsilon-(v_3+\gamma_1+\gamma_2))&<0.
\end{aligned}
\end{equation}
As in \ref{eqnKPT1}, checking whether a solution to the latter system exists is equivalent to checking whether 7 intervals (one for each inequality in \ref{eqnAR23}) have non-empty intersection. This can easily be done by solving a LP as in \ref{eqnKPT2}. Moreover, it follows from \ref{eqnAR23} that the set of patches $F$ is uniquely determined by the relative ordering between the quantities $v_1+\gamma_1y_0+\gamma_2 y_{-1}$, $v_2+\gamma_2y_0$, $v_2+\gamma_1+\gamma_2 y_0$, $v_3$, $v_3+\gamma_2$, $v_3+\gamma_1$ and $v_3+\gamma_1+\gamma_2$.

  Let the chosen values of $x$, $y$ and $\theta$  be such that 
\begin{equation}
\label{eqnAR25}
v_1=0, \ v_2=4, \ v_3=2, \ \gamma_1=3, \ \gamma_2=-4 \text{ and } y=(1,1). 
\end{equation}
This particular choice leads to the ordering:
\begin{equation}
\label{eqnAR24}
v_3+\gamma_2<v_1+\gamma_1y_0+\gamma_2y_{-1}<v_2+\gamma_2y_0<v_3+\gamma_1+\gamma_2<v_3<v_2+\gamma_1+\gamma_2 y_0 <v_3+\gamma_1,
\end{equation}
and the inequalities that I generate below represent all the restrictions that the model (\ref{eqnAR21} and \ref{eqnAR22}) places on the set of CCPs ${\cal P}(y,x,\theta)$ whenever $y$, $x$ and $\theta$ are chosen such that  \ref{eqnAR24} holds. In order to generate all the inequalities that characterize the identified set of the parameter $\theta$, the construction that I discuss below can be repeated for alternative orderings of the quantities in \ref{eqnAR24}. Given the values in \ref{eqnAR25}, the set $F$ of all patches can easily be obtained by solving \ref{eqnAR23} for all potential vector $f\in\{0,1\}^7$; doing so leads to only 8 "non-empty" patches (i.e., |F|=8). As in Section \ref{sectionCond}, this set of patches can then be used to partition $\mathbb{R}^3$, the set of all possible realizations of $\varepsilon=(\varepsilon_1,\varepsilon_2,\varepsilon_3)'$ (the vector of shocks in periods 1 through 3), into "rectangular regions" formed by taking the Cartesian product of elements of $F$. The resulting partition ${\cal R}$ has $8^3=512$ elements/regions. 

The matrix $A$ is a 8-by-512 array with each row indexed by a triple $\textbf{d}=(d,d',d'')\in\{0,1\}^3$ and each column indexed by a region $\textbf{R}=f_1\times f_2\times f_3\in{\cal R}$ such that the entry in the $\textbf{d}^{th}$ row and $\textbf{R}^{th}$ column is equal to 1 (and equal to zero otherwise) if and only if realizations of the shocks in the region $\textbf{R}$ induce the agent to respectively choose alternatives $d$, $d'$ and $d''$ in periods 1, 2 and 3, when the initial condition is given by $y=(1,1)$. That is, the entry in the $\textbf{d}^{th}$ row and $\textbf{R}^{th}$ column of $A$ is equal to 1 iff:
\[
d=f_1(1), \ d'=f_2(2+d), \text{ and } d''=f_3(4+d+2d')
\]
where for $j\in[7]$ and $t\in [3]$, $f_t(j)$ represents the $j^{(th)}$ entry of the patch $f_t$. I now discuss the construction of the matrix $R_S$ which enforces the conditional stationarity restriction on the discretized local model. Let $B_1$ and $B_2$ be two 8-by-512 matrices. Each row of $B_i$, for $i\in[2]$, indexes a patch $f\in F$ (recall that $|F|=8$ for the particular choice in \ref{eqnAR25}) and each column of $B_i$ indexes a region $\textbf{R}=f_1\times f_2 \times f_3 \in {\cal R}$. Let the entries in the $f^{th}$ row and $\textbf{R}^{th}$ column of $B_1$ and $B_2$ be respectively given by
\[
B_1(f,\textbf{R})=\mathbbm{1}\{f_1=f\}-\mathbbm{1}\{f_3=f\},
\]
and 
\[
B_2(f,\textbf{R})=\mathbbm{1}\{f_2=f\}-\mathbbm{1}\{f_3=f\}.
\]
Essentially, the matrix $B_1$ enforces the conditional stationarity stochastic restriction \ref{eqnAR22} between periods 1 and 3, and the matrix $B_2$ enforces the same restriction between periods 2 and 3. The matrix $R_S$ is then obtained by vertically concatenating the matrices $B_1$ and $B_2$, i.e.,
\[
R_S=(B_1',B_2')'.
\]

The set of all inequality restrictions that the model places on the CCP vectors in ${\cal P}(y,x,\theta)$ is then simply the set of solutions of the MOLP \ref{eqnMolp}, with the $A$ and $R_S$ matrices that we constructed in the preceding paragraph. Solving this MOLP (using Benson's algorithm) yields the following solution set: \footnote{The running time of the whole procedure (from the construction of the patches to solving for the MOLP) takes less than two seconds on the same machine used to generate the outputs of Table \ref{Table1} and \ref{Table2}.}

\[
\bordermatrix{& p_{000}&p_{001}&p_{010}&p_{011}&p_{100}&p_{101}&p_{110}&p_{111}\cr
&0&-1&1&0&-1&-1&0&-1\cr
&-1&-1&1&0&-1&-1&0&0\cr
&0&0&-1&-1&1&1&0&0\cr
&0&-1&0&-1&0&0&1&0
}\qquad
\]
where for $d_1,d_2,d_3 \in\{0,1\}$, $p_{d_1d_2d_3}$ represents the entry of the CCP vector equal to $P(Y_1=d_1, Y_2=d_2, Y_3=d_3|X=x,Y_0=1,Y_{-1}=1)$.
The first row corresponds to the inequality
\[
p_{010}\leq p_{001}+p_{100}+p_{101}+p_{111}
\]
which, after adding $p_{110}$ to both sides, can be rewritten as
\begin{equation}
\label{eqnAR26}
\begin{aligned}
P(Y_2=1,Y_3=0| x,Y_0=1,Y_{-1}=1)&\leq P(Y_1=0,Y_2=0,Y_3=1| x,Y_0=1,Y_{-1}=1)\\
&+P(Y_1=1| x,Y_0=1,Y_{-1}=1).
\end{aligned}
\end{equation}
The second row corresponds to the inequality
\begin{equation}
\label{eqnAR27}
P(Y_1=0,Y_2=1,Y_3=0| x,Y_0=1,Y_{-1}=1)\leq P(Y_2=0| x,Y_0=1,Y_{-1}=1).
\end{equation}
The third row (after adding $p_{001}+p_{000}$ to both sides) corresponds to the inequality
\begin{equation}
\label{eqnAR28}
P(Y_2=0| x,Y_0=1,Y_{-1}=1)\leq P(Y_1=0| x,Y_0=1,Y_{-1}=1).
\end{equation}
The fourth row (after adding $p_{010}+p_{000}$ to both sides) corresponds to the inequality
\begin{equation}
\label{eqnAR29}
P(\{Y_1=0\}\cup \{Y_2=1\},Y_3=0| x,Y_0=1,Y_{-1}=1)\leq P(Y_1=0| x,Y_0=1,Y_{-1}=1).
\end{equation}

Analytical proofs of inequalities \ref{eqnAR26}--\ref{eqnAR29} are provided in Section \ref{ineqDetails} of the Appendix. I provide a proof of Inequalities \ref{eqnAR27}--\ref{eqnAR29}, as well as their generalization, in Proposition \ref{propAR2P1} of the Appendix; these proofs are relatively simple, and follow arguments similar to those in \cite{KPT}. However, establishing inequality \ref{eqnAR26} is more challenging and relies on an intricate use of the stationarity assumption \footnote{Of all the inequalities that were generated by the algorithm in this paper, Inequality \ref{eqnAR26} was the most difficult one to verify analytically.}. As such, in the absence of an algorithm like the one presented in this paper, it may be difficult to guess an inequality like \ref{eqnAR26}. Its proof is provided in Proposition \ref{propAR2P2} of the Appendix.\par

The inequalities \ref{eqnAR26}--\ref{eqnAR29} generated by the algorithm represent some, but not all, of the inequalities that characterize the identified set. In particular, they are valid only when the ordering \ref{eqnAR24} holds. To determine \textit{all} of the inequalities that characterize the identified set, it can be helpful to prove analytically the inequalities generated by the algorithm for a few orderings of the quantities in \ref{eqnAR24}. By doing so, one can then guess and prove their generalization to alternative orderings of the quantities in \ref{eqnAR24} (see Proposition \ref{propAR2P1} for instance). These predictions can then be compared against the inequality restrictions that the algorithm generates for alternative orderings of \ref{eqnAR24}. If there is a mismatch (i.e., the algorithm generates some new inequalities that are not implied by the predicted inequalities) for some ordering, the new inequalities can be proven analytically, and the future predictions can be refined by taking into consideration these new inequalities.
\end{example}

\section{Conclusion}
\label{sectionConclusion}
This paper has presented an algorithm for generating the conditional moment inequality restrictions that delineate the sharp identified set of the common parameter for various multinomial choice models. The algorithm is compatible with a broad range of stochastic restrictions, and I illustrated through numerous examples that it can recover established results and generate new ones. \par 
One drawback of the algorithm is that it may not be computationally feasible for large-scale models, particularly when $D^T$ is "large". To overcome this limitation, future research could further explore the specific properties of MOLPs that are linked to multinomial choice models, in order to devise more practical algorithms for larger problems, as I have attempted to do in Section \ref{sectionCuttingPlane}. And although my focus in this paper was on panel multinomial choice models, the methodology introduced here should be applicable to other discrete outcome models, such as the bivariate models and the panel data games discussed in \cite{HDP}, as well as to cross-sectional settings. The key features needed for the approach to succeed are that the local models are polytopes and that the "local model map" $(x,\theta)\rightarrow {\cal P}(x,\theta)$ assumes a finite number of values (as shown in Proposition \ref{propPatch}). As long as these conditions hold, the approach pursued in this paper, with the requisite modifications, ought to be viable. Such extensions are left for future research.
\section{Appendix A}

\begin{proof}[\textbf{Proof of Proposition \ref{prop1}}]
Let $p$ be a given CCP in ${\cal P}_i(x,\theta)$, and let $q$ (satisfying $R_i q=0$) be a discrete probability vector that rationalizes $p$. For each patch $f\in F$, let $g_f$ be a density on $\mathbb{R}^D$ that is supported on a compact subset of $f$ (such densities exist since the patches $f$ are open sets). Then the density $g$ defined on $\mathbb{R}^{DT}$ by
\[
g(x_1,\cdots,x_T):=\sum_{f_1,\cdots,f_T\in F} q_{f_1 \times \cdots, \times f_T} \prod_{t\in[T]} g_{f_t}(x_t)
\]
puts the same mass as $q$ in each region of $R\in {\cal R}$ (hence it rationalizes $p$ and is thus observationally equivalent to $q$), and it can be checked that g is stationary (resp. exchangeable) if q is stationary (resp. exchangeable).
\end{proof}
\begin{proof}[\textbf{Proof of Proposition \ref{propPatch}}]
For fixed $D$ and $T$, and $i\in\{E,S\}$, there are only finitely many possible matrices $A$ and $R_i$ that we can obtain from the discretization step from Section \ref{sectionDiscretization} (the maximum number of columns of the $A$ and $R$ matrices is upper bounded by $D^{T^2}$, and their entries are in $\{0,\pm 1\}$), and if the $A$ and $R$ matrices are the same at two covariate-parameter pairs, then DDCP is the same at these two points, and thus the corresponding MOLPs coincide.\par
Let $\{(A^{(k)},R_i^{(k)})\}_{k\in [m]}$ denote all possible configurations of the matrices $A$ and $R_i$. For $k\in[m]$, let $O_k$ denote the set of all $(x,\theta)\in {\cal X}\times \Theta$ such that the DCP at $(x,\theta)$ is given by $\{p=A^{(k)}q \ | \ R_i^{(k)}q=0, \ \text{and } q\geq 0\}$. And let $I_k$ denote the solution of the MOLP \ref{eqnMolp} when the $A$ and $R_i$ matrices are equal to $A^{(k)}$ and $R_i^{(k)}$. Then the statement of the proposition follows with $m$, $O_k$ and $I_k$ as above.\par
Note that although the suggested upper bound on the quantity $m$ from the preceding argument can be quite large, the actual value of $m$ is much smaller, and by exploiting the symmetries of the problem we can further reduce the number of MOLPs that need to be solved ($m=3$ in the example of Section \ref{sectionHeuristics} where $D=T=2$, and $m=5$ in Examples \ref{exampPPS} and \ref{exampPPE} where $D=4$ and $T=2$).
\end{proof}
\begin{proof}[\textbf{Proof of Proposition \ref{propImp}}]
I prove the proposition here for a two periods binary choice model (setting of Section \ref{sectionHeuristics}) under the conditional IID assumption, but the same argument can be extended (with some additional work) to both dynamic and static models with $D\geq 2$ alternatives, $T\geq 2$ time periods.\par
Let $v_t:=x_{t}'\theta$, for $t\in[2]$, and assume that $v_1>v_2$ (a similar argument holds if $v_1=v_2$ or $v_1<v_2$). Each $p \in {\cal P}(x,\theta)$, is of the form:
\begin{equation}
\label{eqn002}
p=\begin{pmatrix}
           p_{00} \\
           p_{01} \\
           p_{10}\\
          p_{11}         
         \end{pmatrix}
         =\int_{\mathbb{R}} \begin{pmatrix}
           (1-r_\lambda)(1-s_{\lambda}) \\
           (1-r_{\lambda})s_{\lambda} \\
            r_{\lambda}(1-s_{\lambda})\\
            r_{\lambda}s_{\lambda}        
         \end{pmatrix}
         d\nu(\lambda)
\end{equation}

where $\nu$ is a conditional distribution of the fixed effects given $x$, $r_{\lambda}=P(\epsilon_1 +\lambda < v_1|x,\lambda)$, $s_{\lambda}=P(\epsilon_2 +\lambda < v_2|x,\lambda)$. Since $\epsilon_1$ and $\epsilon_2$ are IID conditional on $x$ and $\lambda$ and $v_1>v_2$, we have $r_{\lambda}\geq s_{\lambda}$ for all values of $\lambda$ in the support of $\nu$. Equation \ref{eqn002} implies that ${\cal P}(x,\theta)$ is convex, and is the convex hull of the compact set
\begin{equation}
\label{eqn003}
{\cal S}=\{p=((1-r)(1-s), s(1-r), r(1-s), rs)^T \ | \ r\in [0,1] \ \text{and } 0\leq s\leq r \}.
\end{equation}
which is closed (by Caratheodory's Theorem and the compactness of ${\cal S}$, each element of ${\cal P}(x,\theta)$ can be written as a convex combination of at most 4 elements of ${\cal S}$, and it is without loss of generality to assume that the conditional distribution of the fixed effects given $x$ is supported on at most 4 support points\footnote{The set ${\cal S}$ is a subset of $\mathbb{R}^4$ and  its elements satisfy the restriction $\mathbbm{1}^Tp=1$. Hence by Caratheodory's Theorem, every convex combination of elements of ${\cal S}$ can be represented as a convex combination of at most 4 elements of ${\cal S}$.}\footnote{The set ${\cal C}_1$ obtained by mixing elements of ${\cal S}$ using a general probability measure for the fixed effects is a closed and convex set. The set ${\cal C}_2$ obtained by mixing elements of ${\cal S}$ using probability measures for the fixed effects supported on at most 4 points is convex and dense in ${\cal C}_1$, and it is closed since ${\cal S}$ is compact. Thus ${\cal C}_1={\cal C}_2$ }).  By duality, a point $p\in \mathbb{R}^4$ belongs to ${\cal P}(x,\theta)$ if and only if $w'p\leq \mu_{{\cal P}}(w)$ for all $w \in \mathbb{R}^{4}$. where $\mu_{{\cal P}}$ denotes the support function of the set ${\cal P}(x,\theta)$ and is defined by 
\[
\mu_{{\cal P}}(w)=\sup\{w'p\ | \ p \in {\cal P}(x,\theta)\}=\max\{w'p\ | \ p \in {\cal S}\}.
\]
As the support function of a polytope is piecewise-linear, its derivative is constant on any open set on which it is differentiable. To show that ${\cal P}(x,\theta)$ is not a polytope, it suffices to find an open set on which $\mu_{{\cal P}}$ is differentiable with non-constant derivative (see \cite{HW}, Theorem 2.10). \par
For $w \in \mathbb{R}^4$ an element of the open set $\{w=(w_1,w_2,w_3,w_4)^T \in \mathbb{R}^4\ | \ w_1>w_3>w_4, \ w_3+w_4>2 w_1, \ \text{and } w_2+w_3>w_1+w_4 \}$, solving the maximization problem defining the support function yields
\[
\mu_{{\cal P}}(w)=w_1+ \frac{(w_2+w_3-2w_1)^2}{4(w_2+w_3-w_1-w_4)}
\] 
which is differentiable on the open set under consideration, and is clearly non-linear.
\end{proof}


\begin{proof}[\textbf{Proof of Theorem \ref{thm2}}]
As $supp(X)={x_0}$, let ${\cal P}(\theta):={\cal P}(x_0,\theta)$.  Assume, for a contradiction, that the sharp identified set has a representation of the type \ref{eqn_2}. Then for any observable CCP vector $p_0$ at $x_0$, the sharp identified set $\Theta(p_0)$ (possibly empty, if $p_0$ is not consistent with the model) can be written as:
\begin{equation}
\label{eqnt21}
\Theta(p_0)=\{\theta \in \Theta \ | \ \text{For all } k \in [M], \ \text{s.t. }\theta \in \tilde{\theta}_k,  \ \text{we have } \alpha_k ^T p_0\leq \beta_k \}
\end{equation}
where
\[
\tilde{\Theta}_k:=\{\theta \in \Theta \ | \ (x_0,\theta)\in S_k\},
\]
and $S_k$ is as in \ref{eqn_2}, and some of the $\tilde{\Theta}_k$ possibly empty. Note that since $\cup_{k\in [M]}S_k={\cal X}\times \Theta$, we have 
\[
\bigcup_{k\in [M]} \tilde{\Theta}_k=\Theta.
\]
Let $\theta_0$ be as in the statement of Theorem \ref{thm2}, and let $\Delta=\{k \in [M]\ | \ \theta_0\in \tilde{\Theta}_k\}$. Then $\forall p \in {\cal P}(\theta_0)$, Equation \ref{eqnt21} implies that $\theta_0\in \Theta(p)$ and
\[
{\cal P}(\theta_0)\subseteq {\cal Q}:=\{q \ | \ \alpha_k^T q \leq \beta_k, \ \forall k \in \Delta\}.
\]
As ${\cal Q}$ is a polytope and ${\cal P}(\theta_0)$ is not (by assumption), there exists $\tilde{p}\in {\cal Q}\backslash{\cal P}(\theta_0)$. But then $\theta_0\in \Theta(\tilde{p})$ (by Equation \ref{eqnt21}, the definition of $\Delta$ and ${\cal Q}$), and by the definition of the identified set $\tilde{p}$, $\tilde{p}$ should be a CCP vector consistent with the model at $(x_0,\theta_0)$. But this contradicts $\tilde{q}\notin {\cal P}(\theta_0)$.
\end{proof}

\begin{proof}[\textbf{Proof of Theorem \ref{thmE}}]
Let $x$, $\theta$ and $D\geq 2$ be fixed, and assume (w.l.o.g.) that 
\begin{equation}
\label{eqn0}
\Delta v_1\geq \Delta v_2\geq\cdots\geq\Delta v_D,
\end{equation}
which can be made to hold by relabeling the alternatives. Without loss of generality, we can assume that, $\forall d\in[D]$, $v_{d1}=0$  and $v_{d2}=\Delta v_d$ (see Footnote \ref{foot19}). Let the sets ${\cal \varepsilon}_{d,t}$, for $d\in [D]$ and $t\in[2]$, be defined as in Equation \ref{eqnPart}. Then for $A=\bigcup_{i=1}^m U_{k_i}\times L_{k'_i}$, we have
\[
P((y_1,y_2)\in A)=P\left((\zeta_1,\zeta_2) \in \bigcup_{i=1}^m \left[(\bigcup_{d\leq k_i} {\cal \varepsilon}_{d,1})\times(\bigcup_{d'\geq k'_i} {\cal \varepsilon}_{d',2})\right]\right)
\]
where $\zeta_1,zeta_2\in \mathbb{R}^D$ are the composite errors (i.e, $\zeta_1=\lambda+\epsilon_1$), and where we have dropped the conditioning on $x$ for notational simplicity. \par
I now prove the following two inclusions
\begin{equation}
\label{eqnincl1}
\bigcup_{d\leq k_i} {\cal \varepsilon}_{d,1} \subseteq \bigcup_{d\leq k_i} {\cal \varepsilon}_{d,2},
\end{equation}
and
\begin{equation}
\label{eqnincl2}
\bigcup_{d'\geq k'_i} {\cal \varepsilon}_{d',2} \subseteq \bigcup_{d'\geq k'_i} {\cal \varepsilon}_{d',1},
\end{equation}
from which it will follow that 
\begin{equation}
\label{eqnincl3}
 \bigcup_{i=1}^m \left[(\bigcup_{d\leq k_i} {\cal \varepsilon}_{d,1})\times(\bigcup_{d'\geq k'_i} {\cal \varepsilon}_{d',2})\right]\subseteq  \bigcup_{i=1}^m \left[(\bigcup_{d\leq k_i} {\cal \varepsilon}_{d,2})\times(\bigcup_{d'\geq k'_i} {\cal \varepsilon}_{d',1})\right].
\end{equation}
Combining inclusion \ref{eqnincl3} with the fact that $(\zeta_1,\zeta_2)$ is exchangeable then yields
\[
\begin{split}
P\left((\zeta_1,\zeta_2) \in \bigcup_{i=1}^m \left[\bigcup_{d\leq k_i} {\cal \varepsilon}_{d,1} \times \bigcup_{d'\geq k'_i} {\cal \varepsilon}_{d',2} \right]\right) & \leq P\left((\zeta_1,\zeta_2) \in \bigcup_{i=1}^m \left[(\bigcup_{d\leq k_i} {\cal \varepsilon}_{d,2})\times(\bigcup_{d'\geq k'_i} {\cal \varepsilon}_{d',1})\right]\right)\\
&= P\left((\zeta_2,\zeta_1) \in \bigcup_{i=1}^m \left[(\bigcup_{d\leq k_i} {\cal \varepsilon}_{d,2})\times(\bigcup_{d'\geq k'_i} {\cal \varepsilon}_{d',1})\right]\right)\\
&=P\left((y_1,y_2)\in \bigcup_{i=1}^m L_{k'_i} \times  U_{k_i} \right)\\
&=P\left((y_1,y_2)\in Per(A)\right).
\end{split}
\]
Thus, it remains to show Inclusions \ref{eqnincl1} and \ref{eqnincl2}. To see \ref{eqnincl1}, note that we have (using $v_{d1}=0$ for all $d\in[D]$)
\[
\bigcup_{d\leq k_i} {\cal \varepsilon}_{d,1}=\{\zeta \in \mathbb{R}^D \ | \ \max_{d\leq k_i}\zeta_d > \max_{k_i<d'\leq D}\zeta_{d'}\}.
\]
Inequalities \ref{eqn0} implies that 
\[
\min_{d\leq k_i} v_{d2} \geq \max_{k_i<d'\leq D}v_{d'2}
\]
 (recall that $v_{d2}=\Delta v_d$ for all $d\in[D]$), and thus whenever $\max_{d\leq k_i}\zeta_d > \max_{k_i<d'\leq D}\zeta_{d'}$, we must have 
\[
\max_{d\leq k_i}\zeta_d+v_{d2} > \max_{k_i<d'\leq D}\zeta_{d'}+v_{d'2}.
\]
Combining the preceding observations yields
\[
\bigcup_{d\leq k_i} {\cal \varepsilon}_{d,1} \subseteq \{\zeta \in \mathbb{R}^D \ | \ \max_{d\leq k_i}\zeta_d+v_{d2} > \max_{k_i<d'\leq D}\zeta_{d'}+v_{d'2}\}=\bigcup_{d\leq k_i} {\cal \varepsilon}_{d,2},
\]
and this establishes Inclusion \ref{eqnincl1}. Inclusion \ref{eqnincl2} follows by a similar argument; we have
\[
\bigcup_{d'\geq k'_i} {\cal \varepsilon}_{d',2}=\{\zeta \in \mathbb{R}^D \ | \ \max_{d< k'_i}v_{d2}+\zeta_d < \max_{k'_i\leq d'\leq D}v_{d'2}+\zeta_{d'}\}.
\]
Inequalities \ref{eqn0} imply that
\[
\max_{d< k'_i}-v_{d2} \leq \max_{k'_i\leq d'\leq D}-v_{d'2},
\]
which when combined with the preceding equality yields
\[
\bigcup_{d'\geq k'_i} {\cal \varepsilon}_{d',2}\subseteq \{\zeta \in \mathbb{R}^D \ | \ \max_{d< k'_i} \zeta_d < \max_{k'_i\leq d'\leq D} \zeta_{d'}\}=\bigcup_{d'\geq k'_i} {\cal \varepsilon}_{d',1},
\]
and this establishes \ref{eqnincl2}.
\end{proof}


\subsubsection{Proof of Theorem \ref{thmPPD1} and omitted details of Remarks \ref{remPPD5} and \ref{remPPD3}}
\label{remDetails}
\begin{proof}[\textbf{Proof of Theorem \ref{thmPPD1}}]
As in Remark \ref{rem2}, it is without loss of generality to assume that the fixed effects are identically equal to zero, and I do so below. Fix $A\in {\cal A}(y_0,x,\theta)$, and let $E_A$ denote the event that appears on the left hand side of Inequality \ref{eqnthmPPD}:
\[
E_A:=\bigcup_{d \in [D]}\Bigl\{Y_1=d,Y_2\in \cup \bigl\{B \ | \ B\subseteq A, B \in {\cal L}(d_1,y_0,x,\theta_0)\bigl\}\Bigl\}.
\]
For each $d\in[D]$, let the set $B_{d|A}$$(\subset [D])$, possibly empty, be defined by
\[
B_{d|A}:=\cup \bigl\{B \ | \ B\subseteq A, B \in {\cal L}(d,y_0,x,\theta_0)\bigl\}
\]
and let the \textit{associated} event $E_{d|A}$ be defined by
\[
E_{d|A}:=\{Y_1=d, Y_2 \in B_{d|A}\}.
\]
Note that the events $E_{d|A}$ are pairwise disjoint, and their union is equal to $E_A$. Moreover, it follows from the definition of elements of ${\cal L}(d_1,y_0,x,\theta_0)$, that $B_{d|A}$ is itself a "lower set", in that it satisfies $ B_{d|A}\subsetneq [D]$ (since $A\subsetneq [D]$) and $\Delta(d',d,y_0,x,\theta)\leq \Delta(d'',d,y_0,x,\theta)$ for all $d'\in B_{d|A}$ and $d''\in [D]\backslash B_{d|A}$. I now show that the events $E_{d|A}$ are all included in the event $F_A:=\{\epsilon_2 \ | \max_{d\in A} v_{d1}+\gamma \mathbbm{1}\{d=y_0\}+\epsilon_{d2}>  \max_{d\in [D]\backslash A} v_{d1}+\gamma \mathbbm{1}\{d=y_0\}+\epsilon_{d2}\}$. Assuming the latter is true, it then follows that $P(E_A)\leq P(F_A)$, since the events $E_{d|A}$ are pairwise disjoint, and the stationarity  restriction then gives $P(F_A)=P(\{\epsilon_1 \ |  \max_{d\in A} v_{d1}+\gamma \mathbbm{1}\{d=y_0\}+\epsilon_{d1}>  \max_{d\in [D]\backslash A} v_{d1}+\gamma \mathbbm{1}\{d=y_0\}+\epsilon_{d1}\})=P(Y_1\in A)$, which yields Inequality \ref{eqnthmPPD}. Hence, it remains to show that $E_{d_1|A}\subseteq F_A$ for all $d_1\in [D]$. We have
\begin{align*}
E_{d_1|A} &\subseteq \bigl\{\epsilon_2 \ | \max_{d\in B_{d_1|A}} v_{d2}+\gamma \mathbbm{1}\{d=d_1\}+\epsilon_{d2}> \max_{d\in [D]\backslash B_{d_1|A}} v_{d2}+\gamma \mathbbm{1}\{d=d_1\}+\epsilon_{d2}\bigl\}\\
&\subseteq  \bigl\{\epsilon_2 \ | \max_{d\in B_{d_1|A}} v_{d2}+\gamma \mathbbm{1}\{d=d_1\}-\Delta(d,d_1,y_0,x,\theta)+\epsilon_{d2}>\\
& \max_{d\in [D]\backslash B_{d_1|A}} v_{d2}+\gamma \mathbbm{1}\{d=d_1\}-\Delta(d,d_1,y_0,x,\theta)+\epsilon_{d2}\bigl\}\\
&=\bigl\{\epsilon_2 \ | \max_{d\in B_{d_1|A}} v_{d1}+\gamma \mathbbm{1}\{d=y_0\}+\epsilon_{d2}> \max_{d\in [D]\backslash B_{d_1|A}} v_{d1}+\gamma \mathbbm{1}\{d=y_0\}+\epsilon_{d2}\bigl\}\\
&\subseteq \bigl\{\epsilon_2 \ | \max_{d\in A} v_{d1}+\gamma \mathbbm{1}\{d=y_0\}+\epsilon_{d2}> \max_{d\in [D]\backslash A} v_{d1}+\gamma \mathbbm{1}\{d=y_0\}+\epsilon_{d2}\bigl\}=F_A
\end{align*}
where the second inclusion follows since $B_{d_1|A}$ is a lower set, and where the last inclusion follows from $B_{d_1|A}\subseteq A$.
\end{proof}

\underline{\textbf{Omitted details of Remark \ref{remPPD5}}}
I show below that the inequalities of Theorem \ref{thmPPD1} and Theorem 3.5 of \cite{PP2} are complementary, under Assumption \ref{eqnPPSC1}. By the latter I mean that it is possible for the inequalities of Theorem \ref{thmPPD1} to be informative about the model's parameters while the inequalities \ref{eqnPPSC3} are not, and it is equally possible for the inequalities \ref{eqnPPSC3} to be informative about the model's parameters while those of Theorem \ref{thmPPD1} are not. I will illustrate this point using the application considered in \cite{PP2}.\par

 \cite{PP2} use the following dynamic discrete choice model to estimate the extend of state dependence in the choice of health insurance plans. The indirect utility of individual i for health insurance plan d ($d\in[D]$) at time $t\in[2]$ is modeled by 
\[
U_{dti}=-X_{dti}-\gamma \mathbbm{1}_{[Y_{(t-1)i}\neq d]}+\lambda_{di}+\epsilon_{dti},
\]
where $X_{dti}$ is a scalar that denotes the price of health insurance plan d at time t. The goal is to identify the state dependence parameter $\gamma$ (the switching cost) from choice data, when the only stochastic restriction on the shocks is given by Assumption \ref{eqnPPSC1}. As noted in Footnote \ref{footPPSC1}, the preceding specification of the indirect utilities is observationally equivalent to the following  one
\[
U_{dti}=-X_{dti}+\gamma \mathbbm{1}_{[Y_{(t-1)i}= d]}+\lambda_{di}+\epsilon_{dti}.
\]
Suppose that we know a priori that $\gamma\geq0$, and consider for simplicity a setting where $D=2$, and where the support of the covariates and initial conditions is degenerate. That is, assume that $X$ takes a single value $\bar{x}$, and the initial condition takes the single value $Y_0=1$. Then the conditional choice probability vector $\bar{p}:=\{P(Y_1=d,Y_2=d'|Y_0=1,X=\bar{x})\}_{d,d'\in\{0,1\}}$, $\bar{p}\in\mathbb{R}^4$, is identified from a random sample. For $d\in[2]$, let $\Delta \bar{x}_d:=\bar{x}_{d2}-\bar{x}_{d1}$ (i.e., the change in price of the $d^{th}$ alternative from period 1 to 2). Suppose that the observed vector of prices in period 1 and 2, $\bar{x}$, is such that 
\[
\Delta \bar{x}_2-\Delta \bar{x}_1>0.
\]
 I show below, by considering two cases, that depending on the value of $\bar{p}$, inequalities \ref{eqnthmPPD} rule out some values of the parameter $\gamma$ that inequalities \ref{eqnPPSC3} do not rule out, and vice-versa. \par

Using the definitions of the quantities $\Delta(d,d_1,y_0,x,\theta)$ given in the paragraph that precedes Theorem \ref{thmPPD1}, we have
\[
\Delta(1,1,1,\bar{x},\gamma)=(-\bar{x}_{12}+\gamma)-(\bar{x}_{11}+\gamma)=-\Delta \bar{x}_1 \quad \Delta(2,1,1,\bar{x},\gamma)=(-\bar{x}_{22})-(\bar{x}_{21})=-\Delta \bar{x}_2,
\]
and
\[
\Delta(1,2,1,\bar{x},\gamma)=(-\bar{x}_{12})-(\bar{x}_{11}+\gamma)=-\Delta \bar{x}_1-\gamma \quad \Delta(2,2,1,\bar{x},\gamma)=(-\bar{x}_{22}+\gamma)-(\bar{x}_{21})=-\Delta \bar{x}_2+\gamma.
\]
Note that since $\Delta \bar{x}_2-\Delta \bar{x}_1>0$, we have $\Delta(1,1,1,\bar{x},\gamma)>\Delta(2,1,1,\bar{x},\gamma)$ and ${\cal L}(1,1,\bar{x},\gamma)=\{\{2\}\}$ for all values of $\gamma\geq 0$, where ${\cal L}(1,1,\bar{x},\gamma)$ is as defined in Equation \ref{PPD1}. It thus follows that for all $\gamma\geq 0$, $\{1\}\in {\cal D}(1,\bar{x},\gamma)$, where ${\cal D}(1,\bar{x},\gamma)$ is as defined in Equation \ref{eqnPPSC2}, and the corresponding inequality from \ref{eqnPPSC3} yields
\begin{equation}
\label{eqnPPSCP1}
P(Y_2=1|Y_1=1,Y_0=1,X=\bar{x})\geq P(Y_1=1|Y_0=1,X=\bar{x}).
\end{equation}
As inequality \ref{eqnPPSCP1} is valid for all values $\gamma\geq0$, it must hold for the identified CCP vector $\bar{p}$ (unless the model is misspecified). Thus in both cases that I consider below, I assume that $\bar{p}$ satisfies inequality \ref{eqnPPSCP1}. Note that  $\Delta(1,2,1,\bar{x},\gamma)-\Delta(2,2,1,\bar{x},\gamma)=\Delta \bar{x}_2-\Delta \bar{x}_1 -2\gamma$; the two cases that I consider below depend on whether the latter quantity is positive or negative.

\underline{Case 1}. Suppose that $0\leq \gamma < (\Delta \bar{x}_2-\Delta \bar{x}_1 )/2$. Then $\Delta(1,2,1,\bar{x},\gamma)>\Delta(2,2,1,\bar{x},\gamma)$, ${\cal L}(2,1,\bar{x},\gamma)=\{\{2\}\}$, and we have ${\cal A}(1,\bar{x},\gamma)=\{\{2\}\}$ and ${\cal D}(1,\bar{x},\gamma)=\{\{1\}\}$, where ${\cal A}(1,\bar{x},\gamma)$ is as defined in equations \ref{PPD2}. In this case, inequality \ref{eqnthmPPD} then becomes
\begin{equation}
\label{eqnPPSCP2}
P(Y_2=2|Y_0=1,X=\bar{x})\leq P(Y_1=2|Y_0=1,X=\bar{x})
\end{equation}
and inequality \ref{eqnPPSC3} yields inequality \ref{eqnPPSCP1}. Suppose that $\bar{p}$ satisfies \ref{eqnPPSCP1} but does not satisfy inequality \ref{eqnPPSCP2}\footnote{It can easily be shown that there are probability vectors $p\in \mathbb{R}^4$ that satisfy inequality \ref{eqnPPSCP1} but violate inequality \ref{eqnPPSCP2}.} Then, for such $\bar{p}$, the inequalities in Theorem \ref{thmPPD1} will rule out values of the parameter $\gamma$ in the range $[0, (\Delta \bar{x}_2-\Delta \bar{x}_1 )/2)$, while the inequalities in \cite{PP2} will fail to do so. 

\underline{Case 2} Suppose now that $\gamma> (\Delta \bar{x}_2-\Delta \bar{x}_1 )/2$. Then $\Delta(1,2,1,\bar{x},\gamma)<\Delta(2,2,1,\bar{x},\gamma)$, ${\cal L}(2,1,\bar{x},\gamma)=\{\{1\}\}$, and we have ${\cal A}(1,\bar{x},\gamma)=\{\{1\},\{2\}\}$ and ${\cal D}(1,\bar{x},\gamma)=\{\{1\},\{\}2\}$. The inequalities \ref{eqnthmPPD} then yield
\begin{equation}
\label{eqnPPSCP3}
P(Y_2=2,Y_1=1|Y_0=1,X=\bar{x})\leq P(Y_1=2|Y_0=1,X=\bar{x})
\end{equation}
and
\begin{equation}
\label{eqnPPSCP4}
P(Y_2=1,Y_1=2|Y_0=1,X=\bar{x})\leq P(Y_1=1|Y_0=1,X=\bar{x}).
\end{equation}
And the inequalities \ref{eqnPPSC3} yield inequality \ref{eqnPPSCP1} and 
\begin{equation}
\label{eqnPPSCP5}
P(Y_2=2|Y_1=2,Y_0=1,X=\bar{x})\geq P(Y_1=2|Y_0=1,X=\bar{x}).
\end{equation}
Suppose that the identified CCP vector $\bar{p}$ satisfies inequalities \ref{eqnPPSCP1}, \ref{eqnPPSCP3} and \ref{eqnPPSCP4}, but violates inequality \ref{eqnPPSCP5}\footnote{It is straight forward to show that there exists such probability vectors $p\in\mathbb{R}^4$.}. Then, for such $\bar{p}$, the inequalities in \cite{PP2} rule out values of $\gamma$ in the range $((\Delta \bar{x}_2-\Delta \bar{x}_1 )/2,+\infty)$, while the inequalities in Theorem \ref{thmPPD1} fail to do so.

\underline{\textbf{Omitted details of Remark \ref{remPPD3}}}
I here show that the restrictions of Theorem \ref{thmPPD1} coincide with those of Theorem 2 of \cite{KPT} when $D=T=2$. In the latter setting, the model \ref{eqndyn} becomes: Alternative 2 is chosen at time $t\in[2]$ if and only if (and alternative 1 is chosen otherwise)
\[
X_{2ti}'\beta+\gamma \mathbbm{1}[Y_{t-1,i}=2]+\lambda_{2i}+\epsilon_{2ti} > X_{1ti}'\beta+\gamma \mathbbm{1}[Y_{t-1,i}=1]+\lambda_{1i}+\epsilon_{1ti}.
\]
After relabeling the alternatives to the set of values $\{0,1\}$, the latter is equivalent to

\begin{equation}
\label{eqnPPD3}
Y_t=\mathbbm{1}_{\{X_{1ti}'\beta+\gamma Y_{t-1,i}+\lambda_{1i}+\epsilon_{1ti} > X_{0ti}'\beta+\gamma (1-Y_{t-1,i})+\lambda_{0i}+\epsilon_{0ti}\}}.
\end{equation}
When $d'=0$, the "deterministic utility increments" $\Delta(d,d',y_0,x,\theta)$ are given by:
\[
\Delta(0,0,y_0,x,\theta)=v_{02}-v_{01}+\gamma-\gamma(1- y_0) \quad \text{and} \quad \Delta(1,0,y_0,x,\theta)=v_{12}-v_{11}-\gamma y_0,
\]
and the set ${\cal L}(0,y_0,x,\theta)$ is determined by the order relation between $\Delta(0,0,y_0,x,\theta)$ and $\Delta(1,0,y_0,x,\theta)$. Specifically, ${\cal L}(0,y_0,x,\theta)=\{\{0\}\}$ (and is equal to $\{\{1\}\}$ otherwise \footnote{Here I have ignored ties for simplicity. Note that if $\Delta(0,0,y_0,x,\theta)=\Delta(1,0,y_0,x,\theta)$ then ${\cal L}(0,y_0,x,\theta)=\{\{0\},\{1\}\}$.}) if and only if
\begin{equation}
\label{eqnPPD4}
a_2-a_1>\tilde{\gamma} y_0
\end{equation}
where 
\begin{equation}
\label{eqnPPD0}
\tilde{\gamma}=2\gamma \quad \text{and}\quad a_t:=v_{1t}-v_{0t}
\end{equation}
 for $t\in[2]$. Similarly, when $d'=1$, the deterministic utility increments are given by
\[
\Delta(0,1,y_0,x,\theta)=v_{02}-v_{01}-\gamma(1- y_0) \quad \text{and} \quad \Delta(1,1,y_0,x,\theta)=v_{12}-v_{11}+\gamma-\gamma y_0,
\]
and, ignoring ties, we get  ${\cal L}(1,y_0,x,\theta)=\{\{0\}\}$ (and is equal to $\{\{1\}\}$ otherwise) if and only if
\begin{equation}
\label{eqnPPD5}
a_2-a_1+\tilde{\gamma} (1-y_0)>0.
\end{equation}

Hence, ignoring potential ties between the deterministic utility differences, the inequalities of Theorem \ref{thmPPD1}, in the $D=T=2$ setting, are explicitly given by:

\begin{enumerate}[a.]
\item If inequalities \ref{eqnPPD4} and \ref{eqnPPD5} are true, i.e., $a_2-a_1>\tilde{\gamma} y_0$ and $a_2-a_1+\tilde{\gamma} (1-y_0)>0$ (which can easily be shown to be equivalent to $a_2-a_1+\min\{0,\tilde{\gamma}\}>\tilde{\gamma} y_0$), then we have ${\cal L}(0,y_0,x,\theta)={\cal L}(1,y_0,x,\theta)= {\cal A}(y_0,x,\theta_0)=\{\{0\}\}$, and inequality \ref{eqnthmPPD} becomes
\begin{equation}
\label{eqnPPD6}
 P(Y_2=0|y_0,x) \leq P(Y_1=0|y_0,x).  
\end{equation}

\item If inequalities \ref{eqnPPD4} and \ref{eqnPPD5} are false, i.e., $a_2-a_1<\tilde{\gamma} y_0$ and $a_2-a_1+\tilde{\gamma} (1-y_0)<0$ (which is equivalent to $a_2-a_1+\max\{0,\tilde{\gamma}\}<\tilde{\gamma} y_0$), then we have ${\cal L}(0,y_0,x,\theta)={\cal L}(1,y_0,x,\theta)={\cal A}(y_0,x,\theta_0)=\{\{1\}\}$, and inequality \ref{eqnthmPPD} becomes
\begin{equation}
\label{eqnPPD7}
 P(Y_2=1|y_0,x) \leq P(Y_1=1|y_0,x).  
\end{equation}

\item If inequality \ref{eqnPPD4} is true and inequality \ref{eqnPPD5} is false, i.e., $a_2-a_1>\tilde{\gamma} y_0$ and $a_2-a_1+\tilde{\gamma} (1-y_0)<0$, then we have ${\cal L}(0,y_0,x,\theta)=\{\{0\}\}$, ${\cal L}(1,y_0,x,\theta)=\{\{1\}\}$, $ {\cal A}(y_0,x,\theta_0)=\{\{0\},\{1\}\}$, and Theorem \ref{thmPPD1} yields the following two inequalities:
\begin{equation}
\label{eqnPPD8}
\begin{aligned}
P(Y_1=0,Y_2=0|y_0,x)&\leq P(Y_1=0|y_0,x)\\
&\text{and}\\
P(Y_1=1,Y_2=1|y_0,x)&\leq P(Y_1=1|y_0,x).
\end{aligned}
\end{equation}

\item If inequality \ref{eqnPPD4} is false and inequality \ref{eqnPPD5} is true, i.e., $a_2-a_1<\tilde{\gamma} y_0$ and $a_2-a_1+\tilde{\gamma} (1-y_0)>0$, then we have ${\cal L}(0,y_0,x,\theta)=\{\{1\}\}$, ${\cal L}(1,y_0,x,\theta)=\{\{0\}\}$, $ {\cal A}(y_0,x,\theta_0)=\{\{0\},\{1\}\}$, and Theorem \ref{thmPPD1} yields the following two inequalities:
\begin{equation}
\label{eqnPPD9}
\begin{aligned}
P(Y_1=0,Y_2=1|y_0,x)&\leq P(Y_1=1|y_0,x)\\
&\text{and}\\
P(Y_1=1,Y_2=0|y_0,x)&\leq P(Y_1=0|y_0,x).
\end{aligned}
\end{equation}
\end{enumerate}
Inequalities \ref{eqnPPD6} -- \ref{eqnPPD9} represent all of the inequality restrictions on the CCPs that are implied by Theorem \ref{thmPPD1} when $D=2$. \par

To see that these are the same restrictions as those derived in Theorem 2 of \cite{KPT}, when $T=2$, note first that model \ref{eqnPPD3} is equivalent to
\begin{equation}
\label{eqnPPD10}
Y_t=\mathbbm{1}_{\{\tilde{X}_{ti}'\beta+\tilde{\gamma} Y_{t-1,i}+\tilde{\lambda}_{i}+\tilde{\epsilon}_{ti} >0\}}
\end{equation}
where $\tilde{X}_{ti}:=X_{1ti}-X_{0ti}$, $\tilde{\lambda}_{i}:=\lambda_{1i}-\lambda_{0i}$, and $\tilde{\epsilon}_{ti}:=\epsilon_{1ti}-\epsilon_{0ti}-\gamma$. It can easily be shown that $\tilde{\epsilon}_{ti}$ satisfies the conditional stationarity assumption given $y_0$, $\tilde{x}$ and $\tilde{\lambda}$, since the shocks $\epsilon_{ti}$ are assumed to satisfy Assumption  \ref{assumptionCS}. Moreover, given $y_0$, $x$ and $\theta$, the quantities $a_1$ and $a_2$, defined in \ref{eqnPPD0}, coincide respectively with $\tilde{x}_{1}'\beta$ and $\tilde{x}_{2}'\beta$. When choice is modeled as in \ref{eqnPPD10} and $T=2$, it is shown in \cite{KPT} that the inequality restrictions that characterize the sharp identified set are given by (rewritten here in a way that facilitates comparison):
\begin{enumerate}
\item If $a_2-a_1+\min\{0,\tilde{\gamma}\}\geq \tilde{\gamma} y_0$ then 
\begin{equation}
\label{eqnPPD11}
P(Y_2=0|y_0,x)\leq P(Y_1=0| y_0,x).
\end{equation}

\item If $a_2-a_1+\max\{0,\tilde{\gamma}\}\leq \tilde{\gamma} y_0$ then 
\begin{equation}
\label{eqnPPD12}
P(Y_2=1|y_0,x)\leq P(Y_1=1| y_0,x).
\end{equation}

\item If $a_2-a_1+\tilde{\gamma}(1-y_0)\geq 0$ then
\begin{equation}
\label{eqnPPD13}
P(Y_1=1,Y_2=0|y_0,x)\leq P(Y_1=0|y_0,x).
\end{equation}

\item If $a_2-a_1\geq \tilde{\gamma} y_0$ then
\begin{equation}
\label{eqnPPD14}
P(Y_1=0,Y_2=0|y_0,x)\leq P(Y_1=0|y_0,x).
\end{equation}

\item If $a_2-a_1+\tilde{\gamma}(1-y_0)\leq 0$ then
\begin{equation}
\label{eqnPPD15}
P(Y_1=1,Y_2=1|y_0,x)\leq P(Y_1=1|y_0,x).
\end{equation}

\item If $a_2-a_1\leq \tilde{\gamma} y_0$ then
\begin{equation}
\label{eqnPPD16}
P(Y_1=0,Y_2=1|y_0,x)\leq P(Y_1=1|y_0,x).
\end{equation}
\end{enumerate}

When inequalities \ref{eqnPPD4} and \ref{eqnPPD5} hold, Theorem \ref{thmPPD1} and Theorem 2 of \cite{KPT} have equivalent implications. Specifically, Theorem \ref{thmPPD1} yields inequality \ref{eqnPPD6}, while Theorem 2 of \cite{KPT} yields inequalities \ref{eqnPPD11}, \ref{eqnPPD13} and \ref{eqnPPD14}, and the latter two are easily seen to be redundant given \ref{eqnPPD11}. When inequality \ref{eqnPPD4} is true and inequality \ref{eqnPPD5} is false, Theorem \ref{thmPPD1} provides restrictions \ref{eqnPPD8}, while \cite{KPT} provides restrictions \ref{eqnPPD14} and \ref{eqnPPD15}, and both sets of restrictions coincide. When \ref{eqnPPD4} is false and \ref{eqnPPD5} is true, Theorem \ref{thmPPD1} yields restrictions \ref{eqnPPD9}, while \cite{KPT} yields restrictions \ref{eqnPPD13} and \ref{eqnPPD16}, and both sets of restrictions coincide. Finally, when \ref{eqnPPD4} and \ref{eqnPPD5} are false, Theorem \ref{thmPPD1} yields restriction \ref{eqnPPD7}, while \cite{KPT} yields restrictions \ref{eqnPPD12}, \ref{eqnPPD15}, and \ref{eqnPPD16}, and both sets of restrictions are equivalent as inequalities \ref{eqnPPD15} and \ref{eqnPPD16} are redundant given \ref{eqnPPD12}. Therefore, the inequality restrictions provided by Theorem \ref{thmPPD1} encompass those of Theorem 2 in \cite{KPT} when $D=T=2$.

\subsection{Proofs (and generalization) of Inequalities \ref{eqnAR26}--\ref{eqnAR29}}
\label{ineqDetails}
The following proposition gives the general form of inequalities \ref{eqnAR27}--\ref{eqnAR29}. The proof of inequality \ref{eqnAR26} is discussed further below.\par
 Consider model \ref{eqnAR21}. Given $X=x$, for $t\in[T]$, let $v_{t}:=x_t'\beta$. Let $\Delta_{1,2}(d_1)$, for $d_1\in\{0,1\}$, be defined by
\begin{equation}
\Delta_{1,2}(d_1|y_0,y_{-1},x,\theta):=(v_2+\gamma_{1} d_1+\gamma_2 y_{0})-(v_1+\gamma_1y_0+\gamma_2 y_{-1}),
\end{equation}
and for $t\geq 3$, and $d_{t-1},d_{t-2}\in\{0,1\}$, let $\Delta_{1,t}(d_{t-1},d_{t-2})$ be defined by
\begin{equation}
\Delta_{1,t}(d_{t-1},d_{t-2}|y_0,y_{-1},x,\theta):=(v_t+\gamma_{1} d_{t-1}+\gamma_2 d_{t-2})-(v_1+\gamma_1y_0+\gamma_2 y_{-1}).
\end{equation}
For $2 <t \leq T$ and $d_{t-1},d_{t-2}\in \{0,1\}$, let $\Delta_{2,t}^+(d_{t-1},d_{t-2})$ and $\Delta_{2,t}^-(d_{t-1},d_{t-2})$be defined by
\begin{equation}
\Delta_{2,t}^+(d_{t-1},d_{t-2}|y_0,y_{-1},x,\theta):=(v_t+\gamma_{1} d_{t-1}+\gamma_2 d_{t-2})-(v_2+\max\{\gamma_1,0\}+\gamma_2 y_{0}),
\end{equation}
and 
\begin{equation}
\Delta_{2,t}^-(d_{t-1},d_{t-2}|y_0,y_{-1},x,\theta):=(v_t+\gamma_{1} d_{t-1}+\gamma_2 d_{t-2})-(v_2+\min\{\gamma_1,0\}+\gamma_2 y_{0}).
\end{equation}
For $3\leq s <t \leq T$ and $d_{t-1},d_{t-2}\in\{0,1\}$,  let $\Delta_{s,t}^+(d_{t-1},d_{t-2})$ and $\Delta_{s,t}^-(d_{t-1},d_{t-2})$ be defined by
\begin{equation}
\Delta_{s,t}^+(d_{t-1},d_{t-2}|y_0,y_{-1},x,\theta):=(v_t+\gamma_{1} d_{t-1}+\gamma_2 d_{t-2})-(v_s+\max\{\gamma_1,0\}+\max\{\gamma_2,0\}).
\end{equation}
and
\begin{equation}
\Delta_{s,t}^-(d_{t-1},d_{t-2}|y_0,y_{-1},x,\theta):=(v_t+\gamma_{1} d_{t-1}+\gamma_2 d_{t-2})-(v_s+\min\{\gamma_1,0\}+\min\{\gamma_2,0\}).
\end{equation}

\begin{proposition}
\label{propAR2P1}
Consider model \ref{eqnAR21}, with $T\geq 2$, and suppose that Assumption \ref{eqnAR22} holds. Let $\theta=(\beta,\gamma_1,\gamma_2)$ denote the true parameter value. Then for all $(y_0,y_{-1},x)\in supp(Y_0,Y_{-1},X)$, the following conditional moment inequalities hold (all probabilities are computed conditional on $Y_0=y_0$, $Y_{-1}=y_{-1}$ and $X=x$):\par
 \begin{equation}
\label{eqnAR2P1}
P\left(\bigcup \bigl\{Y_1=d_1,Y_2=0 \ | \ d_1\in \{0,1\}, \ s.t. \ \Delta_{1,2}(d_1)\geq 0 \bigl\}\right)\leq P(Y_1=0),
\end{equation}
\begin{equation}
\label{eqnAR2P2}
P\left(\bigcup \bigl\{Y_1=d_1,Y_2=1 \ | \ d_1\in \{0,1\}, \ s.t. \ \Delta_{1,2}(d_1)\leq 0 \bigl\}\right)\leq P(Y_1=1).
\end{equation}
For $t\geq 3\leq T$, we have
\begin{equation}
\label{eqnAR2P3}
P\left(\bigcup \bigl\{Y_{t-2}=d_{2},Y_{t-1}=d_1,Y_t=0 \ | \ d_1,d_2\in \{0,1\}, \ s.t. \ \Delta_{1,t}(d_{1},d_{2})\geq 0 \bigl\}\right)\leq P(Y_1=0),
\end{equation}
\begin{equation}
\label{eqnAR2P4}
P\left(\bigcup \bigl\{Y_{t-2}=d_{2},Y_{t-1}=d_1,Y_t=1 \ | \ d_1,d_2\in \{0,1\}, \ s.t. \ \Delta_{1,t}(d_{1},d_{2})\leq 0 \bigl\}\right)\leq P(Y_1=1),
\end{equation}
 \begin{equation}
\label{eqnAR2P5}
P\left(\bigcup \bigl\{Y_{t-2}=d_{2},Y_{t-1}=d_1,Y_t=0 \ | \ d_1,d_2\in \{0,1\}, \ s.t. \ \Delta_{2,t}^+(d_{1},d_{2})\geq 0 \bigl\}\right)\leq P(Y_2=0),
\end{equation}
 \begin{equation}
\label{eqnAR2P6}
P\left(\bigcup \bigl\{Y_{t-2}=d_{2},Y_{t-1}=d_1,Y_t=1 \ | \ d_1,d_2\in \{0,1\}, \ s.t. \ \Delta_{2,t}^-(d_{1},d_{2})\leq 0 \bigl\}\right)\leq P(Y_2=1).
\end{equation}
For $3\leq s <t\leq T$, we have
 \begin{equation}
\label{eqnAR2P7}
P\left(\bigcup \bigl\{Y_{t-2}=d_{2},Y_{t-1}=d_1,Y_t=0 \ | \ d_1,d_2\in \{0,1\}, \ s.t. \ \Delta_{s,t}^+(d_{1},d_{2})\geq 0 \bigl\}\right)\leq P(Y_s=0),
\end{equation}
 \begin{equation}
\label{eqnAR2P8}
P\left(\bigcup \bigl\{Y_{t-2}=d_{2},Y_{t-1}=d_1,Y_t=1 \ | \ d_1,d_2\in \{0,1\}, \ s.t. \ \Delta_{s,t}^-(d_{1},d_{2})\leq 0 \bigl\}\right)\leq P(Y_s=1).
\end{equation}
\end{proposition}

\begin{proof}
As in Remark \ref{rem2}, we can assume without loss of generality that the fixed effects are identically equal to zero. To see why Inequality \ref{eqnAR2P1} holds, for $d_1\in\{0,1\}$, let the event $E_{d_1,0}^{1,2}$ be defined by 
\begin{align*}
E_{d_1,0}^{1,2}&:=\{(\epsilon_1,\epsilon_2) \in\mathbb{R}^2 \ | \  Y_1=d_1,Y_2=0 \}\\
&=\{(\epsilon_1,\epsilon_2)\in\mathbb{R}^2\ | \ (2d_1-1)[v_1+\gamma_1y_0+\gamma_{2} y_{-1}-\epsilon_1]>0, \ v_2+\gamma_1 d_1+\gamma_2y_0-\epsilon_2<0\}.
\end{align*}
When  $\Delta_{1,2}(d_1)\geq 0$ we have $\{\epsilon_2 \in\mathbb{R} \ | \ v_2+\gamma_1 d_1+\gamma_2y_0-\epsilon_2<0\}\subseteq \{\epsilon_2\in\mathbb{R}\ | \ v_1+\gamma_1y_0+\gamma_{2} y_{-1} -\epsilon_2<0\}$. Inequality \ref{eqnAR2P1} then follows since the events $E_{0,0}^{1,2}$ and $E_{1,0}^{1,2}$ are disjoint, and the stationarity restriction (Assumption \ref{eqnAR22}) yields
\[
P(\{\epsilon_2\in\mathbb{R}\ | \ v_1+\gamma_1y_0+\gamma_{2} y_{-1} -\epsilon_2<0\})=P(\{\epsilon_2\in\mathbb{R}\ | \ v_1+\gamma_1y_0+\gamma_{2} y_{-1} -\epsilon_1<0\})=P(Y_1=0).
\]
The proof of Inequality \ref{eqnAR2P2} follows by a similar argument and is omitted.\par

I now prove Inequality \ref{eqnAR2P3}. For $d_1,d_2,d\in \{0,1\}$, let the event $E_{d_2,d_1,d}^{t-2,t-1,t}$ be defined by
\begin{equation}
\label{eqnAR2P9}
E_{d_2,d_1,d}^{t-2,t-1,t}:=\{(\epsilon_{t-2},\epsilon_{t-1},\epsilon_t) \in\mathbb{R}^3 \ | \  Y_{t-2}=d_2,Y_{t-1}=d_1,Y_t=d \}
\end{equation}
We have $E_{d_2,d_1,0}^{t-2,t-1,t}\subseteq \{\epsilon_t \in\mathbb{R}\ | \ v_t+\gamma_1d_1+\gamma_2 d_2<\epsilon_t\}$, and the inequality $\Delta_{1,t}(d_{1},d_{2})\geq 0$ implies that $\{\epsilon_t \in\mathbb{R}\ | \ v_t+\gamma_1d_1+\gamma_2 d_2<\epsilon_t\}\subseteq \{\epsilon_t \in\mathbb{R}\ | \ v_1+\gamma_1y_0+\gamma_2 y_{-1}<\epsilon_t\}$. As above, Inequality \ref{eqnAR2P3} then follows since the events $E_{d_2,d_1,0}^{t-2,t-1,t}$ are pairwise disjoint for $d_1,d_2\in \{0,1\}$ and the stationarity restriction yields
\[
P(\{\epsilon_t \in\mathbb{R}\ | \ v_1+\gamma_1y_0+\gamma_2 y_{-1}<\epsilon_t\})=P(\{\epsilon_1 \in\mathbb{R}\ | \ v_1+\gamma_1y_0+\gamma_2 y_{-1}<\epsilon_1\})=P(Y_1=0).
\]
Inequality \ref{eqnAR2P4} follows by a similar argument and its proof is omitted.\par

To establish Inequality \ref{eqnAR2P5}, for the events $E_{d_2,d_1,0}^{t-2,t-1,t}$ defined as in \ref{eqnAR2P9}, note that when $\Delta_{2,t}^+(d_{1},d_{2})\geq 0$, we have 
\[
E_{d_2,d_1,0}^{t-2,t-1,t}\subseteq \{\epsilon_t \in\mathbb{R}\ | \ v_t+\gamma_1d_1+\gamma_2 d_2<\epsilon_t\}\subseteq \{\epsilon_t \in\mathbb{R}\ | \ v_2+\max\{\gamma_1,0\}+\gamma_2 y_0<\epsilon_t\}.
\]
It then follows from the fact that the events $E_{d_2,d_1,0}^{t-2,t-1,t}$ are pairwise disjoint, and by the stationarity restriction, that the term on the left of Inequality \ref{eqnAR2P5} is upper bounded by
\[
P(\{\epsilon_2 \in\mathbb{R}\ | \ v_2+\max\{\gamma_1,0\}+\gamma_2 y_0<\epsilon_2\})\leq P(\{\epsilon_2 \in\mathbb{R}\ | \ v_2+\gamma_1 Y_1+\gamma_2 y_0<\epsilon_2\})=P(Y_2=0).
\]
The proof of Inequality \ref{eqnAR2P6} proceeds by a similar argument and is omitted.\par
I now prove Inequality \ref{eqnAR2P8}. For $d_1,d_2\in \{0,1\}$ and $\Delta_{s,t}^-(d_{1},d_{2})\leq 0$, we have
\begin{align*}
E_{d_2,d_1,1}^{t-2,t-1,t}&\subseteq \{ \epsilon_t \in\mathbb{R}\ | \ v_t+\gamma_1d_1+\gamma_2 d_2>\epsilon_t \}\\
&\subseteq \{ \epsilon_t \in\mathbb{R}\ | \ v_s+\min\{\gamma_1,0\}+\min\{\gamma_2,0\}>\epsilon_t \}.
\end{align*}
Since the events $E_{d_2,d_1,1}^{t-2,t-1,t}$ (for $d_1,d_2\in \{0,1\}$) are pairwise disjoint, the left hand side of Inequality \ref{eqnAR2P8} is upper bounded by $P(\{ \epsilon_t \in\mathbb{R}\ | \ v_s+\min\{\gamma_1,0\}+\min\{\gamma_2,0\}>\epsilon_t \} )$ which, by the stationarity restriction, is equal to $P(\{ \epsilon_s \in\mathbb{R}\ | \ v_s+\min\{\gamma_1,0\}+\min\{\gamma_2,0\}>\epsilon_s \} )$. Inequality \ref{eqnAR2P8} then follows since
\[
P(\{ \epsilon_s \in\mathbb{R}\ | \ v_s+\min\{\gamma_1,0\}+\min\{\gamma_2,0\}>\epsilon_s \} )\leq P(\{ \epsilon_s \in\mathbb{R}\ | \ v_s+\gamma_1 Y_{s-1}+\gamma_2 Y_{s-2}>\epsilon_s \} )=P(Y_s=1).
\]
The proof of Inequality \ref{eqnAR2P7} follows by a similar argument, and is omitted.
\end{proof}

\textbf{\underline{Analytic verification of Inequalities \ref{eqnAR27}--\ref{eqnAR29}:}} Recall that in the setting of Example \ref{exampDyn2}, we have $v_1=0$, $v_2=4$, $v_3=2$, $\gamma_1=3$, $\gamma_2=-4$, $y_0=y_{-1}=1$.\par
Inequality \ref{eqnAR27} is an instance of Inequality \ref{eqnAR2P5}. Indeed, since $v_2+\max\{0,\gamma_1\}+\gamma_2y_{0}=v_2+\gamma_1+\gamma_2=3$, we have $\Delta_{2,t}^+(1,0)=v_3+\gamma_1-(v_2+\max\{0,\gamma_1\}+\gamma_2y_{0})=5-3>0$. Similar computations show that $\Delta_{2,t}^+(0,1)=-4<0$, $\Delta_{2,t}^+(0,0)=-1<0$ and $\Delta_{2,t}^+(1,1)=-2<0$. \par
Inequality \ref{eqnAR28} is an instance of Inequality \ref{eqnAR2P1}. Indeed, both $\Delta_{1,2}(0)$ and $\Delta_{1,2}(1)$ are positive ($\Delta_{1,2}(0)=(v_2+\gamma_2)-(v_1+\gamma_1+\gamma_2)=1$ and $\Delta_{1,2}(1)=4$).\par
Finally, Inequality \ref{eqnAR29} is an instance of Inequality \ref{eqnAR2P3}. Indeed, $\Delta_{1,3}(0,0)=v_3-(v_1+\gamma_1+\gamma_2)=3$, $\Delta_{1,3}(0,1)=-1$, $\Delta_{1,3}(1,0)=6$ and $\Delta_{1,3}(1,1)=2$.\par

In the following proposition, I now turn to the proof of Inequality \ref{eqnAR26}.

\begin{proposition}
\label{propAR2P2}
Consider model \ref{eqnAR21}, with $T=3$, and suppose that Assumption \ref{eqnAR22} holds. Let $\theta=(\beta,\gamma_1,\gamma_2)$ denote the true parameter value, and suppose that $(y_0,y_{-1},x)\in supp(Y_0,Y_{-1},X)$ is such that 
\begin{equation}
\label{eqnAR2P10}
v_3+\gamma_2\leq v_1+\gamma_1y_0+\gamma_2y_{-1}\leq v_2+\gamma_2y_0\leq v_3+\gamma_1+\gamma_2\leq v_3 \leq v_2+\gamma_1+\gamma_2 y_0 \leq v_3+\gamma_1.
\end{equation}
where, for $t\in [3]$, $v_t:=x_t'\beta$. Then 
\begin{equation}
\label{eqnAR2P11}
\begin{aligned}
P(Y_2=1,Y_3=0| x,y_0,y_{-1})&\leq P(Y_1=0,Y_2=0,Y_3=1| x,y_0,y_{-1})\\
&+P(Y_1=1| x,y_{0},y_{-1}).
\end{aligned}
\end{equation}
\end{proposition}
\begin{proof}
Below, all probabilities are computed conditional on $\{X=x,Y_0=y_0,Y_{-1}=y_{-1}\}$ and, as in Remark \ref{rem2}, it is without loss of generality to assume that the fixed effects are identically equal to zero. We have
\begin{align*}
P(Y_2=1,Y_3=0)&=P(Y_1=1,Y_2=1,Y_3=0)+P(Y_1=0,Y_2=1,Y_3=0)\\
&=P(\epsilon_1<a_1,\epsilon_2<a_2+\gamma_1,\epsilon_3>a_3+\gamma_2)+P(\epsilon_1>a_1,\epsilon_2<a_2,\epsilon_3>a_3)\\
&=\circled{a}+\circled{b},
\end{align*}
where, for notational simplicity, $a_1$, $a_2$ and $a_3$ are defined by
\[
a_1=v_1+\gamma_1y_0+\gamma_2y_{-1}, \quad \quad a_2=v_2+\gamma_2y_0, \quad \quad a_3=v_3+\gamma_1,
\]
and note that Inequality \ref{eqnAR2P10} implies that $a_1\leq a_2\leq a_3$.
The second term yields
\begin{align*}
\circled{b}&=P(\epsilon_1>a_1,\epsilon_2<a_2,\epsilon_3>a_3)=P(\epsilon_1>a_1,\epsilon_2>a_3,\epsilon_3<a_2)\\
&+P(\epsilon_1>a_1,\epsilon_2<a_2,\epsilon_3>a_3)-P(\epsilon_1>a_1,\epsilon_2>a_3,\epsilon_3<a_2)\\
&=P(\epsilon_1>a_1,\epsilon_2>a_3,\epsilon_3<a_2)+\circled{c}.
\end{align*}
The first term of $\circled{b}$ can be expanded as follows

\begin{align*}
P(\epsilon_1>a_1,\epsilon_2>a_3,\epsilon_3<a_2)&=P(\epsilon_1>a_1, \epsilon_2>a_2,\epsilon_3<a_3-\gamma_1)-P(\epsilon_1>a_1,a_2<\epsilon_2<a_3,\epsilon_3<a_2)\\
&-P(\epsilon_1>a_1,\epsilon_2>a_2,a_2<\epsilon_3<a_3-\gamma_1)\\
&=P(Y_1=0,Y_2=0,Y_3=1)-\circled{d}-\circled{e}
\end{align*}

By the stationarity restriction,\footnote{Note that since $P(X_1\in A, \{X_2\in B\} \cup \{X_3\in C\})+P(\{X_1 \in A\}\cup \{X_2 \in B\})+P(\{X_1\in A\}\cup \{X_3\in C\})+P(X_1\in A, X_2\in B, X_3\in C)=2P(X_1\in A)+P(X_2\in B)+P(X_3\in C)$, the quantity $P(X_i\in A, \{X_j\in B\} \cup \{X_k\in C\})+P(\{X_i \in A\}\cup \{X_j \in B\})+P(\{X_i\in A\}\cup \{X_k\in C\})+P(X_i\in A, X_j\in B, X_k\in C)$ does not depend on the choice of $i,j,k\in[3]$, if $X_1$, $X_2$ and $X_3$ have the \textit{same marginal distribution} and are arbitrarily coupled. In the present context, the latter invariance yields $P(\epsilon_1>a_1, \epsilon_2<a_2, \epsilon_3>a_3)+P(\epsilon_1>a_1, \{\epsilon_2<a_2\}\cup \{\epsilon_3>a_3\})+P(\{\epsilon_1>a_1\}\cup \{\epsilon_2<a_2\})+P(\{\epsilon_1>a_1\}\cup \{\epsilon_3>a_3\})=P(\epsilon_1>a_1, \epsilon_2>a_3, \epsilon_3<a_2)+P(\epsilon_1>a_1, \{\epsilon_2>a_3\}\cup \{\epsilon_3<a_2\})+P(\{\epsilon_1>a_1\}\cup \{\epsilon_2>a_3\})+P(\{\epsilon_1>a_1\}\cup \{\epsilon_3<a_2\})$.} $\circled{c}$ can be rewritten as
\begin{align*}
\circled{c}&=P(\{\epsilon_1>a_1\}\cup\{\epsilon_3<a_2\})+P(\{\epsilon_1>a_1\}\cup\{\epsilon_2>a_3\})-P(\{\epsilon_1>a_1\}\cup\{\epsilon_2<a_2\})\\
&-P(\{\epsilon_1>a_1\}\cup\{\epsilon_3>a_3\})+P(\epsilon_1>a_1,\{\epsilon_2>a_3\}\cup\{\epsilon_3<a_2\})-P(\epsilon_1>a_1,\{\epsilon_2<a_2\}\cup\{\epsilon_3>a_3\}).
\end{align*}
The first four terms of the last expression can be rewritten as
\begin{align*}
&P(\{\epsilon_1>a_1\}\cup\{\epsilon_3<a_2\})+P(\{\epsilon_1>a_1\}\cup\{\epsilon_2>a_3\})-P(\{\epsilon_1>a_1\}\cup\{\epsilon_2<a_2\})\\
&-P(\{\epsilon_1>a_1\}\cup\{\epsilon_3>a_3\})=P(\epsilon_1<a_1,\epsilon_3<a_2)+P(\epsilon_1<a_1,\epsilon_2>a_3)-P(\epsilon_1<a_1,\epsilon_2<a_2)\\
&-P(\epsilon_1<a_1,\epsilon_3>a_3)=\circled{f}+\circled{g}-\circled{h}-\circled{i}.
\end{align*}
Also, since $\circled{f}+\circled{g}=P(\epsilon_1<a_1,\epsilon_2<a_3,\epsilon_3<a_2)+P(\epsilon_1<a_1,\epsilon_2>a_3,\epsilon_3>a_2)+2P(\epsilon_1<a_1,\epsilon_2>a_3,\epsilon_3<a_2)$ and the events $\{\epsilon_1<a_1,\epsilon_2<a_3,\epsilon_3<a_2\}$, $\{\epsilon_1<a_1,\epsilon_2>a_3,\epsilon_3>a_2\}$, $\{\epsilon_1<a_1,\epsilon_2>a_3,\epsilon_3<a_2\}$ and $\{\epsilon_1<a_1,\epsilon_2<a_2+\gamma_1,\epsilon_3>a_3+\gamma_2\}$ are pairwise disjoint (this follows since $a_2<a_3+\gamma_2$ and $a_2+\gamma_1<a_3$) and included in the event $\{\epsilon_1<a_1\}$, we get
\[
\circled{a}+\circled{f}+\circled{g}\leq P(Y_1=1)+P(\epsilon_1<a_1,\epsilon_2>a_3,\epsilon_3<a_2)=P(Y=1)+\circled{j}.
\]
Thus to establish Inequality \ref{eqnAR2P11}, it suffices to show that
\begin{equation}
\label{eqnAR2P12}
P(\epsilon_1>a_1,\{\epsilon_2>a_3\}\cup\{\epsilon_3<a_2\})-P(\epsilon_1>a_1,\{\epsilon_2<a_2\}\cup\{\epsilon_3>a_3\})-\circled{d}-\circled{e}-\circled{h}-\circled{i}+\circled{j}\leq 0.
\end{equation}
We have
\begin{align*}
&P(\epsilon_1>a_1,\{\epsilon_2>a_3\}\cup\{\epsilon_3<a_2\})-P(\epsilon_1>a_1,\{\epsilon_2<a_2\}\cup\{\epsilon_3>a_3\})\\
&=P(\epsilon_1>a_1,\epsilon_2>a_3)-P(\epsilon_1>a_1,\epsilon_2<a_2)+P(\epsilon_1>a_1,\epsilon_2<a_3,\epsilon_3<a_2)\\
&-P(\epsilon_1>a_1,\epsilon_2>a_2,\epsilon_3>a_3)=\circled{k}-\circled{l}+\circled{m}-\circled{n},
\end{align*}
and Inequality \ref{eqnAR2P12} is equivalent to 
\[
\circled{j}+\circled{k}+\circled{m}\leq \circled{d}+\circled{e}+\circled{h}+\circled{i}+\circled{l}+\circled{n}.
\]
I first show that $\circled{m}\leq \circled{d}+\circled{h}+\circled{l}$. Indeed, $\circled{h}+\circled{l}=P(\epsilon_1<a_1,\epsilon_2<a_2)+P(\epsilon_1>a_1,\epsilon_2<a_2)=P(\epsilon_2<a_2)$ and $\circled{m}-\circled{d}=P(\epsilon_1>a_1,\epsilon_2<a_3,\epsilon_3<a_2)-P(\epsilon_1>a_1,a_2<\epsilon_2<a_3,\epsilon_3<a_2)=P(\epsilon_1>a_1, \epsilon_2<a_2,\epsilon_3<a_2)$, since $a_3\geq a_2$. Hence
\begin{align*}
\circled{m}- \circled{d}-\circled{h}-\circled{l}&=P(\epsilon_1>a_1, \epsilon_2<a_2,\epsilon_3<a_2)-P(\epsilon_2<a_2)\\
&=-P(\epsilon_2<a_2,\{\epsilon_1<a_1\}\cup \{\epsilon_3>a_2\})\\
&=-P(\epsilon_1<a_1,\epsilon_2<a_2)-P(\epsilon_1>a_1,\epsilon_2<a_2,\epsilon_3>a_2)=-\circled{o}-\circled{p},
\end{align*}
and Inequality \ref{eqnAR2P12} is equivalent to
\begin{equation}
\label{eqnAR2P13}
\circled{j}+\circled{k}\leq \circled{e}+\circled{i}+\circled{n}+\circled{o}+\circled{p}.
\end{equation}
We have 
\begin{equation}
\label{eqnAR2P14}
\circled{j}+\circled{k}=P(\epsilon_1<a_1,\epsilon_2>a_3,\epsilon_3<a_2)+P(\epsilon_1>a_1,\epsilon_2>a_3)\leq P(\epsilon_2>a_3)=P(\epsilon_3>a_3)
\end{equation}
where the last equality follows from the stationarity restriction. Also, $\circled{n}+\circled{p}=P(\epsilon_1>a_1,\epsilon_2>a_2,\epsilon_3>a_3)+P(\epsilon_1>a_1,\epsilon_2<a_2,\epsilon_3>a_2)\geq P(\epsilon_1>a_1,\epsilon_3>a_3)$ (since $a_2\leq a_3$), and the latter implies that
\begin{equation}
\label{eqnAR2P15}
\circled{n}+\circled{p}+\circled{i}\geq P(\epsilon_1>a_1,\epsilon_3>a_3)+P(\epsilon_1<a_1,\epsilon_3>a_3)=P(\epsilon_3>a_3).
\end{equation}
Therefore, combining inequalities \ref{eqnAR2P14} and \ref{eqnAR2P15} yields
\[
\circled{j}+\circled{k}\leq P(\epsilon_3>a_3) \leq \circled{n}+\circled{p}+\circled{i},
\]
and Inequality \ref{eqnAR2P13} follows.
\end{proof}

\subsection{Appendix B (\textbf{Proof of Theorem \ref{thmIntE}})}
In this section I give a proof of Theorem \ref{thmIntE}, starting first with the exchangeable case (Assumption \ref{assumptionE}). I present the proof 
under the assumption of stationarity further below.\par
Let us assume for simplicity that the index function differences for the $D$ alternatives are distinct (a similar argument can be made if there are ties among them), and assume (after relabeling alternatives is necessary) that 
\begin{equation}
\label{eqnPatchIneq}
\Delta v_1<\Delta v_2<\cdots<\Delta v_D.
\end{equation}

From Remark \ref{rem3}, the set F of patches is given by $F:=\{(d,d') \in [D]\times [D] \ | \ d\leq d'\}$. Let $F^c$ denote the complement of $F$ in $[D]\times [D]$. 
\subsubsection{Exchangeable case}
The following claim gives an alternative description of the polytope ${\cal Q}_E$ using its  representation given by Equation \ref{eqnRep}.
\begin{claim}
\label{clB1}
The polytope ${\cal Q}_E$ is the set of all $y\in \mathbb{R}^{D^2}$ (with components denoted by $y_{d,d'}$, with $d,d'\in D$) such that:
\begin{equation}
\label{eqnc11}
y_{c,d} \geq -1 \quad \quad \forall c,d \in [D]
\end{equation}
and 
\begin{equation}
\label{eqnc13}
y_{c,d'}+y_{c',d}\leq 0 \quad \forall c,c',d,d' \in [D], \ s.t. \ c\leq d\  and \  c'\leq d'.
\end{equation}
Moreover, Inequality \ref{eqnc13} implies
\begin{equation}
\label{eqnc12}
y_{c,d}\leq 0 \quad \quad \forall (c,d) \in F.
\end{equation}

\end{claim}
\begin{proof}
From \ref{eqnRep} we have ${\cal Q}_E=\{y | \ P_E A^T y\leq 0, \ ||y||_{\infty}\leq 1\}$, where each row the matrix $P_E$ corresponds to an extreme point of the set of exchangeable distributions on the set of rectangular regions ${\cal R}$. Such distributions are determined (up to scale) by vectors of the type $q^{(f,f')}=\delta_{(f,f')}+\delta_{(f',f)}$ where $f,f'\in [F]$, and $\delta_{(f,f')}\in \mathbb{R}^{|F|^2}$ is the vector with all entries equal to zero, except for an entry of 1 in the position corresponding to the region $R=f\times f'$. The corresponding inequality ${q^{(f,f')}}^T A ^T y \leq 0$ can be written as (for $f=(c,d)$ and $f'=(c',d')$):
\begin{equation}
\label{eqnc1p1}
y_{c,d'}+y_{c',d}\leq 0
\end{equation}
which yields Inequality \ref{eqnc13} (recall from the construction of the matrix $A$ that each row corresponding to a choice sequence $(d,d')$ will have an entry of 1 in each column corresponding to a region $f_1\times f_2$ such that $f_1,f_2 \in [F]$ with $f_1=(d,d_1)$ and $f_2=(d_2,d')$, for some $d_1,d_2\in [D]$). Hence the inequalities in $P_E A^T y\leq 0$ are equivalent to the inequalities in \ref{eqnc13}. That Inequality \ref{eqnc12} follows from Inequality \ref{eqnc13} is immediate. That Inequality \ref{eqnc11} holds for all elements of ${\cal Q}_E$, follows from the fact that all elements of ${\cal Q}_E$ satisfy $\|y\|_{\infty}\leq 1$. It thus remains to show that the inequalities 
\begin{equation}
\label{eqnc1p2}
y_{c,d}\leq 1 \quad \quad \forall c,d\in [D]
\end{equation}
are redundant given inequalities \ref{eqnc11} and \ref{eqnc13}. When $c\leq d$, \ref{eqnc1p2} follows from \ref{eqnc12}. Assume for a contradiction that $y_{c,d}>1$ for some $c>d$. Then Inequality \ref{eqnc13} implies that 
\[
y_{c,d}+y_{1,D}\leq 0
\]
which can only hold if $y_{1,D}<-1$, contradicting the constraint \ref{eqnc11}.
\end{proof}

\begin{lemma}
\label{lemB1}
Let ${\cal Q}$ be a polytope given by 
\[
{\cal Q}=\{y \in \mathbb{R}^n\  | \ Ay\leq 0, \ a\leq y\leq b\}
\]
Where the matrix $A$ ($A\in \mathbb{R}^{m\times n}$) has exactly two entries equal to 1 in each row (with all other entries equal to zero), and the vectors $a$ and $b$ (both in $\mathbb{R}^n$) are integral (i.e, all their entries are integers). Then ${\cal Q}$ is integral (i.e., all of its extreme points are integral)
\end{lemma}
\begin{proof}
Let $y$ be an extreme point of ${\cal Q}$. Then by the \textit{rank lemma} (see Theorem 5.7 in \cite{AS}), there exists sets $I_1\subseteq [m]$ and $I_2, I_3\subseteq [n]$ (with some sets possibly empty), such that $|I_1|+|I_2|+|I_3|=n$, and the point $y$ is the unique solution of the system of equations
\begin{equation}
\label{eqnl1p1}
A_{I_1}y=0 \quad y_{I_2}=a_{I_2} \quad y_{I_3}=b_{I_3}
\end{equation}
where $A_{I}$ represents the $ |I|-by-n$ matrix formed by the rows of $A$ with index in $I$, and $a_{I}$ represents the $|I|$ dimensional vector obtained by restricting the vector $a$ to its components with indices in $I$. I now show that the solution of the system \ref{eqnl1p1} is integral, under the assumptions of the Lemma.\par
Let $G=(V,E)$ be the graph with vertex set $V=[n]$, and edge set $E$ consisting of all $(i,j)$ (with $i,j\in[n]$) such that there exists a row of the matrix $A_{I_1}$ with a 1 in both the $i^{th}$ and $j^{th}$ column. For $i\in[n]$, let $e_i\in \mathbb{R}^n$ be the vector with an entry of 1 in the $i^{th}$ position, with all other entries equal to 0. Since the system \ref{eqnl1p1} has a unique solution, and $|I_1|+|I_2|+|I_3|=n$, the rows of $A_{I_1}$ and the vectors $e_j$, with $j\in I_2 \cup I_3$, must form a linearly independent set. In particular, no edge in the graph $G$ connects any $i\in I_2$ to any $j \in I_3$; if such an edge exists, then there is a row of $A_{I_1}$, say its $k^{th}$ row, with 1 in the $i^{th}$ and $j^{th}$ position, and we necessarily have that the $k^{th}$ row of $A_{I_1}$ and the vectors $e_i$ and $e_j$ are linearly dependent. Let $V_0\subseteq [n]$ be given by $V_0=I_2\cup I_3$. Note that if a vertex $i\in[n]$ is not in $V_0$, $i$ must be incident on an edge in $E$; otherwise, the system \ref{eqnl1p1} puts no restriction on the $i^{th}$ component of $y$, and there cannot be a unique solution. The solution $y$ of the system \ref{eqnl1p1} is such that $y_i$ is an integer for all $i\in V_0$ ($y$ satisfies  $y_{I_2}=a_{I_2}$ and $ y_{I_3}=b_{I_3}$). Let $V_1\subseteq [n]$ denote all vertices of the graph G that are reachable from the vertices in $V_0$; that is $V_1$ denotes the set of all vertices $i\in V$ such that there exists a path (in G) from i to a vertex in $V_0$. Arguing by induction on the distance from $v$ to $V_0$, we must have $y_i\in \mathbb{Z}$ for all $i\in V_1$; for instance, all elements of $i\in V_1$ at a distance at most 1 to $V_0$ are either in $V_0$ (and the corresponding entry of $y$ is an integer), or $i \in V_1\backslash V_0$ and there exists $j\in V_0$ such that $(i,j)\in E$, and the latter implies that $y_i+y_j=0$, hence $y_i \in \mathbb{Z}$.\par
Let $V_3=[n]\backslash V_1$ be the remaining vertices. Then since each element of $[n]\backslash V_0$ is incident on an edge in E (see the argument in the preceding paragraph), there exists $I_4\subseteq I_1$ such that the only restrictions on the elements of $V_3$ are from the equations $A_{I_4}y=0$, and since the solution to \ref{eqnl1p1} is unique, we must have $y_i=0$ for all $i\in V_3$. 
\end{proof}

\begin{remark}
Note that if we replace the constraint $Ay\leq 0$ by the constraint $Ay\leq c$ for some integral vector $c\neq 0$, then the conclusion of Lemma \ref{lemB1} is no longer valid. Take 
\[
A=\begin{pmatrix}
1&1&0\\
1&0&1\\
0&1&1
\end{pmatrix},
\quad c=\begin{bmatrix} 1\\1\\1 \end{bmatrix},
\quad a=\begin{bmatrix} 0\\0\\0 \end{bmatrix},
\quad \text{and}
\quad b=\begin{bmatrix} 1\\1\\1 \end{bmatrix}.
\]
Then $y=(0.5,0.5,0.5)'$ is a non-integral extreme point of ${\cal Q}$ (it is the unique solution to the equation $Ay=c$, and it satisfies the other restrictions $a\leq y \leq b$. See Theorem 5.7 in \cite{AS}). If one tries to prove integrality of ${\cal Q}$ following the reasoning of the proof of Lemma \ref{lemB1}, then the argument in the last paragraph, concerning the integrality of $y_i$ for $i\in V_3$,  will no longer be valid.\par
The validity of Lemma \ref{lemB1} can be "appreciated" numerically using Proposition \ref{propEG}.
\end{remark}
\begin{proposition}
\label{propB0}
The polytopes ${\cal Q}_E$ are integral (i.e., their extreme points are vectors with entries in $\{0,\pm 1\}$).
\end{proposition}
\begin{proof}
The proof follows directly from Claim \ref{clB1} and \ref{lemB1}. Indeed the Inequalities \ref{eqnc13}, for all $c,c',d,d'\in[D]$ s.t. $c\neq c'$ or $d\neq d'$, can written in matrix form as $Cy\leq 0$, for some matrix $C$ that has each row containing exactly two entries equal to 1, and the remaining inequalities in the Claim \ref{clB1} can be written in the form $a\leq y \leq b$, where the vector $a$ has all entries equal $-1$, and $b$ has all entries equal to $0$ or $1$ (recall that from the proof of Claim \ref{clB1}, the inequalities $y_{c,d}\leq 1$, for all $c,d\in [D]$, are redundant given Inequalities \ref{eqnc11} and \ref{eqnc13}). 
\end{proof}
\begin{proposition}
All undominated extreme points $y$ of ${\cal Q}_E$ satisfy:
 \begin{equation}
\label{eqnpB1}
y_{c,d}\in \{0,1\} \quad\quad \forall \ (c,d) \in F^c
\end{equation}
 \begin{equation}
\label{eqnpB2}
y_{c,d}\in \{-1,0\} \quad\quad \forall \ (c,d) \in F
\end{equation}
and
 \begin{equation}
\label{eqnpB3}
y_{c,d}=-y_{d,c}  \quad\quad \forall \ c,d \in[D].
\end{equation}
As a consequence of \ref{eqnpB3}, all undominated extreme points of ${\cal Q}_E$ have rank 0.
\end{proposition}
\begin{proof}
Let $\tilde{y}$ denote an undominated extreme point of ${\cal Q}_E$. Then there exists a $w>0$ s.t. $\tilde{y}=argmax\{w^Ty \ | \ y\in {\cal Q}_E\}$ (see Footnote \ref{foot13}). Inequality \ref{eqnpB2} follows from Inequality \ref{eqnc12} and the fact that the extreme points of ${\cal Q}_E$ are integral. \par
To get inequality \ref{eqnpB1}, let $y \in {\cal Q}_E$ be such that $y_{i,j}=-1$, for some $i,j \in [D]$ such that $i>j$, and define $y'$ by: $y'_{a,b}=y_{a,b}$, $\forall (a,b)\neq (i,j)$, and $y'_{i,j}=0$. I show that $y'$ is feasible and has larger objective value, i.e.,  $y' \in {\cal Q}_E$ and $w^Ty'>w^Ty$ (the latter necessarily follows since $w>0$). Indeed, $y'$ clearly satisfies Inequalities \ref{eqnc11} and \ref{eqnc12}, and also satisfies Inequalities \ref{eqnc13} whenever $(c,d')\neq (i,j)$ and $(c',d)\neq (i,j)$ (in all these cases the restriction on $y'$ coincide with the corresponding restriction on $y$). If we now consider \ref{eqnc13} when $(c,d')=(i,j)$, then for all $c',d$ such that $c'<j$ and $i<d$, we have $(c',d)\in F$ (as $c'<d$) and 
\[
y'_{i,j}+y'_{c',d}= 0+ y_{c',d}\leq 0
\] 
since $y_{c',d}\leq 0$ by \ref{eqnc12}. Thus $y'$ is feasible, and we conclude that any undominated extreme point $\tilde{y}$ must satisfy $\tilde{y}_{i,m}\geq 0$ for all $(i,j)\in F^c$, and since $\tilde{y}$ is integral (by Proposition \ref{propB0}), we have \ref{eqnpB1}.\par
I now prove \ref{eqnpB3}. let $y \in {\cal Q}_E$ be such that $y_{i,i}=-1$, for some $i \in [D]$, and define $y'$ by: $y'_{a,b}=y_{a,b}$, $\forall (a,b)\neq (i,i)$, and $y'_{i,i}=0$. I show that $y'$ is feasible and has larger objective value (the latter immediately follows since $w>0$). Indeed, arguing as in the preceding paragraph, $y'$ clearly satisfies all inequalities of the type \ref{eqnc11} and \ref{eqnc13}, and $y'$ also satisfies all inequalities of the type \ref{eqnc12} such that $(c,d')$ and $(c',d)$ are not equal to $(i,i)$ (in all these cases the restrictions on $y'$ coincides with the corresponding restriction on  $y$). If we now consider inequalities of the type \ref{eqnc12} when $(c,d')=(i,i)$, then for all $c'\leq i$ and $d\geq i$, since $(c',d) \in F$, we have
\[
y'_{i,i}+y'_{c,d'}= 0+ y'_{c',d}\leq 0.
\]
We thus conclude (since \ref{eqnpB2} also holds) that $\tilde{y}_{i,i}=0$ and \ref{eqnpB3} holds for all $(c,d)=(i,i)$. \par

I now prove that $\tilde{y}_{i,j}=1$ with $i>j$ implies that $\tilde{y}_{j,i}=-1$. But this follows directly from \ref{eqnc13} (in conjunction with the integrality of $\tilde{y}$ from Proposition \ref{propB0}), as it implies
\[
\tilde{y}_{i,j}+\tilde{y}_{a,b}\leq 0
\]
for all $a\leq j$ and $i\leq b$ (hence $(a,b)\in F$), and the conclusion follows by taking $(a,b)=(j,i)$.\par

It now remains to show that $\tilde{y}_{i,j}=-1$, for some $i<j$ ($i,j\in [D]$), implies that $\tilde{y}_{j,i}=1$. I first show that it must be the case that there exists $c,d\in[D]$ such that $i\leq d <c \leq j$ such that $\tilde{y}_{c,d}=1$. If not, then letting $y_{a,b}=0$ for $(a,b)=(i,j)$ and $y_{a,b}=\tilde{y}_{a,b}$ otherwise, would yield a feasible improvement over $\tilde{y}$, and thus contradict the optimality of $\tilde{y}$. Indeed, $y$ clearly satisfies all inequalities of the type \ref{eqnc11} and \ref{eqnc12}, as well as all inequalities of the type \ref{eqnc13} where $(c,d')\neq (i,j)$ and $(c',d)\neq (i,j)$. For inequalities of the type \ref{eqnc13} where $(c,d')=(i,j)$, if $i\leq c'\leq d\leq j$, then inequality
\[
y_{i,j}+y_{c',d}\leq 0
\] 
follows from \ref{eqnc12} (since $(i,j), (c',d)\in F$). If $i\leq d<c' \leq j$, then
\[
y_{i,j}+y_{c',d}=0+\tilde{y}_{c',d}=0
\] 
since by assumption $\tilde{y}_{c,d}=0$ for all $i\leq d <c \leq j$. Hence $\tilde{y}_{i,j}=-1$ implies that there exists $c,d\in[D]$ such that $i\leq d <c \leq j$ and $\tilde{y}_{c,d}=1$. We then have 
\begin{equation}
\label{eqnpBp1}
\tilde{y}_{a,b}=-1 \quad \forall \ a\leq d \ \text{and } c\leq b,
\end{equation}
which follows from Inequality \ref{eqnc13}, as it implies that
\[
\tilde{y}_{c,d}+\tilde{y}_{a,b}=1+\tilde{y}_{a,b}\leq 0
\]
for all such $a$ and $b$.  Suppose now that $\tilde{y}_{j,i}\neq 1$. Then letting $y_{a,b}=1$ for $(a,b)=(j,i)$ and $y_{a,b}=\tilde{y}_{a,b}$ otherwise, yields a feasible improvement over $\tilde{y}$, thus contradicting its optimality. Indeed, all inequalities of the type \ref{eqnc11} and \ref{eqnc12} are clearly satisfied, as well as all inequalities of the type \ref{eqnc13} that do not involve the $(j,i)^{th}$ coordinate of $y$. For the inequalities of the type \ref{eqnpB3} involving $y_{j,i}$, it suffices to show that
\[
y_{j,i}+y_{a,b}\leq 0
\]
for all $a,b\in [D]$ such that $a\leq i $ and $j\leq b$. But for such $a$ and $b$, we have $a\leq d$ and $c\leq b$, and \ref{eqnpBp1} yields $y_{a,b}=\tilde{y}_{a,b}=-1$, thus $y_{j,i}+y_{a,b}=1-1=0$.

\end{proof}

\subsubsection{Stationary case}
Under Assumption \ref{eqnPatchIneq}, the set ${\cal R}$ of all the regions (see Section \ref{sectionDiscretization}) is given by ${\cal R}=\{f\times g \ | \ f,g \in F\}$, where $F=\{(d,d') \ | \ d,d' \in [D], \ d\leq d'\}$. In the following claim, I first give an alternative description of the polytopes ${\cal Q}_S$.

\begin{claim}
The polytopes ${\cal Q}_S$ is alternatively given by ${\cal Q}_S=\{y | \ A^T y\leq R_S^T z, \ y\geq -\textbf{1}\}$. Equivalently,  ${\cal Q}_S$ is the set of all $y\in \mathbb{R}^{D^2}$ (indexed by tuples $(a,b)\in [D]\times [D]$) for which there exists a $z \in \mathbb{R}^{|F|}$ (indexed by tuples $(a,b)\in [D]\times [D]$, with $a\leq b$) such that:
\begin{equation}
\label{eqnsc11}
y_{a,b} \geq -1 \quad \quad \forall a,b \in [D]
\end{equation}
and 
\begin{equation}
\label{eqnsc12}
y_{a,b}\leq z_{a,c}-z_{d,b} \ \ \forall a,b,c,d \in [D], \ s.t. \ a\leq c \ and \ d\leq b
\end{equation}
\end{claim}

\begin{proof}
Recall (see Section \ref{sectionDDCP}) that the polytopes ${\cal Q}_S=\{y | \ A^T y\leq R_S^T z, \ ||y||_{\infty}\leq 1\}$. The inequalities \ref{eqnsc12} are just an alternative way of writing the restrictions $A^Ty\leq R_S^T z$; recall from the definition of the matrix $R_S$ in Section \ref{sectionDiscretization}, that its column corresponding to the region $f\times g \in {\cal R}$ has an entry of 1 in the row that corresponds to the patch f, and an entry of -1 in the row that corresponds to the patch g, and is equal to zero otherwise. Inequalities \ref{eqnsc11} are implied by the restriction $\|y\|_{\infty}\leq 1$ which is equivalent to $-1\leq y_{a,b}\leq 1$ for all $a,b\in [D]$. It remains to show that the inequalities $y_{a,b}\leq 1$ for all $a,b\in [D]$, are redundant given inequalities \ref{eqnsc11} and \ref{eqnsc12}. Indeed, inequalities \ref{eqnsc12} imply that $\forall a,b \in [D]$ we have
\[
-1\leq y_{1,D}\leq z_{1,b}-z_{a,D},
\]
and
\[
y_{a,b}\leq z_{a,D}-z_{1,b}.
\]
Combining the preceding inequalities yields $y_{a,b}\leq 1$.
\end{proof}
To establish the integrality of the polytope ${\cal Q}_S$, I use Proposition \ref{propEG}; more concretely, I will show below that for all $w\in \mathbb{Z}^{D^2}$, $max\{w^Ty\ | \ y\in {\cal Q}_S\}$ is an integer. 
Note that by the (strong) duality theorem of linear programing, we have
\[
\begin{aligned}
\max\{ w^Ty \ | \ A^Ty\leq R_S^Tz ,\ -y\leq \mathbbm{1}\}&= \min\{ u^T\mathbbm{1}\ | \ u,\lambda \geq 0,\ -u+A\lambda=w, \ R_S\lambda=0 \}\\
&=-w^T\mathbbm{1}+\min\{ \mathbbm{1}^TA\lambda \ | \ \lambda \geq 0,\  A\lambda\geq w, \ R_S\lambda=0 \}\\
&=-w^T\mathbbm{1}+\min\{ \mathbbm{1}^T\lambda \ | \ \lambda \geq 0,\  A\lambda\geq w, \ R_S\lambda=0 \}
\end{aligned}
\]
where the last equality follows from the fact that the matrix $A$ has a single non-zero entry equal to 1 in each column, which implies that $\mathbbm{1}^TA$ is a row vector with all entries equal to 1. As the goal is to show that the value of the latter LP is integral whenever $w$ is integral, it suffices to show that the value of the LP 
\begin{equation}
\label{eqnslp}
\min\{ \mathbbm{1}^T\lambda \ | \ \lambda \geq 0,\  A\lambda\geq w, \ R_S\lambda=0 \}
\end{equation}
is integral whenever $w$ is integral (since $-w^T\mathbbm{1}$ is necessarily integral). I establish the integrality of \ref{eqnslp} by relating it to a network flow problem. In technical language, the integrality proof that I give below consists of showing that the system of inequalities that define the polytope ${\cal Q}_S$ is \textit{total dual integral} (see \cite{EG1}). A good reference for the concepts introduced below relating to network flows is \cite{AMO}.\par
Let $G=(V,\mathbb{A})$ be a directed bipartite graph with vertex set $V$ given by $V=V_1\cup V_2$ where $V_i=\{(f,i) \ | \ f\in F\}$ for $i\in[2]$ (i.e., V is composed of two copies of the set of patches F), and the set of arcs $\mathbb{A}=\mathbb{A}_1\cup \mathbb{A}_2$ where 
\[
\mathbb{A}_1=\{(v,v')\ | \ v\in V_1 \ \text{and } v' \in V_2\}
\]
 and 
 \[
 \mathbb{A}_2=\{(v,v') \ | \ v\in V_2, \ v' \in V_1, \ \text{and } v=(f,2), \ v'=(f,1), \text{for some } f \in F\}. 
 \]
 That is arcs in $\mathbb{A}_1$ connect arbitrary vertices in $V_1$ to arbitrary vertices in $V_2$, and arcs in $\mathbb{A}_2$ connect vertices in $V_2$ to their "twin" vertex in $V_1$. Given an arc $a=(v,v') \in \mathbb{A}$, let $T(a)=v$ denote the tail of the arc a, and $H(a)=v'$ denote the head of the arc a. Given a vertex $v\in V$ and a "flow" $\chi$ ($\chi:\mathbb{A}\rightarrow \mathbb{R}$) define the \textit{inflow} at $v$ by 
 \[
 \delta_{\chi}^-(v)=\sum_{\{a\in \mathbb{A}\ | \ H(a)=v\}}\chi(a),
 \]
and define the \textit{outflow} at $v$ by
\[
\delta_{\chi}^+(v)=\sum_{\{a\in \mathbb{A}\ | \ T(a)=v\}}\chi(a).
\]
A \textit{circulation} is a flow that satisfies the flow conservation constraint at each vertex: i.e.,
\[
\delta_{\chi}^+(v)=\delta_{\chi}^-(v)\quad \forall \  v \in V.
\]
\begin{remark}
\label{rems1}
Note that by construction, for each vertex $v\in V_1$ there is a unique incoming arc at v; i.e., there is a unique $a\in \mathbb{A}$ such that $H(a)=v$. Similarly for each vertex $v\in V_2$ there is a unique outgoing arc at v; i.e., there is a unique arc $a\in \mathbb{A}$ such that T(a)=v. Indeed, if $v\in V_2$ is given by $v=(f,2)$ for some $f\in F$, then the unique outgoing arc at v is given by $a=(v,v')$, where $v'=(f,1)$ is the "twin" vertex to $v$ in $V_1$, and the arc a is also the unique incoming arc at $v'$. As a consequence for all flows $\chi$ on $G$, if $v=(f,1)$ and $v'=(f,2)$ (for some $f\in F$), we have $\delta_{\chi}^+(v')=\delta_{\chi}^-(v)$; moreover, if $\chi$ is a circulation, we have 
\begin{equation}
\label{eqnsfc}
\delta_{\chi}^+(v)=\delta_{\chi}^-(v').
\end{equation}
\end{remark}
 Given $f=(d,d')\in F$, let the functions $\Pi_1$ and $\Pi_2$, from $F$ to $[D]$, be defined by $\Pi_1(f)=d$ and $\Pi_2(f)=d'$.\par

Lemma \ref{lems1} below shows that the value of the LP \ref{eqnslp} coincides with that of a minimum cost flow LP (under some "lower (set) capacity constraints")  on the network given by the graph G. And Lemma \ref{lems2} (further below) will also show that the value of the LP \ref{eqnslp} is equal to that of another "simpler" minimum cost flow LP. These lemma makes it possible to leverage results from the study of network flows to show that the value of the LP \ref{eqnslp} is an integer whenever the vector $w$ is integral. Also, as these lemmas essentially state that the values of some LPs coincide, their validity can be ``assessed" numerically by checking that the values of these LPs are in fact equal for a large number of (randomly generated) instances of the LPs under consideration.

\begin{lemma}
\label{lems1}
Let the directed graph $G=(V,\mathbb{A})$ be defined as above. The value of the LP \ref{eqnslp} coincides with that of the LP
\begin{equation}
\label{eqnslp2}
\min_{\chi \in \mathbb{R}^{|\mathbb{A}|}} \sum _{a\in \mathbb{A}_1} \chi(a)
\end{equation}
subject to the constraints
\begin{equation}
\label{eqnsc1}
\chi(a)\geq 0 \quad \forall a \in \mathbb{A}
\end{equation}
\begin{equation}
\label{eqnsc2}
\delta_{\chi}^+(v)=\delta_{\chi}^-(v)\quad \forall v \in V
\end{equation}
and
\begin{equation}
\label{eqnsc3}
\sum_{a \in I_{d,d'}} \chi(a)\geq w_{d,d'} \quad \forall d,d'\in [D]
\end{equation}
where $I_{d,d'}:=\{a=(v,v')\in \mathbb{A}_1\  | \ v=(f,1), \ v'=(g,2), \text{ and } \Pi_1(f)=d, \ \Pi_2(g)=d' \}$. Moreover, as $d$ and $d'$ vary in $[D]$, the sets $I_{d,d'}$ that appear in the lower (set) capacity constraints \ref{eqnsc3} form a partition of $\mathbb{A}_1$.
\end{lemma}
\begin{proof}
I establish the claim by showing that for each flow $\chi$ that satisfies the constraints \ref{eqnsc1}, \ref{eqnsc2} and \ref{eqnsc3}, there corresponds a vector $\lambda \in \mathbb{R}^{|{\cal R}|}$ such that $\lambda$ satisfies the constraints of the LP \ref{eqnslp}  (and vice versa), i.e.
\begin{equation}
\label{eqnsc4}
\lambda\geq 0
\end{equation}
\begin{equation}
\label{eqnsc5}
R_S \lambda =0
\end{equation}
and
\begin{equation}
\label{eqnsc6}
A\lambda \geq w,
\end{equation}
 and $\sum _{a\in \mathbb{A}_1} \chi(a)=\mathbbm{1}^T\lambda$.\par

\underline{Step 1}. Let $\chi$ be a flow that satisfies constraints \ref{eqnsc1}, \ref{eqnsc2} and \ref{eqnsc3}. For each region $f\times g \in {\cal R}$, associate to it the arc $a=(v,v')\in \mathbb{A}_1$ where $v=(f,1)$, $v'=(g,2)$, and let the corresponding $f\times g^{th}$ entry of the vector $\lambda \in \mathbb{R}^{|{\cal R}|}$ (associated to the flow $\chi$) be defined by $\lambda_{f\times g}=\chi(a)$. We clearly have that $\sum _{a\in \mathbb{A}_1} \chi(a)=\mathbbm{1}^T\lambda$, and that $\lambda$ satisfies constraint \ref{eqnsc4}. I claim that $\lambda$ also satisfies constraints \ref{eqnsc5} and \ref{eqnsc6}. Recall that constraint \ref{eqnsc5} simply encodes the stationarity restriction, and is equivalent to: $\forall$ $f\in F$, we have
\begin{equation}
\label{eqns01}
\sum_{g\in F}\lambda_{f\times g}=\sum_{g\in F}\lambda_{g\times f}.
\end{equation}
Fix $f\in F$, and let $v,v'\in V$ be defined by $v=(f,1)$ and $v'=(f,2)$. By definition of $\lambda$, we have $\delta_{\chi}^+(v)=\sum_{g\in F}\lambda_{f\times g}$ and $\delta_{\chi}^-(v')=\sum_{g\in F}\lambda_{g\times f}$. Since $\chi$ is a circulation (it satisfies \ref{eqnsc2}), equation \ref{eqnsfc} implies that we also have $\delta_{\chi}^+(v)=\delta_{\chi}^-(v')$, and combining these latter observations implies that $\lambda$ satisfies restriction \ref{eqns01}. To show that $\lambda$ satisfies restriction \ref{eqnsc6}, recall that by construction of the matrix A (see Section \ref{sectionDiscretization}), for each $d,d'\in [D]$, there corresponds a row of the matrix $A$ with entries equal to 1 (and equal to 0 otherwise) in each column associated to a region $f\times g$ such that $\Pi_1(f)=d$ and $\Pi_2(g)=d'$. Hence the restriction \ref{eqnsc6} is equivalent to: For each $d,d'\in[D]$, we have
\begin{equation}
\label{eqns02}
\sum_{\{f\times g \in {\cal R} \ | \ \Pi_1(f)=d, \ \Pi_2(g)=d'\}}\lambda_{f\times g} \geq w_{d,d'}.
\end{equation}
But the latter is equivalent to the lower capacity constraint \ref{eqnsc3} by definition of $\lambda$.\par
\underline{Step 2}. Let now $\lambda\in \mathbb{R}^{|{\cal R}|}$ be a vector that satisfies constraints \ref{eqnsc4}, \ref{eqnsc5} and \ref{eqnsc6}. I now associate to it a flow $\chi$ that satisfies constraints \ref{eqnsc1}, \ref{eqnsc2}, \ref{eqnsc3}, and such that the identity $\sum _{a\in \mathbb{A}_1} \chi(a)=\mathbbm{1}^T\lambda$ holds. Let the associated (to $\lambda$) flow $\chi$ be defined by: For each $a=(v,v')\in \mathbb{A}_1$, with $v=(f,1)$ and $v'=(g,2)$, set $\chi(a):=\lambda_{f\times g}$, and for each $a=(v,v')\in \mathbb{A}_2$, with $v=(f,2)$ and $v'=(f,1)$, set $\chi(a):=\sum_{g\in F}\lambda_{g\times f}$. Clearly, by construction, $\chi$ satisfies \ref{eqnsc1} and the identity $\sum _{a\in \mathbb{A}_1} \chi(a)=\mathbbm{1}^T\lambda$. Also, by the definition of $\chi$ on the arcs in $\mathbb{A}_2$, we have that $\delta_{\chi}^+(v)=\delta_{\chi}^-(v)$ for all $v=(f,2)\in V_2$. To establish that $\chi$ satisfies \ref{eqnsc3}, it thus remains to show that $\delta_{\chi}^+(v)=\delta_{\chi}^-(v)$ for all $v=(f,1)\in V_1$. This follows since \ref{eqnsc5} is equivalent to \ref{eqns01} which implies that $\delta_{\chi}^+(v)=\delta_{\chi}^-(v')$ (with $v=(f,1)$, $v'=(f,2)$ and $f\in F$), and by Remark \ref{rems1} $\delta_{\chi}^+(v')=\delta_{\chi}^-(v)$, and thus $\delta_{\chi}^+(v)=\delta_{\chi}^-(v')=\delta_{\chi}^+(v')=\delta_{\chi}^-(v)$ (where the second inequality from the the fact that $\delta_{\chi}^+(v')=\delta_{\chi}^-(v')$ for all $v'\in V_2$). That $\chi$ satisfies restriction \ref{eqnsc3} follows since restriction \ref{eqnsc6} is equivalent to restriction \ref{eqns02}, and the latter is equivalent to restriction \ref{eqnsc3} by the construction of $\chi$. 

\end{proof}

Consider the auxiliary directed bipartite graph $\tilde{G}=(\tilde{V},\tilde{\mathbb{A}})$, with vertex set $\tilde{V}=\tilde{V}_1\cup \tilde{V}_2$ where (for $i\in[2]$) $\tilde{V}_i=\{(d,i) \ | \ d\in [D] \}$ (that is, each $\tilde{V}_i$ represents a copy of the set $[D]$), and the set of arcs $\tilde{\mathbb{A}}=\tilde{\mathbb{A}}_1\cup \tilde{\mathbb{A}}_2$ is such that 
\[
\tilde{\mathbb{A}}_1\{ (v,v') \ | \  v\in \tilde{V}_1 \text{ and } v' \in \tilde{V}_2\}
\]
and
\[
\tilde{\mathbb{A}}_2\{ (v',v) \ | \   v'=(d',2) \in \tilde{V}_2, \ v=(d,1)\in \tilde{V}_1, \text{ and } d\leq d'\}.
\]
The following lemma shows that the value of the LP \ref{eqnslp2}-\ref{eqnsc3}, which is a minimum cost circulation problem with lower capacity constraints on sets of arcs, coincides with the value of a minimum cost circulation problem on the auxiliary graph $\tilde{G}$ with lower capacity constraints on arcs (instead of sets of arcs). In Proposition \ref{props1} below, I will use the \textit{Integral Circulation Theorem} (see Theorem 12.1 in \cite{LEL}) to conclude that the value of the LP (\ref{eqnslp3}-\ref{eqnsc33}) (and thus that of the LP \ref{eqnslp}) is an integer whenever $w$ is an integer.

\begin{lemma}
\label{lems2}
Let the directed graph $\tilde{G}$ be defined as above. The value of the LP \ref{eqnslp2}-\ref{eqnsc3} coincides with that of the LP
\begin{equation}
\label{eqnslp3}
\min_{\eta \in \mathbb{R}^{|\tilde{\mathbb{A}}|}} \sum _{a\in \mathbb{\tilde{A}}_1} \eta(a)
\end{equation}
subject to the constraints
\begin{equation}
\label{eqnsc11}
\eta(a)\geq 0 \quad \forall a \in \tilde{\mathbb{A}}
\end{equation}
\begin{equation}
\label{eqnsc22}
\delta_{\eta}^+(v)=\delta_{\eta}^-(v)\quad \forall v \in \tilde{V}
\end{equation}
and, $\forall$ $ a=(v,v')\in \tilde{\mathbb{A}}_1$ with $v=(d,1)$ and $v'=(d',2)$
\begin{equation}
\label{eqnsc33}
\eta(a)\geq w_{d,d'}.
\end{equation}

\end{lemma}
Note that the lower capacity constraint \ref{eqnsc33} is imposed on the flow along individual arcs, in contrast to the constraint \ref{eqnsc3} which imposes a lower bound on the aggregate flow on sets of arcs (the sets $I_{d,d'}$).

\begin{proof}
As in Lemma \ref{lems1}, the proof proceeds in two steps.\par
 \underline{Step 1} Let $\chi$ be a feasible flow for LP \ref{eqnslp2}-\ref{eqnsc3}. I construct from $\chi$ a flow $\eta$ on $\tilde{G}$ that satisfies constraints \ref{eqnsc11}-\ref{eqnsc33} such that $\sum _{a\in \mathbb{\tilde{A}}_1} \eta(a)=\sum _{e\in \mathbb{A}_1} \chi(e)$. Indeed, for $a=(v,v')\in\tilde{ \mathbb{A}}_1$, with $v=(d,1)$ and $v'=(d',2)$, set $\eta(a)=\sum_{e\in I_{d,d'}}\chi(e)$ (where $I_{d,d'}$ is as defined in Lemma \ref{lems1}). And for $a=(v',v)\in \tilde{\mathbb{A}}_2$, with $v'=(d',2)$ and $v=(d,1)$, set $\eta(a)=\chi(e)$ where $e\in \mathbb{A}_2$ is given by $e=((f,2),(f,1))$ with $f=(d,d')$. As the sets $I_{d,d'}$ form a partition of $\mathbb{A}_1$, it easily follows that $\sum _{a\in \mathbb{\tilde{A}}_1} \eta(a)=\sum _{e\in \mathbb{A}_1} \chi(e)$. That $\eta$ satisfies the non-negativity constraint \ref{eqnsc11} directly follows from the definition of $\eta$ and the constraint \ref{eqnsc1}. That $\eta$ satisfies the arc capacity constraint \ref{eqnsc33} follows from the definition of $\eta$ on the arcs in $\tilde{\mathbb{A}}_1$. It thus remains to show that $\eta$ satisfies the flow conservation constraint \ref{eqnsc22}. Let $v=(d,1)\in \tilde{V}_1$; we have 
\[
\delta^+_{\eta}(v)=\sum_{\{a=(v,v') \ | \ v'\in \tilde{V}_2\}} \eta(a)=\sum_{d'\in[D]} \sum_{e\in I_{d,d'}}\chi(e)=\sum_{\{f\in F \ | \ \Pi_1(f)=d\}}\delta^+_{\chi}((f,1)).
\]
and
\[
\delta^-_{\eta}(v)=\sum_{\{a=(v',v)  | v'=(d',2)\in \tilde{V}_2, \ d'\geq d \}} \eta(a)=\sum_{\{e\in \mathbb{A}_2| e=((f,2),(f,1)), \ f\in F, \ \Pi_1(f)=d\}} \chi(e)=\sum_{\{f\in F | \Pi_1(f)=d\}}\delta^-_{\chi}((f,1)).
\]
Combining these two expressions with the fact that $\chi$ satisfies constraint \ref{eqnsc2} yields that $\eta$ satisfies the flow conservation constraint at $v$. Suppose now that $v'=(d',2)\in \tilde{V}_2$. We have
\[
\delta^-_{\eta}(v')=\sum_{\{a=(v,v') \ | \ v\in \tilde{V}_1\}} \eta(a)=\sum_{d\in[D]} \sum_{e\in I_{d,d'}}\chi(e)=\sum_{\{f\in F \ | \ \Pi_2(f)=d'\}}\delta^-_{\chi}((f,2))
\]
and
\[
\delta^+_{\eta}(v')=\sum_{\{a=(v',v)  | v=(d,1)\in \tilde{V}_1, \ d\leq d' \}} \eta(a)=\sum_{\{e\in \mathbb{A}_2| e=((f,2),(f,1)), \ f\in F, \ \Pi_2(f)=d'\}} \chi(e)=\sum_{\{f\in F | \Pi_2(f)=d'\}}\delta^+_{\chi}((f,2)).
\]
combining these two expressions with the fact that $\chi$ is a circulation yields that $\eta$ satisfies the flow conservation constraint at $v'$. The foregoing establishes that the value of the LP \ref{eqnslp3}-\ref{eqnsc33} is less than or equal to that of the LP \ref{eqnslp2}-\ref{eqnsc3}. The second step will establish the converse.\par

\underline{Step 2}. Let $\eta$ be a feasible flow for the LP \ref{eqnslp3}-\ref{eqnsc33}. I construct from $\eta$ a flow $\chi$ on the graph $G$, such that $\chi$ satisfies constraints \ref{eqnsc1}-\ref{eqnsc3}, and $\sum _{a\in \mathbb{\tilde{A}}_1} \eta(a)=\sum _{a\in \mathbb{A}_1} \chi(a)$. By the \textit{Flow Decomposition Theorem} (see Theorem 3.5 and Property 3.6 in \cite{AMO}, p.80), the flow $\eta$ can be decomposed into 
\[
\eta=\sum_{k\in[K]} \eta^{(k)}
\]
where each $\eta^{(k)}$ is a non-negative circulation supported on a directed cycle ${\cal C}_k$ (i.e., $\eta^{(k)}(a)=0$ for all $a\in \tilde{\mathbb{A}}\backslash {\cal C}_k$), and where $K$ is an integer upper bounded by $|\tilde{\mathbb{A}}|$. Here a directed cycle ${\cal C}$ in $\tilde{G}$ is a sequence of arcs $(v_1,v_2), (v_2,v_3),\cdots,$ $(v_{n-1},v_n),(v_n,v_1)$ such that each for each $i\in[n]$ we have $(v_i,v_{i+1})\in \tilde{\mathbb{A}}$ (with $v_{n+1}:=v_1$), and $v_i\neq v_j$ whenever $i,j\in[n]$ and $i\neq j$. As $\tilde{G}$ is a bipartite graph, each directed cycle ${\cal C}_k$ is necessarily even \footnote{A bipartite graph has no odd cycle.}, i.e., each cycle has an even number of vertices (or equivalently an even number of edges). Thus for each $k$, ${\cal C}_k$ can be represented by a sequence of arcs $(v^{(k)}_1,v^{(k)}_2),(v^{(k)}_2,v^{(k)}_3),\cdots,(v^{(k)}_{2n_k-1},v^{(k)}_{2n_k}),(v^{(k)}_{2n_k},v^{(k)}_1)$ where $2n_k$ is the length of the cycle ${\cal C}_k$ and, for $i\in[n_k]$, each $v^{(k)}_{2i-1}\in \tilde{V}_2$, $v^{(k)}_{2i}\in \tilde{V_1}$, and $(v^{(k)}_{2i-1},v^{(k)}_{2i})\in \tilde{\mathbb{A}}_2$, $(v^{(k)}_{2i},v^{(k)}_{2i+1})\in \tilde{\mathbb{A}}_1$ (with $v^{(k)}_{2n_k+1}:=v^{(k)}_1$). For each $i\in[n_k]$, let $v^{(k)}_{2i-1}=(\bar{d}^{(k)}_i,2)$ and $v^{(k)}_{2i}=(d^{(k)}_i,1)$, for some $\bar{d}^{(k)}_{i}, d^{(k)}_i\in[D]$. As $(v^{(k)}_{2i-1},v^{(k)}_{2i})\in \tilde{\mathbb{A}}_2$, it must be the case that $d^{(k)}_i\leq \bar{d}^{(k)}_i$ for each $i\in[n_k]$. Let $f^{(k)}_i\in F$ be defined by $f^{(k)}_i=(d^{(k)}_i,\bar{d}^{(k)}_i)$. Note that if a circulation $\beta$ is supported on a directed cycle ${\cal C}$, then the flow conservation at each node implies that the flow must be constant along arcs in ${\cal C}$, i.e., $\beta(a)=\beta(a')$ for all arcs $a,a'\in {\cal C}$. Let $\alpha_k$, for each $k\in [K]$, denote the common value of the flow $\eta^{(k)}$ on the arcs in the cycle ${\cal C}_k$ (i.e., $\eta^{(k)}(a)=\alpha_k$ for all $a\in {\cal C}_k$).\par
 For each $k\in [K]$, I now construct a  circulation $\chi^{(k)}$ supported on a directed cycle ${\cal C'}_k$ of $G$. I will show further below that the circulation 
\begin{equation}
\label{eqns001}
\chi=\sum_{k\in[K]}\chi^{(k)}
\end{equation}
has the desired properties. For each $i \in [n_k]$, let $u^{(k)}_{2i-1}:=(f^{(k)}_i,2)$ and $u^{(k)}_{2i}:=(f^{(k)}_i,1)$; clearly $u^{(k)}_{2i-1}\in V_2$, $u^{(k)}_{2i}\in V_1$, and the sequence of arcs $(u^{(k)}_1,u^{(k)}_2),(u^{(k)}_2,u^{(k)}_3),$ $\cdots,(u^{(k)}_{2n_k-1},u^{(k)}_{2n_k}),(u^{(k)}_{2n_k},u^{(k)}_1)$ forms a directed cycle in $G$ (can be easily checked), which I denote by ${\cal C'}_k$. Let $\chi^{(k)}$ be the circulation supported on the cycle ${\cal C'}_k$ with common value of the flow on each arc in ${\cal C'}_k$ given by $\alpha_k$: That is, 
\begin{equation}
\label{eqns002}
\chi^{(k)}((u^{(k)}_{i},u^{(k)}_{i+1})):=\alpha_k \quad \forall i\in[2n_k]\quad (\text{with } u^{(k)}_{2n_k+1}=u^{(k)}_1),
\end{equation}
and set $\chi^{(k)}$ equal to zero on all arcs of $G$ that are not in ${\cal C'}_k$. 
Clearly $\chi^{(k)}$ is a non-negative circulation on $G$. Since the sum of circulations is a circulation, it follows that $\chi$ defined in \ref{eqns001} satisfies \ref{eqnsc1} and \ref{eqnsc2}. In order to see that  $\chi$ satisfies the constraint \ref{eqnsc3}, it is important to note that if the circulation $\eta^{(k)}$ traverses the arc $a=((d,1),(d',2))$ (i.e., $a\in {\cal C}_k$), then the circulation $\chi^{(k)}$ will traverse the set $I_{d,d'}$ exactly once (i.e., the set $I_{d,d'}$ will have a single arc in the cycle ${\cal C'}_k$). Indeed, it follows from the construction of the cycle ${\cal C'}_k$ and the facts that the arcs of a directed cycle are distinct, that
\begin{equation}
\label{eqns003}
  a=((d,1),(d',2))\in {\cal C}_k \text{  if and only if  } |{\cal C'}_k \cap I_{d,d'}|=1.
\end{equation}
 It thus follows that, for $d,d'\in[D]$, we have

\[
\begin{aligned}
\sum_{e\in I_{d,d'}}\chi(e)&=\sum_{e\in I_{d,d'}}\sum_{k\in [K]} \chi^{(k)}(e)\\
&=\sum_{k\in[K]}\sum_{e\in I_{d,d'}} \chi^{(k)}(e)\\
&=\sum_{k\in [K]}\sum_{e\in I_{d,d'}} \alpha_k \mathbbm{1}[e\in {\cal C'}_k]\\
&=\sum_{k\in [K]} \alpha_k \mathbbm{1}[((d,1),(d',2))\in {\cal C}_k]\\
&=\sum_{k\in [K]}\eta^{(k)}((d,1),(d',2))\\
&=\eta(((d,1),(d',2)))\geq w_{d,d'}
\end{aligned}
\]
where the fourth equality follows from \ref{eqns003} and the last inequality follows since $\eta$ satisfies \ref{eqnsc33}. Hence $\chi$ satisfies the lower (set) capacity constraints \ref{eqnsc3}. Furthermore, since for each $k\in [K]$ we have  $|{\cal C}_k\cap \tilde{\mathbb{A}}_1|=|{\cal C'}_k\cap \mathbb{A}_1|$ (follows from the way the $u^{(k)}_i$'s are constructed from the $v^{(k)}_i$'s), and the flow $\eta^{(k)}$ (resp. $\chi^{(k)}$) has the common value of $\alpha_k$ along arcs in its support ${\cal C}_k$ (resp. ${\cal C'}_k$), it follows that $\sum _{a\in \mathbb{\tilde{A}}_1} \eta^{(k)}(a)=\sum _{a\in \mathbb{A}_1} \chi^{(k)}(a)$ which yields (after summing over $k\in[K]$) $\sum _{a\in \mathbb{\tilde{A}}_1} \eta(a)=\sum _{a\in \mathbb{A}_1} \chi(a)$. 

\end{proof}

\begin{proposition}
\label{props1}
The polytopes ${\cal Q}_S$ are integral (i.e., all extreme points have entries in $\{0,\pm 1\}$).
\end{proposition}
\begin{proof}
From Proposition \ref{propEG}, it suffices to show that $\max \{w^Ty\ | \ y\in {\cal Q}_S\}\in \mathbb{Z}$\footnote{Since ${\cal Q}_S$ is bounded, the value of the LP $\max \{w^Ty\ | \ y\in {\cal Q}_S\}$ is always finite.} for all integral objective vectors $w\in \mathbb{Z}^{D^2}$. By the strong duality theorem of LP, the latter is equivalent to showing that the value of the LP $\min\{\mathbbm{1}^T\lambda \ | \ \lambda \geq 0, \ A \lambda \geq w, \ R_S \lambda=0\}$ is an integer, for each integral objective vector $w$ (see the discussion preceding equation \ref{eqnslp}). By Lemma \ref{lems1}, the value of the LP \ref{eqnslp} coincides with that of the LP \ref{eqnslp2}-\ref{eqnsc3}. And by Lemma \ref{lems2}, the value of the LP 
\ref{eqnslp2}-\ref{eqnsc3} coincides with the value of the LP \ref{eqnslp3}-\ref{eqnsc33}. The LP \ref{eqnslp3}-\ref{eqnsc33} is a minimum cost circulation problem with integral lower capacity constraints. From the Integral Circulation Theorem (see Theorem 12.1 in \cite{LEL}), the LP \ref{eqnslp3}-\ref{eqnsc33} admits an optimal solution $\eta^*$ that is integral. It thus follows that the value of the LP \ref{eqnslp3}-\ref{eqnsc33} (equal to $\sum _{a\in \mathbb{\tilde{A}}_1} \eta^*(a)$) is an integer, and we conclude that $\max \{w^Ty\ | \ y\in {\cal Q}_S\}$ is an integer whenever $w$ is integral.
\end{proof}


%

\printbibliography

\end{document}

%% file: Runningtime1.tex
%
%
%

\begin{tabular}{|cc|c|c|c|c|c|c|}

\hline
\multicolumn{2}{|c|}{} & $T=2,\ D=2$ & $T=2,\ D=4$ & $T=2,\ D=6$ & $T=2,\ D=8$ & $T=3,\ D=3$ & $T=3,\ D=4$ \tabularnewline
\hline 

\textbf{Benson's Algorithm}& & 0.05 & 0.18 & * & * & * & * \tabularnewline
\textbf{Cutting-plane Algorithm} & & 0.07 & 0.26 & 1.02 & 13.15 & 3.39 & 439.22\tabularnewline
 \hline

\end{tabular}

%% file: Runningtime2.tex
\begin{tabular}{|cc|c|c|c|c|}

\hline
\multicolumn{2}{|c|}{} & $T=2,\ D=2$ & $T=2,\ D=4$ & $T=2,\ D=6$ & $T=2,\ D=8$ \tabularnewline
\hline 

\textbf{Benson's Algorithm}& & 0.02 & 0.38 & 48.47& * \tabularnewline
\textbf{Cutting-plane Algorithm} & & 0.14 & 0.97 & 35.52 & 3182.2\tabularnewline
 \hline
\end{tabular}

%% file: Simulation.tex
%
%
%
%
%
%
%
%
%
%
%

\begin{tabular}{cc|cccccccc|cc}

\hline
\multicolumn{2}{c}{\#of objectives} & $M=1$& $M=2$& $M=3$& $M=4$& $M=5$& $M=6$& $M=7$& $M=8$ &  \multicolumn{2}{c}{Running time (in seconds)} \tabularnewline
\hline 
\multicolumn{2}{c}{$N=50$} & 0 & 0 & 3& 14 & 37& 33& 11 & 2 & \multicolumn{2}{c}{146.57}\\

 \hline

\multicolumn{2}{c}{$N=100$} & 0 & 0 & 0 & 2 & 16 & 35 & 36 & 11 & \multicolumn{2}{c}{322.13}\\

\hline

\multicolumn{2}{c}{$N=500$} & 0 & 0 & 0 & 0 & 0 & 0 & 12 & 88 & \multicolumn{2}{c}{1289.68}\\

\hline

\multicolumn{2}{c}{$N=1000$} & 0 & 0 & 0 & 0  & 0 & 0 & 1 & 99 & \multicolumn{2}{c}{3297.08}\\

\hline
\end{tabular}